\documentclass[11pt]{article}
\usepackage[margin=0.7in]{geometry}

\usepackage[]{graphicx}
\usepackage{amssymb, float, authblk}
\usepackage{enumitem}
\usepackage{hyperref}
\usepackage{amsmath}
\usepackage{amsthm}
\usepackage{multicol}
\usepackage{multirow}
\usepackage{comment}
\usepackage{color}
\usepackage{stmaryrd}
\usepackage{braket}
\usepackage{thmtools} 
\usepackage{thm-restate}
\usepackage{rotating}
\usepackage[labelformat=empty]{subfig}
\usepackage{qcircuit}
\usepackage{xcolor}

\theoremstyle{definition}
\newtheorem{theorem}{Theorem}

\newtheorem{problem}[theorem]{Problem}

\newtheorem{definition}[theorem]{Definition}
\newtheorem{remark}[theorem]{Remark}

\newtheorem{lemma}[theorem]{Lemma}

\newcommand{\FF}{\mathbb{F}}

\newcommand{\MG}{\mathcal{G}}
\newcommand{\MH}{\mathcal{H}}
\newcommand{\MU}{\mathcal{U}}
\newcommand{\ML}{\mathcal{L}}
\newcommand{\MS}{\mathcal{S}}
\newcommand{\MR}{\mathcal{R}}
\newcommand{\mc}[1]{\mathcal{#1}}

\newcommand{\DG}{\mc{D}^{\mc{G}}}

\newcommand{\Danc}{\mc{D}^{\text{anc}}}

\DeclareMathOperator{\init}{Init}
\DeclareMathOperator{\supp}{supp}
\DeclareMathOperator{\rank}{rank}
\DeclareMathOperator{\rn}{null}
\DeclareMathOperator{\im}{im}
\DeclareMathOperator*{\argmin}{arg\,min}

\title{Improved QLDPC Surgery: Logical Measurements and Bridging Codes}

\author[1]{Andrew W.~Cross}
\author[2]{Zhiyang He (Sunny)}
\author[3]{Patrick J.~Rall}  
\author[1]{Theodore J.~Yoder}
\affil[1]{IBM Quantum, IBM T.J. Watson Research Center}
\affil[2]{Massachusetts Institute of Technology}
\affil[3]{IBM Quantum, IBM Research Cambridge}

\begin{document}

\maketitle

\begin{abstract} 
In this paper, we introduce the gauge-fixed QLDPC surgery scheme, an improved logical measurement scheme based on the construction of Cohen et al.~(Sci.~Adv.~8, eabn1717). 
Our scheme leverages expansion properties of the Tanner graph to substantially reduce the space overhead of QLDPC surgery.
In certain cases, we only require $\Theta(w)$ ancilla qubits to fault-tolerantly measure a weight $w$ logical operator.
We provide rigorous analysis for the code distance and fault distance of our schemes, and present a modular decoding algorithm that achieves maximal fault-distance. 
We further introduce a bridge system to facilitate fault-tolerant joint measurements of logical operators. 
Augmented by this bridge construction, our scheme can connect different families of QLDPC codes into one universal computing architecture.

Applying our toolbox, we show how to perform all logical Clifford gates on the [[144,12,12]] bivariate bicycle code. Our scheme adds 103 ancilla qubits into the connectivity graph, and one of the twelve logical qubits is used as an ancilla for gate synthesis. Logical measurements are combined with the automorphism gates studied by Bravyi et al.~(Nature 627, 778-782) to implement 288 Pauli product measurements. We demonstrate the practicality of our scheme through circuit-level noise simulations,
leveraging our proposed modular decoder that combines BPOSD with matching.
\end{abstract}
\tableofcontents

\section{Introduction}
Reliable quantum computing at scale requires encoding quantum information into error-correcting codes and acting on that information without spreading errors too quickly \cite{aharonov2008ft,kitaev1997imperfect,knill1998resilient,aliferis2005threshold,gottesman2009intro}. This necessitates a fault-tolerant architecture that both stores quantum information and also computes with it \cite{Campbell2017Roads}. The surface code provides the foundation for one of the leading fault-tolerant quantum computing architectures \cite{bravyi1998codes,dennis2002memory,fowler2012surface,Litinski_2019}. It has a high accuracy threshold, allowing it to tolerate realistic noise rates, and its low-weight, geometrically local check operators can be measured on hardware that implements a two-dimensional lattice \cite{krinner2022repeated,google2023suppressing,sundaresan2023multiround}.

The technique of lattice surgery is an essential element of the surface code architecture that enables computation at realistic noise rates within the same local connectivity of the code itself \cite{Horsman2012LatticeSurgery}. Lattice surgery provides a way to measure products of logical Pauli operators by measuring many low-weight geometrically local checks. Not all of these checks belong to the code when they are measured, so the code changes dynamically by a process called gauge fixing \cite{vuillot2019deformation}. Sequences of lattice surgery operations implement the subset of logical gates called Clifford gates. The surface code architecture is completed by incorporating magic state injection and distillation to construct a universal set of gates \cite{bravyi2005magic,li2015magic,Litinski_2019,Chamberland_2022}.

However, the surface code architecture comes with significant overhead. Resource estimates for scaled applications indicate that surface code architectures may require many millions of qubits \cite{gidney2021factoring,beverland2022assessing}. One of the reasons for this daunting space overhead is that codes that are constrained to have local checks in two dimensions, such as the surface code, have parameters that are governed by the Bravyi-Poulin-Terhal bound \cite{bravyi2010tradeoffs}. This bound limits how efficiently codes with geometrically local checks encode quantum information. An implication is that, because the surface code is the exemplar of all such geometrically local codes, it is difficult to find better codes under these constraints.

One way forward is to relax the requirement of geometric locality, while retaining the need for qubits and checks to connect to only a few neighbors. Such high-rate low-density parity check (LDPC) codes can have much better parameters, both in terms of their encoding rate and their ability to correct errors \cite{tillich2014hgp,panteleev2021degenerate,breuckmann2021ldpc,breuckmann2021balanced,panteleev2022good,leverrier2022tanner,dinur2023good}. Using these codes, fault-tolerant quantum computing becomes possible in principle with only constant space overhead in the asymptotic limit \cite{gottesman2014constant,fawzi2020constant}.

Recently, several protocols for quantum memories based on LDPC codes have been developed \cite{breuckmann2017hyperbolic,poulsen2017fault,higgott2023hyperbolic,bravyi2024highthreshold,xu2024constant}. Efficient memories could be used in concert with the surface code to reduce space overhead. In some instances, these memory protocols simultaneously have high thresholds and low overheads at practical code block sizes, making them attractive candidates for near-term demonstrations \cite{bravyi2024highthreshold}. A growing body of research aims to find improved codes \cite{scruby2024highthreshold,voss2024trivariate}, decoders \cite{gong2024decoding,wolanski2024ambiguity}, and implementations for these quantum memories \cite{viszlai2023bicycle,poole2024fast,berthusen2024local}.

Beyond storing information in high-rate quantum LDPC codes, computing directly on that information is an active area of research. There is the need for fast, parallel, and space-efficient gates, and also the challenge that logical qubits are encoded together in the same block, making them difficult to address individually. For stabilizer codes, there is a general solution that teleports qubits or gates and measures logical observables using specially prepared ancilla states \cite{gottesman1997thesis,gottesman1999teleport,steane2005cssgates,knill2005detected}. However, these large Steane- or Knill-type ancilla states are expensive to prepare, so new techniques for logical computation on LDPC codes have been developed \cite{breuckmann2017hyperbolic,lavasani2018hyperbolic,tjoc2019gates,krishna2021gates,cohen2022low,huang2023homomorphic,zhu2023gates}. 
One of these approaches generalizes the idea of lattice surgery to LDPC codes \cite{breuckmann2017hyperbolic,lavasani2018hyperbolic,krishna2021gates,cohen2022low}.

In this paper, we consider the problem of computing with LDPC codes using generalized lattice surgery. Our high level goal is to efficiently measure a set of logical Pauli operators of one or more LDPC code blocks. Using Pauli-based computation \cite{Bravyi_2016,Litinski_2019}, logical Pauli measurements can be used to implement Clifford gates and, provided additional magic states \cite{bravyi2005magic}, non-Clifford gates.
There are two simple, well-known ways to measure a weight $d$ logical operator: individually measure each qubit in the support of the logical operator and infer the logical qubit state, or entangle a single ancilla qubit to the operator using $d$ CNOT gates and measure the ancilla. It is easy to see that the first method can be highly disruptive to the remaining logical information, while the second method requires qubit degree $d$ and is not fault-tolerant, since an error on the ancilla can propagate to many other qubits. The problem of designing logical measurement schemes, therefore, is to obtain the measurement result fault-tolerantly while maintaining sparse connectivity and reasonable space-time overhead. 

We measure the quality of a logical measurement scheme in terms of four metrics: space-time overhead, parallelism and applicability, connectivity, and fault-tolerance. Let us now detail these metrics.
\begin{enumerate}
    \item Space-time overhead: A scheme makes use of a system of ancilla qubits to implement a set of measurements. The total number of qubits in this system is an important metric. A circuit acting on the original code and the ancilla system enacts the measurement, so the depth of this measurement circuit is another key metric.

    \item Parallelism and applicability: A scheme should be applicable to various settings, including single-logical qubit measurements, multi-logical qubit joint measurements, and joint measurements across multiple 
    code blocks. A scheme should also support parallel measurements of commuting logical operators.

    \item Connectivity: Extending the code with an ancilla system means adding more connections between qubits. Connections are needed between pairs of qubits that interact. In this paper, the connections correspond to edges in the Tanner graph of a quantum code. We measure this cost in terms of qubit degree, which is the maximum number of neighboring qubits connected to a single physical qubit.

    \item Fault-tolerance: There are several ways to measure the fault-tolerance of a scheme. In this paper, we study the following measures.
    \begin{enumerate}
        \item Code distance: During the measurement procedure, the base LDPC code changes into a merged code. The logical information is protected by a subsystem code with all the logical qubits from the original code (except the one being measured), and some gauge qubits which carry no logical information. Ideally, this subsystem code should have distance $d$, so that the unmeasured logical information remains well-protected.
        
        \item Fault distance: Besides the subsystem code distance, we also consider the phenomenological fault distance of logical measurement schemes, which is the minimum number of measurement or qubit errors needed to produce a logical error or an incorrect logical measurement outcome without triggering any detectors. A logical measurement scheme should have fault distance $d$. We further distinguish between the logical fault distance and the measurement fault distance in our analysis.
   
        \item Circuit distance: The procedure is implemented by a noisy quantum circuit whose gates and other circuit locations can fail. The circuit distance is the minimum number of faults needed to produce a logical error or incorrect logical measurement outcome without triggering any detectors. Ideally, a logical measurement circuit should have circuit distance $d$.
    \end{enumerate}
\end{enumerate}

In an earlier work~\cite{cohen2022low}, Cohen, Kim, Bartlett, and Brown introduced a scheme based on an ancilla system which is a hypergraph product code. {We henceforth refer to their scheme as the CKBB scheme} and explain its mechanism in Section~\ref{sec:CKBBschemes}. For a logical operator of weight $d$, the CKBB ancilla system has $2d-1$ layers, each layer isomorphic to the Tanner graph of the measured operator, totaling to a space overhead of $O(d^2)$.
This ancilla system can be naturally extended to support joint measurements.
The scheme then use $d$ rounds of QEC on the merged code to obtain the logical measurement result.  
If the original code has qubit degree $\delta$ and distance $d$, the merged code will have degree $\delta+1$ and distance $d$. 
While the CKBB scheme has many desirable properties, its space overhead is daunting: for typical LDPC codes where the distance scales as $d \sim \Omega(\sqrt{n})$, the ancilla system needed for measuring one logical operator may be larger than the code block itself. This poses a significant barrier to utilizing the CKBB scheme in practice.

In this paper, we propose \textit{gauge-fixed QLDPC surgery} as an improved version of the CKBB scheme with significantly lower space overhead. 
Our main theoretical contributions can be summarized as follows.
\begin{enumerate}
    \item We gauge-fix the CKBB ancilla system so that the merged code contains no gauge qubits.
The resulting \textit{gauge-fixed ancilla system} ensures that the merged code has distance at least $d$, while using much fewer than $2d-1$ layers if the Tanner graph induced by the logical operator is expanding (Theorem~\ref{thm:Xdistance}). In cases where the graph is sufficiently expanding, the space overhead for measuring a weight $w$ logical operator is just $\Theta(w)$.
\item As a result of reducing the number of layers, the constructions of joint measurement systems from \cite{cohen2022low} are no longer distance preserving. We propose new, distance preserving constructions of $Y$ and joint measurement systems, utilizing a set of \textit{bridge} qubits (Section~\ref{sec:jointmeasurement} and~\ref{sec:ysystem}). 

    \item We show that our gauge-fixed system with bridge can be used as an \textit{code adapter}, which can fault-tolerantly connect two logical operators from two arbitrary code blocks (Section~\ref{sec:bridge}). This paves a path towards large-scale architectures based on multiple families of QLDPC codes.
    
    \item We present a rigorous analysis of fault distance, and prove that $d$ rounds of syndrome extraction ensures fault distance $d$ for all of our proposed measurement schemes (Theorem~\ref{thm:logical_and_measurement_distance},\ref{thm:XX_faultdistance_sameblock},\ref{thm:Y_faultdistance}), including the code adapter (Theorem~\ref{thm:XX_faultdistance_twoblocks}).

    \item Utilizing such an ancilla system increases the decoding complexity of the merged code. 
    We propose a modular decoding algorithm on the $d$-round space-time decoding graph of the merged code (Section~\ref{sec:modular_decoder}) which decomposes the decoding graph into two parts, one supported on the original code, the other supported on the gauge-fixed ancilla system, and decodes them separately. 
    We prove that this modular decoder achieves fault distance $d$. 
\end{enumerate}
Our gauge-fixed surgery scheme comes with a noteworthy caveat: without additional assumptions on the measured operator or the code, we do not have upper bounds on the weight of the gauge-fixed operators. 
In other words, gauge-fixing these operators may break the LDPC property on the merged code. 
We discuss this caveat in Section~\ref{sec:checkweight}, where we show theoretical evidence (Lemma~\ref{lem:redundancy}) suggesting that in the average case, these gauge checks may not break the LDPC property.
Our bridge construction also lacks a worst case guarantee on the LDPC property, which we discuss in Section~\ref{sec:bridge}. 
We remark that after this work first appeared, these caveats are fully resolved in the schemes proposed by~\cite{williamson2024low} and~\cite{swaroop2024universal}. 
See section~\ref{sec:related_works} for a more detailed discussion on related works. 

Our theoretical analysis offers strong support for our practical optimizations and numerical simulations on the $[[144, 12, 12]]$ gross code, a 6-limited bivariate bicycle code \cite{bravyi2024highthreshold}. Our main practical contributions are as follows.
\begin{enumerate}
    \item We demonstrate that a mono-layer ancilla system with 103 additional qubits suffices to fault-tolerantly measure eight distinct logical Pauli operators, which significantly improves on the state of the art as summarized in Table~\ref{tab:summarygross}. We verified numerically with CPLEX that all eight measurement protocols have fault distance $12$.
    
    \item In the connectivity graph of the merged code there are exactly two qubits with degree 8, and the remaining qubits have degree 7 or less. The use of two degree-8 qubits is not a side effect of gauge-checks, but instead a result of using one ancilla system for multiple distinct Pauli measurements. For details, please see Figure~\ref{fig:gross_ancilla_counts}.

    \item We implemented a modular decoder using a combination of BPOSD and matching, and find that it gives comparable logical error rates to a fully BPOSD decoder with significantly faster decoding time (Figure~\ref{fig:matching}).

    \item Using the modular decoder, we numerically study the logical error rates of $\bar{X}$ and $\bar{Z}$ logical measurements under circuit-level noise and find they are within a factor of 10x and 5x larger than that of the base code, respectively (Figure~\ref{fig:numericalexperiments}). The measurement error rate is approximately 10x lower than the logical error rate, and we find that 7 repetitions of the merged code syndrome measurement optimizes the logical and measurement error rates at a physical error rate of 0.001.

    \item Finally, we combine permutation automorphisms of the gross code with logical measurements to implement a set of 288 native Pauli measurements on the 12 logical qubits. These measurements yield a further set of 95 Pauli rotations $\{\exp(i\frac{\pi}{4}P)\}$ that generate the eleven-qubit Clifford group on the gross code block. 
    Furthermore the twelve logical qubits can be individually measured. 
\end{enumerate}

\begin{table}[H]
    \centering
    \begin{tabular}{| c || c | c | c |}
    \hline
    & Ref~\cite{bravyi2024highthreshold} & Ref~\cite{cowtan2024ssip} & This work \\
    \hline 
    Measurement set & $X(f,0)$, & $X(f,0)$, & $\bar{X}:=X(p,q)$, $\bar{Z}:=Z(r,s)$, \\
    & $Z(h^\top ,g^\top )$ & $Z(h^\top ,g^\top )$ & $\bar{X}':=X(ws^\top ,wr^\top )$, $\bar{Z}':=Z(wq^\top ,wp^\top )$, \\
    & & & $\bar{Y}:=i\bar{X}\bar{Z}$, $\bar{Y}':=i\bar{X}'\bar{Z}'$, \\
    & & & $\bar{X}\bar{X}'$, $\bar{Z}\bar{Z}'$ \\
    \hline
    Ancilla system size  
    & 1380 & 180 & 103 \\
    \hline 
    Number of layers & 23 & 5 & 1 \\
    \hline 
    Maximum qubit degree & 7 & 7 & 8 \\
    \hline 
    \end{tabular}
    \caption{Parameters of logical measurement schemes for the $[[144,12,12]]$ gross code. Each scheme allows one measurement from the measurement set to be performed at a time. The Pauli operators in the measurement set are presented using the notation of \cite{bravyi2024highthreshold}. The number of layers is the maximum over all merged codes in the scheme. We note that \cite{cowtan2024ssip} presents different versions with varying amounts of parallelism and size; we choose a version that is most similar to \cite{bravyi2024highthreshold}.}
    \label{tab:summarygross}
\end{table}

This practical progress reshapes our vision of a medium-to-large scale computing architecture using bivariate bicycle codes. We further remark that most of our results on the gross code are better than our worst case theoretical analysis. 
In particular, the fact that a mono-layer system is distance-preserving for all eight measurement protocols is a pleasant surprise. For this reason, to apply the gauge-fixed surgery scheme to other families of QLDPC codes, we suggest using our theoretical analysis as an ``instruction manual'', and expect case-by-case optimized results to be better.

The paper is organized as follows. Section~\ref{sec:prelim} establishes notation and reviews the CKBB measurement scheme. Section~\ref{sec:gauged} presents our gauge-fixed surgery scheme for CSS codes and all aforementioned theoretical analysis. In Section~\ref{sec:gross_code}, we apply our theoretical toolbox to construct an ancilla system for the $[[144,12,12]]$ gross code, numerically investigate its performance in the circuit model, and synthesize logical Clifford gates on the code block.

\subsection{Related work}~\label{sec:related_works}
We remark on a collection of related works. 
The independent work of~\cite{cowtan2024ssip} appeared shortly before ours, in which QLDPC surgery is studied in the language of chain complexes and homology. 
\cite{cowtan2024ssip} further applied the CKBB construction to various small-scale QLDPC codes, and presented numerical results showing that in concrete instances, the space-overhead is often much less than the $O(d^2)$ upper bound, which is reflected here in Table~\ref{tab:summarygross}.

Since this work first appeared, a few others have followed with improvements to the central ideas. First, while we point out here that expansion existing in the Tanner graph of the code can drastically reduce the space-overhead of QLDPC surgery, the techniques of~\cite{williamson2024low} and \cite{ide2024fault} can simply add the required expansion to the ancilla system to achieve the same or better space overhead. In addition,~\cite{williamson2024low,ide2024fault} think of the ancilla system as built on a graph and this perspective has other benefits in that certain cellulation \cite{hastings2021weight,williamson2024low,ide2024fault} and decongestion \cite{freedman2021building,hastings2021weight,williamson2024low} ideas can be applied to guarantee that the gauge-fixed ancilla system remains LDPC. 
Second, while we show that any two ancilla systems can be joined via a bridge to measure a joint logical operator,~\cite{swaroop2024universal} constructs the bridge/adapter in such a way to guarantee the newly introduced checks are still LDPC. This construction relies also on the graphical perspective from~\cite{williamson2024low,ide2024fault}. 
Third, both~\cite{xu2024fast} and \cite{zhang2024time} present approaches to parallelizing related logical operations and appeared around the same time as our work. 
Ref.~\cite{xu2024fast} applied the technique of homomorphic measurements \cite{huang2023homomorphic} to homological product codes, while 
Ref.~\cite{zhang2024time} proposed parallelized QLDPC surgery methods applicable to generic codes, utilizing the CKBB scheme. 
Building on the methods of our work and~\cite{williamson2024low,swaroop2024universal},~\cite{he2025extractors} and~\cite{yoder2025tour} proposed QLDPC architectures capable of performing universal fault-tolerant computation with fixed hardware connectivity, while~\cite{cowtan2025parallel} improved the techniques from~\cite{zhang2024time} to enable more efficient parallel surgery. 
We refer readers to Section~3.2 of~\cite{he2025extractors} for a brief review of recent developments.

 \section{Preliminaries}\label{sec:prelim}

We begin by establishing some notation. We then apply this notation to review the logical measurement scheme presented by \cite{cohen2022low}.

\subsection{Tanner graph notation}\label{subsec:Tanner_graph}

 Consider a CSS code $\mathcal{G}$ with check matrices $H^X \in \mathbb{F}_2^{m_x \times n}$, $H^Z \in \mathbb{F}_2^{m_z \times n}$ and minimum distance $d$.  Let $\mathcal{V}$ denote the set of qubits, with $|\mathcal{V}| = n$, $\mathcal{C}^Z$ denote the set of $Z$-checks with $|\mathcal{C}^Z| = m_z$, and $\mathcal{C}^X$ denote the set of $X$-checks with $|\mathcal{C}^X| = m_x$. The \textit{Tanner graph} is a tripartite graph with vertex sets $\mc{C}^X, \mc{V}, \mc{C}^Z$ such that each vertex in $\mc{C}^X$ and $\mc{C}^Z$ is connected to the qubits they act on in $\mc{V}$.

We describe quantum codes with or without ancilla systems by drawing their Tanner graphs. Instead of drawing the Tanner graphs in their entirety, we combine nodes into several sets of either check nodes or qubit nodes. These sets should be chosen so that we can easily describe the connections between each pair of qubit and check sets via an incidence matrix. To be more specific, if there is a set of qubits $V$ and a set of checks $C$, the edge between them will be labeled by a matrix $F\in\mathbb{F}_2^{|C|\times|V|}$ where $F_{ij}=1$ if and only if check $i$ from $C$ is connected to qubit $j$ from $V$. If these checks are $Z$ checks, then we use $\to_Z$ to summarize the above situation as $F : C \to_Z V$. For example the code $\mathcal{G}$ above has connections $H^X : \mathcal{C}^X \to_X \mathcal{V}$ and $H^Z : \mathcal{C}^Z \to_Z \mathcal{V}$. Since our graphs are sparse, many pairs of qubit and check sets will not be connected (or, formally, connected by the $0$ matrix), in which case we leave out that edge.

This way of drawing the Tanner graph is essentially a description of the code as a chain complex. We can therefore think of sets of checks $C$ and sets of qubits $V$ as vector spaces over $\mathbb{F}_2$ with $F$ the boundary map between them. A subset $c \subset C$ can be seen as a row vector in $\mathbb{F}_2^{|C|}$. Then, if $F: C \to_Z V$, the qubits in $V$ that correspond to the product of the subset of checks $c$ in $C$ are $v= cF \in \mathbb{F}_2^{|V|}$. Due to this vector space interpretation, for the rest of this paper we will often write $c\in C$ to denote the subset indicated by $c$, and similarly for $v\in V$.

In our figures, an edge denoting the connection $F : C \to_Z V$ is labeled by the matrix $F$ as well as the Pauli matrix $Z$, or $X,Y$ for $\to_X,\to_Y$ respectively, to indicate how the checks act on connected qubits. When all edges from a set of check nodes are labeled the same, we simply label the set of check nodes with $X$, $Y$, or $Z$. Checks are always rows and qubits are always columns, although vectors are always row vectors.

To specify the type of Pauli explicitly, we adopt the notation $Z(v\in V)$ (or $X(\cdot)$ or $Y(\cdot)$), to indicate Pauli $Z$s acting on exactly the qubits indicated by $v$. If it is clear from context that $v$ is a member of $V$, then we can shorten this to $Z(v)$. If we want $Z$ to act on all qubits in $V$, we write simply $Z(V)$.

In this paper, we frequently discuss the merged code obtained by attaching an ancilla system to the original code being measured. To differentiate between the two codes, we use $H$ to denote the stabilizer check matrices in the original code and use $\mathcal{H}$ for the merged code. 
For a set of checks $c\in C$, we denote their product by $\mathcal{H}(c\in C)=\mathcal{H}(c)$, and the product of all checks in $C$ by $\mathcal{H}(C)$. 
We also use $\mathcal{H}^Z(\cdot)$ or $\mathcal{H}^X(\cdot)$ if the checks are entirely $Z$- or $X$-type, respectively. 
This is meant to indicate the Pauli operator obtained by multiplying together the checks indicated by vector $c$.

Finally, if two Paulis $P_1,P_2$ differ only by an element of a Pauli group $\mathcal{S}$, i.e.~$P_1=sP_2$ for $s\in\mathcal{S}$, we write $P_1\equiv_{\mathcal{S}}P_2$.

\subsection{The CKBB ancilla system}\label{sec:CKBBschemes}

Consider any CSS code $\mc{G}$. In general, we say a logical $X$ operator $\bar X_M$ is irreducible if there are no other logical $X$ operators whose support is a subset of that of $\bar X$. We focus on logical $X$ operators here, but an ancilla system for $Z$ operators can be obtained by swapping $X$ and $Z$ in the construction. In the measurement scheme introduced by~\cite{cohen2022low}, for any logical operator $\bar X_M$ that is irreducible, they create an extended Tanner graph with $\bar X_M$ in the stabilizer.

Key to their construction and ours is the subgraph consisting of $\bar X_M$ and its neighborhood of $Z$ checks. For an $X$-logical operator $\bar{X}_M$ supported on a set of qubits $V_0\subset \mc{V}$, let $C_0 \subset \mc{C}^Z$ be the set of $Z$-checks that has neighbors in $V_0$. The induced Tanner graph of $\bar{X}_M$ is the restriction of the Tanner graph to the vertices $V_0$ and $C_0$. Formally, since the induced Tanner graph only features $Z$ checks, we can define a matrix $F: C_0 \to_Z V_0$ whose rows are the supports of the $Z$ checks. If $J_{V_0}: \mathcal{V} \to V_0$ and $J_{C_0} : \mathcal{C}^Z \to C_0$ are projection isometries, then $F := J_{C_0}^\top H^Z J_{V_0}$.

 A pair of $V_i, C_i$, connected as in the Tanner graph constitute one `layer' of the construction. Having identified $C_0, V_0$ within the original code $\mathcal{G}$'s vertices $\mathcal{C}^Z, \mathcal{V}$, we append $2d-1$ more such layers. We assign these vertices alternating roles, as follows. 
\begin{align}\label{eqn:vertexroles}
V_i \text{ for } \Bigg\{ \begin{array}{l}\text{odd }i \to \text{ X check} \\ \text{even } i \to \text{ physical qubit} \end{array}\\
C_i \text{ for } \Bigg\{ \begin{array}{l}\text{odd }i \to \text{ physical qubit} \\ \text{even } i \to \text{ Z check} \end{array}
\end{align}
Layers are connected identically via $F : C_i \to_Z V_i$ for even $i$, and $F^\top : V_i \to_X C_i$ for odd $i$. Furthermore, adjacent layers are connected by the identity: $I : C_0 \to_Z  C_1$, $I : V_1 \to_X V_0$, $I : V_1 \to_X V_2$, $I : C_2 \to_Z  C_1$, etc. This construction is visualized in Figure~\ref{fig:ckbb_Xancilla}.

\begin{figure}[t]
    \centering
    \includegraphics*[scale = 1.25]{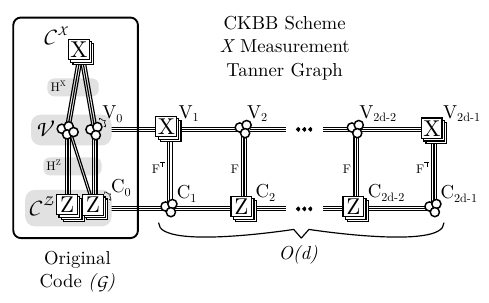}
    \caption{\label{fig:ckbb_Xancilla}Diagram of the CKBB scheme for measuring a logical $X$ operator. The original code's Tanner graph is a subgraph of this graph supported on $\mathcal{V},\mathcal{C}^X,$ and $\mathcal{C}^Z$.}
    \label{fig:hpm}
\end{figure}

In the resulting code the operator $\bar{X}_M$ is a product of stabilizers, namely all the appended $X$-checks $V_1, V_3, ..., V_{2d-1}$.  Therefore a code-switching scheme from $\mathcal{G}$ to the larger code can be used for a fault-tolerant logical measurement of $\bar{X}_M$. 
If any other $X$ logical operators were contained in the support of $\bar{X}_M$, then these would be measured as well. Therefore it is key that $\bar{X}_M$ is irreducible.

The number of ancilla qubits needed is $(2d-1)(|V_0| + |C_0|) > 2d^2$, which is a daunting space overhead. For codes where the distance $d$ scales as $\sqrt{n}$, this number of ancilla qubits may be more than the number of qubits used in the code itself. One of the primary goals of this paper is to reduce this space overhead by using fewer layers. 
The time overhead is $d$ rounds of error correction which is standard albeit on a bigger code.  We observe that if $\mc{G}$ has maximum degree $\delta$, then the CKBB ancilla system has maximum degree $\delta+1$.

 The CKBB ancilla system is a subsystem code with all the logical qubits from the original code $\mc{G}$ except the one being measured, and some gauge qubits which carry no logical information. It was proven in~\cite{cohen2022low} that the CKBB ancilla system has code distance $d$. The fault distance or circuit distance of such a measurement protocol is not analyzed.
 \section{Gauge-Fixed QLDPC Surgery}\label{sec:gauged}
In this section we present our main theoretical result: the gauge-fixed QLDPC surgery scheme. In Section~\ref{sec:gaugedXancilla} we introduce the basic form of the gauge-fixed ancilla system, under the context of measuring a $X$ logical operator. 
In Section~\ref{sec:schedule} we detail the $X$ logical measurement protocol and space-time detector graph. We prove the code distance of the merged code, as well as the fault distance of the space-time detector graph in Section~\ref{sec:dist}. We analyze the check weights of the gauge-fixed ancilla system in Section~\ref{sec:checkweight}, and introduce the modular decoder in Section~\ref{sec:modular_decoder}. 
In Sections~\ref{sec:jointmeasurement} and~\ref{sec:ysystem} we generalize the construction to multi-qubit measurements and $Y$ measurements, which require an additional bridge system. 
In Section~\ref{sec:bridge}, we discuss how to use our bridge construction as a code adapter and connect logical operators from different families of codes. 
We prove the code and fault distances of these measurements in their respective sections. 

\subsection{Gauge-fixed \texorpdfstring{$X$}{X} Ancilla Systems}\label{sec:gaugedXancilla}

We follow the notation used to describe the CKBB scheme in Section~\ref{sec:CKBBschemes}, since the initial part of the construction is the same. Recall that we are measuring an $X$-logical operator $\bar{X}_M$ with support $V_0 \subset \mc{V}$ and adjacent $Z$-checks $C_0 \subset\mc{C}^Z$. As before, we assume there are no other $X$ operators supported on a subset of $V_0$, and consider the induced Tanner graph of $\bar{X}_M$ with adjacency matrix $F: C_0 \to_Z V_0$.

We introduce $L$ additional layers $(C_1,V_1), (C_2,V_2), ..., (C_L,V_L)$ with vertex roles as in Eq.~\eqref{eqn:vertexroles} and the connectivity as in the CKBB system with $F: C_i \to_Z V_i$ for even $i$, $F^\top: V_i \to_X C_i$ for odd $i$ and $I$ for connections between layers. 

Attaching such an ancilla system introduces new stabilizers and gauge degrees of freedom to the code, which could reduce the code distance. 
Specifically, multiplying a logical operator of the original code by new stabilizers or gauge operators could reduce its weight below $d$. 
In \cite{cohen2022low}, this issue is resolved by taking $L = 2d-1$.

In this work, we promote these gauge degrees of freedom to stabilizers, resulting in a non-subsystem stabilizer code $\mathcal{G}_X$.
The $Z$ gauge operators of the CKBB scheme correspond to the nullspace of $F$, that is, $\mathrm{null}(F) = \{c \text{ s.t. } cF = 0\}$. This is because the rows of $F^\top : V_{i} \to_X C_i$ for odd $i$ denote $X$ checks within the layer $(C_i, V_i)$, and vectors orthogonal to these rows commute with these $X$ checks. Any particular gauge operator $Z(c)$ has equivalent realizations on any $C_j$ for any odd $j$, since we can multiply by $Z$ checks in the layers $C_i$ with $i$ even. That is, $Z(c\in C_j)\mathcal{H}^Z(c\in C_{j+1})=Z(cF\in V_{j+1})Z(c\in C_{j+2})=Z(c\in C_{j+2})$.

To find a set of additional $Z$ checks that remove these gauge qubits, we find a matrix $G$ whose rows span $\mathrm{null}(F)$, that is, $GF = 0$. 
Then we introduce a new set of $Z$ checks $U_L$ and connect $G : U_L \to_Z C_L$. 
Minimizing row and column weight of $G$ minimizes the degrees added to the Tanner graph. 
This completes the construction of $\mathcal{G}_X$. A sketch of $\mathcal{G}_X$ is shown in Figure~\ref{fig:Xancilla}. 

\begin{figure}[t]
    \centering
    \includegraphics*[scale = 1.25]{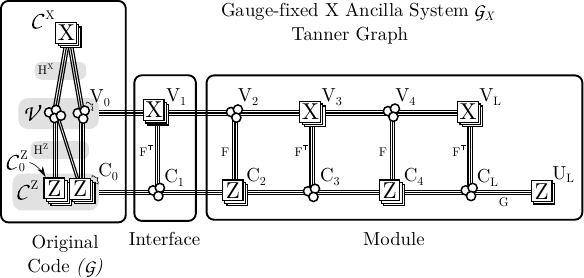}
    \caption{\label{fig:Xancilla} Diagram of the gauge-fixed $X$ ancilla system $\mathcal{G}_X$ with $L = 5$. The original code $\mathcal{G}$'s Tanner graph is a subgraph of this graph supported on $\mathcal{V},\mathcal{C}^X,$ and $\mathcal{C}^Z$. The first layer $(V_1,C_1)$ is called the interface, and the remaining layers and the gauge checks $U_L$ are called the module. }
\end{figure}

The resulting code has the following properties.

\begin{restatable}{theorem}{logicalqubit}\label{thm:validXancilla} 
Say $\mathcal{G}$ is any CSS code and $\bar{X}_M$ is a $X$-logical Pauli operator with no smaller $X$-logical operators in its support. Let $\mathcal{G}_X$ be the code defined above. Then $\bar{X}_M$ is a stabilizer of $\mathcal{G}_X$, and $\mathcal{G}_X$ has exactly one fewer logical qubit than $\mathcal{G}$.
\end{restatable}

We defer the proof of this theorem to Appendix~\ref{apx:Xproofs}. We note that without additional knowledge on the structure of the code or the measured operator, we do not have upper bounds on the row or column weight of $G$. 
While this is an undesirable caveat, 
we show theoretical evidence in Section~\ref{sec:checkweight} suggesting that in the average case, these gauge checks may not break the LDPC property.
In particular, the gauge checks may be products of other $Z$ stabilizers of $\mc{G}_X$ and hence do not need to be explicitly measured.

\subsection{\texorpdfstring{$X$}{X} Measurement Protocol}\label{sec:schedule}

We now describe the measurement protocol that allows us to extract the measurement result of $\bar{X}_M$ while leaving the remaining logical information undamaged. 
This protocol can be studied as a space-time decoding hypergraph $\mathcal{D}$, made of copies of the codes $\mathcal{G}$ and $\mathcal{G}_X$ and detectors in between. 
Our objective is to demonstrate that, although the check measurement outcomes are sometimes non-deterministic, we can still find a set of deterministically zero detector parities. 
We consider the noise-free case in this section, and assess the fault-tolerance properties of $\mathcal{D}$ in the next section.

$\mathcal{D}$ splits into steps $\mathcal{D}_0,...,\mathcal{D}_R$, corresponding to $R$ measurements of the checks of $\mathcal{G}_X$. The first and final steps $\mathcal{D}_0$  and $\mathcal{D}_d$ handle the code deformation between $\mathcal{G}$ and $\mathcal{G}_X$.
As shown in Figure~\ref{fig:Xancilla}, we call the first layer ($V_1, C_1$) of $\mc{G}_X$ the \textit{interface}, and the remaining layers are the \textit{module}. There is no module when $L = 1$, so the gauge checks $U_1$ are considered part of the interface. 
We call a check \emph{unreliable} if, even in the noise-free case, the outcome of the check is random since it doesn't commute with the stabilizer of the current state. A \emph{reliable} check is deterministic in the noise-free case, but not necessarily $+1$. 

In the context of this section, we use the notation $\mathcal{H}^Z(C_0)$ to denote the collection of checks in $C_0$ with qubit support given by $\mathcal{H}^Z$, instead of the product of these checks. 
The stabilizers $\mathcal{H}^Z(C_0)$ and $\mathcal{H}^Z(C_2)$ have support on the interface qubits $C_1$. 
$H^Z(C_0)$ are the stabilizers from $\mathcal{G}$ and therefore have no support on the interface.
Let $H_{\text{mod}}^Z(C_2)$ denote $Z$-checks at $C_2$ on the module, which also have no support on the interface.
These notation helps us keep track of which checks are reliable.

$\mathcal{D}$ is constructed with the following measurement schedule.  A small example is shown in Figure~\ref{fig:matchinggraph}.

\begin{figure}
\centering
\includegraphics[width=0.8\textwidth]{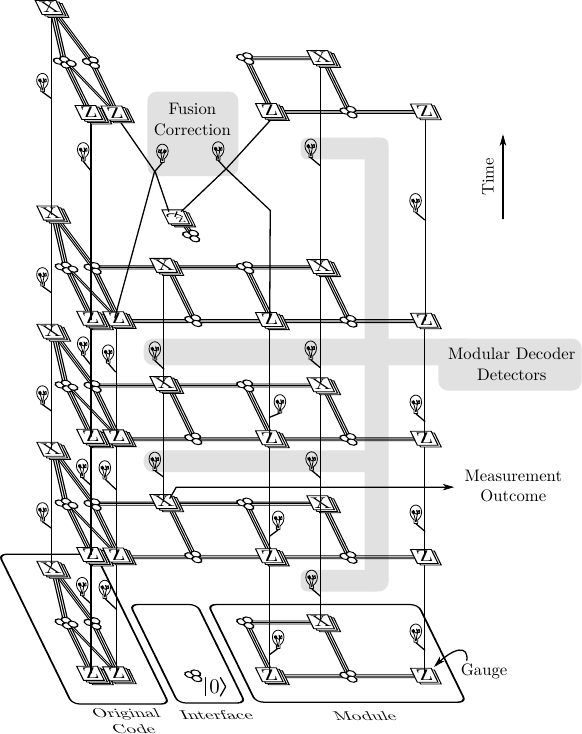}
\caption{\label{fig:matchinggraph} Space-time diagram and matching graph of a measurement protocol using an $X$ ancilla system. Detectors are denoted by lines connecting different checks and a lightbulb symbol. The errors identified by certain detectors can be decoded separately using a modular decoding approach - see Section~\ref{sec:modular_decoder}.}
\end{figure}

\begin{enumerate}[topsep = 0pt]
    \item Initialization:
    Initialize the interface qubits, namely qubits in $C_1$, to $\ket{0}$. If there is a module ($L > 1$), initialize a stabilizer state on the module qubits $C_2 ,V_3, \cdots, V_L$. It has the stabilizers $H_{\text{mod}}^Z(C_2)$, $\mathcal{H}^Z(C_i)$ for even $i \geq 4$,  $\mathcal{H}^Z(U_L)$, and $\mathcal{H}^X(V_i)$ for odd $i \geq 3$. We refer to these checks as the \textit{module stabilizer}. 

    Measure the module stabilizers. Additionally, either measure the stabilizers of $\mathcal{G}$, or take the measurement outcomes from a prior round.

    \item Fusion: 
    Measure the stabilizers of $\mathcal{G}_X$. Of these, the $\mathcal{H}^X(V_1)$ checks are unreliable, but everything else is reliable. Let $\mathcal{D}_0$ contain parities between pairs of reliable checks from step 1 and step 2 only. This involves taking parities between step 1's $H^Z(C_0)$ and step 2's $\mathcal{H}^Z(C_0)$, as well as step 1's $H_{\text{mod}}^Z(C_2)$ and step 2's $\mathcal{H}^Z(C_2)$. 
    The module stabilizers except $H_{\text{mod}}^Z(C_2)$ are identical in step 1 and step 2 so these parities are zero. From round 1 to 2, checks $H^Z(C_0)$ and $H_{\text{mod}}^Z(C_2)$ have gained support on the interface $C_1$ and become $\mathcal{H}^Z(C_0)$ and $\mathcal{H}^Z(C_2)$. But since the interface qubits are in the $\ket{0}$ state, the measurement outcome is unaffected, so these parities should be zero as well.

    Output the parity of the $\mathcal{H}^X(V_1)$ checks as the measurement result, since their product $ \mathcal{H}^X(V_1)$ is $\bar X_M$ in the noiseless case. This value should be subject to decoder corrections.

    \item Measure the stabilizers of $\mathcal{G}_X$ for $R-1 \ge 0$ more times. Let the detectors $\mathcal{D}_i$  for $0 < i < R$ be the parities between successive measurements. The parities should all be zero. 

    \item Split: Measure the stabilizers of $\mathcal{G}$, the interface qubits $C_1$ in the $Z$ basis, and the module stabilizers.

    Checks $H^Z(C_0)$, $H_{\text{mod}}^Z(C_2)$, and the interface qubit $Z$-measurements $C_1$ are all unreliable. However, we can still compute what the outcome of the reliable checks $\mathcal{H}^Z(C_0)$ and $\mathcal{H}^Z(C_2)$ would have been by considering the appropriate parities of $H^Z(C_0)$, $H_{\text{mod}}^Z(C_2)$ with the $Z$-measurements of $C_1$. The parities of these values with the measurements of $\mathcal{H}^Z(C_0)$ and $\mathcal{H}^Z(C_2)$ of the prior step feature in $\mathcal{D}_R$. These are shown as the \emph{fusion correction} in Figure~\ref{fig:matchinggraph}.

    As in step 1, if there is no module then gauge checks $\mathcal{H}^Z(U_1)$ can still be computed by considering the appropriate parities of the $Z$ measurements of $C_1$. $\mathcal{D}_R$ also features detectors comparing these parities to the $\mathcal{H}^Z(U_1)$ measurement from the prior round.
\end{enumerate}

The initialization of the module can be done independently and in advance. The module when viewed as a code itself has no logical qubits --- it is a single stabilizer state.  Therefore, the module is reusable in future rounds of logical measurements, if we keep it error-free. In other words, the module can be viewed as a resource state that facilitates fault-tolerant logical measurements. A similar module system appeared for a different purpose in \cite{xu2024constant}.

Although we let $R\ge1$, the number of times the stabilizers of the merged code $\mathcal{G}_X$ are measured, be arbitrary in this description, we will conclude below in Theorem~\ref{thm:logical_and_measurement_distance} that $R\ge d$ is sufficient for the fault-tolerance of the protocol. In practice, we may want to balance logical error rates and the error rate of the logical measurement, and this may mean $R<d$ offers better performance. We probe this trade-off in our case study in Section~\ref{sec:gross_code}.

\subsection{Code Distance and Fault Distance}\label{sec:dist}
We demonstrate that, if the original code $\mathcal{G}$ has distance $d$, then there exists a choice of layers $L$ such that $\mathcal{G}_X$ also has distance $d$. With $L$ chosen as such, we then demonstrate that the fault distance of the detector graph $\mathcal{D}$ is also $d$. Hence the protocol underlying $D$ implements a fault-tolerant logical measurement.

To understand the code distance of $\mathcal{G}_X$, we must analyze how the existing logical operators could change in shape after fusing the ancilla patch. We begin with $Z$-type logical operators, whose weight does not depend on $L$.

\begin{theorem}~\label{thm:Zdistance} 
Each nontrivial logical $Z$ operator in $\mathcal{G}_X$ has weight at least $d$.
\end{theorem}
\begin{proof}
Suppose $P$ is a $Z$ operator that commutes with the stabilizer of $\mathcal{G}_X$. We need to show that $P$ is either weight at least $d$ or is a product of $Z$ checks. Consider the the restriction of $P$ to the qubits $\mathcal{V}$ of the code $\mathcal{G}$, written $P_{|\mathcal{V}}$.  $P_{|\mathcal{V}}$ commutes with all $X$ checks of the original code $\mathcal{G}$ and is therefore a logical operator of that code. Assuming it is not in the stabilizer, its weight is at least $d$ which implies the weight of $P$ is also at least $d$. On the other hand, suppose $P_{|\mathcal{V}}$ is a product of $Z$ checks of the original code. This implies $P$ can be multiplied by a product of $Z$ checks from $\mathcal{G}_X$ to obtain an operator $S$ supported entirely on the interface and module qubits $(C_1\cup C_3\cup\dots C_L)\cup(V_2\cup V_4\cup\dots V_{L-1})$. Notably, $S$ is equal to $P$ modulo stabilizers from $\mathcal{G}_X$. We denote the support of $S$ with $c_j \in C_j$ for odd $j$  and $v_i \in V_i$ for even $i$, allowing us to write
\begin{equation}
S=\prod_{\substack{j=1\\j\text{ odd}}}^L Z(c_j \in C_j) \cdot \prod_{\substack{i=2\\i\text{ even}}}^{L-1} Z(v_i \in V_i),
\end{equation}
recalling notation from section~\ref{subsec:Tanner_graph}. Because $S$ commutes with all $X(V_L)$, we must have $ c_{L} F = v_{L-1}$. Similarly, because of commutation with checks in $V_j$ for all odd $j < L$ we have $c_j F  = v_{j-1}+v_{j+1}$. Finally, because of commutation with checks in $V_1$, $c_1 F =v_2$. Adding all these equations together modulo two results in the right-hand sides telescoping down to zero. If we define $c_{\text{tot}}=c_1+c_3+\dots+c_L$, we find $ c_{\text{tot}} F =0$. Because rows of $G$ span the row nullspace of $F$, there exists $r$ such that $rG=c_{\text{tot}}$ and $\mathcal{H}_Z(r\in U_L)=Z(c_{\text{tot}}\in C_L)$.

We can multiply $S$ by checks $\mc{H}_Z(c_1\in C_2)\mc{H}_Z(c_1+c_3\in C_4)\dots \mc{H}_Z(c_1+c_3+\dots+c_{L-2}\in C_{L-1})$ to obtain
\begin{equation}
S'= Z(c_{\text{tot}}\in C_L)\prod_{\substack{i=2\\i\text{ even}}}^{L-1} Z(v'_i\in V_i) = \mc{H}_Z(r\in U_L)\prod_{\substack{i=2\\i\text{ even}}}^{L-1} Z(v'_i\in V_i):= \mc{H}_Z(r\in U_L)S'',
\end{equation}
where $v'_i$ are some vectors that can be determined in terms of $c_j$ and $v_i$. However, $S''$ must therefore commute with all $X$ checks while being supported only on $V_2\cup V_4\cup\dots V_{L-1}$. It is easy to see this implies $S''=I$ and so $S'$ and $S$ are products of $Z$ checks in $\mathcal{G}_X$.
\end{proof}

Ensuring that the $X$ operators have high weight requires selecting $L$. The best choice of $L$ depends on the nature of the induced Tanner graph determined by $\bar X_M$. We derive a sufficient condition in terms of the \textit{boundary Cheeger constant} of the induced Tanner graph, which has adjacency matrix $F$.

\begin{definition}[Boundary Cheeger Constant] For a bipartite graph $F^\top : V \to C$ and a set of vertices $v \subset V$, define the \textit{boundary} $\partial v\subset C$ to be the set of vertices with an odd number of neighbors in $v$. In terms of indicator vectors, we may write $\partial v = vF^\top$. 
    The \emph{boundary Cheeger constant} of a Tanner graph is defined as
\begin{align}
    \beta = \min_{v\subset V, |v|\le |V|/2}|\partial v|/|v|.
\end{align}
\end{definition}

\begin{lemma}[Expansion Lemma]~\label{lem:expansion}
    If the induced Tanner graph of $\bar{X}_M$ is $F : C_0 \to V_0$, suppose $F^\top : V_1 \to C_1$ has boundary Cheeger constant $\beta$.
    Given a set $s\in V_0$, suppose we multiply $X(s)$ by a set of $X$-stabilizers in the ancilla patch, namely $v_j\in V_j$ for odd $j\in [L]$. Suppose $m = \argmin_{i}\min(|v_i|, |V_0| - |v_i|)$. If
    $L/2\ge 1/\beta$, then 
    \begin{align}
        |X(s)\cdot \prod_{\substack{j=1\\j\text{ odd}}}^L \MH_X(v_j\in V_j)| 
        &\ge |s+v_m| + |v_m + v_L| + \min(|v_m|, |V_0| - |v_m|), \\
        &\ge \min(|s|, |V_0| - |s|).
    \end{align}
\end{lemma}
\begin{proof}
Here we have
\begin{align}
    |X(s)\cdot \prod_{\substack{j=1\\j\text{ odd}}}^L \MH_X(v_j\in V_j)| 
    &= |s + v_1| + \sum_{\substack{i=3\\i\text{ odd}}}^L |v_i + v_{i-2}| + \sum_{\substack{j=1\\j\text{ odd}}}^L |v_jF^\top|, \\
    &\ge |s + v_1| + \sum_{\substack{i=3\\i\text{ odd}}}^L |v_i + v_{i-2}| + \sum_{\substack{j=1\\j\text{ odd}}}^L \beta \min(|v_j|,|V_0|-|v_j|), \\
    &\ge |s + v_1| + \sum_{\substack{i=3\\i\text{ odd}}}^L |v_i + v_{i-2}| + \min(|v_m|, |V_0| - |v_m|), \\
    &\ge |s + v_m| + \sum_{\substack{i=m+2\\i\text{ odd}}}^L |v_i + v_{i-2}| + \min(|v_m|, |V_0| - |v_m|), \\
    &\ge |s + v_m| + |v_m + v_L| + \min(|v_m|, |V_0| - |v_m|).
\end{align}
Note that the last two inequalities follow from triangle inequality. Our claim holds by another application of triangle inequality.
\end{proof}

The fact that $\bar X_M$ is a logical operator implies a condition on the support of other logical operators.
\begin{lemma}~\label{lem:Xopsprt}
Suppose $\bar{X}$ is a logical $X$ operator (inequivalent to $\bar{X}_M$) supported on $v \subset \mathcal{V}$. Then $|v\setminus V_0| + |V_0\setminus v| \geq d$.
\end{lemma}
\begin{proof} Notice $\bar{X}_M\bar{X}$, which is supported on $(v\setminus V_0) \cup (V_0\setminus v)$, is also a nontrivial logical operator, and therefore has weight at least $d$.
\end{proof}

Combining the two preceding lemmas, we obtain the following bound on the code distance of the merged code.

\begin{theorem}~\label{thm:Xdistance} If the induced Tanner graph of $\bar{X}_M$ is $F : C_0 \to V_0$, suppose $F^\top : V_1 \to C_1$ has boundary Cheeger constant $\beta$. If $\mathcal{G}_X$ is constructed with $L$ layers such that $\lceil L/2\rceil \ge 1/\beta$, then the code $\mathcal{G}_X$ has distance at least $d$.
\end{theorem}
\begin{proof} Theorem~\ref{thm:Zdistance} bounds the minimum weight of $Z$-type logical operators in $\mathcal{G}_X$ for all $L$. Only the logical $X$ operators remain to be bounded. 

Any $X$-type logical operator of code $\mathcal{G}$ is still a logical operator of the merged code $\mathcal{G}_X$. Irreducibility of $\bar{X}_M$ implies that this set of logical $X$-type operators provides a complete basis of $X$-type logical operators of the $k-1$ qubit code $\mathcal{G}_X$. Hence, consider any nontrivial $X$-type logical operator $P$ of $\mathcal{G}_X$ supported entirely on $\mathcal{V}$. It is also a logical operator of $\mathcal{G}$ so it has weight at least $d$. We will show $P$ cannot be reduced in weight to below $d$ by multiplying by the $X$-checks of $\mathcal{G}_X$. 

Let $P=X(t\in\mathcal{V})$ be an $X$-type logical operator of $\mathcal{G}$. Consider
\begin{align}
    P'=P\prod_{\substack{j=1\\j\text{ odd}}}^L\mathcal{H}_X(v_j).
\end{align}
Let $s = t\cap V_0$, from Lemma~\ref{lem:expansion}, we see that $|P'| \ge |t| - |s| + \min(|s| , |V_0| - |s|)$. If $|s|\le |V_0|/2$, we see that $|P'| \ge |t| \ge d$. Otherwise, we have 
\[
|P'|\ge |t| - |s| + |V_0| - |s| = |t\setminus V_0| + |V_0\setminus t| \ge d,
\]
by Lemma~\ref{lem:Xopsprt}. We conclude that the $X$ distance of $\mathcal{G}_X$ is at least $d$.
\end{proof} 

As a corollary, we see that the space overhead of our ancilla system is $O(d/\beta)$. 
We remark that boundary expansion is a demanding property, and for general codes we expect $1/\beta$ to be super-constant. 
However, in the case where every check in the induced Tanner graph acts on exactly two qubits in the support of $\bar{X}_M$, the induce Tanner graph can be interpreted as a simple graph by turning every check into an edge.
Boundary Cheeger constant on the induced Tanner graph then translates to the usual Cheeger constant on this simple graph. 
Since edge expansion is a much more common property, for such codes the space overhead may be linear.  
We note that in~\cite{williamson2024low} and~\cite{ide2024fault}, instead of using the induced Tanner graph (which is generally a hypergraph if we treat each induced check as a hyperedge) for measurement, the proposed schemes constructed and used a compatible expander graph with sufficient edge expansion (Cheeger constant $\ge 1$) for measurement. 
This novelty guarantees a low space-overhead ancilla system for any QLDPC code. We also note that in~\cite{swaroop2024universal}, the requirement on edge expansion is further relaxed to a notion called relative expansion.

We also consider Theorem~\ref{thm:Xdistance} a worst case upper bound on the number of layers needed and expect practical cases to be much better. 
The reasoning behind is simple: stabilizers in the ancilla system may only affect $X$-logical operators which intersect $\bar{X}_M$. 
To provide a worse case analysis, we assume no knowledge of the intersection structure of $X$-operators and consider product of all subsets of stabilizers $v\in V$. 
In reality, only those subsets that touches another $X$-logical operator needs to be expanding. 
This is exemplified by our case study on the [[144,12,12]] gross code, where we found operators for which a mono-layer ancilla is distance preserving, see Section~\ref{sec:gross_code}. It is also supported by
results from~\cite{cowtan2024ssip}, where it was verified numerically (for the CKBB scheme) that for many small-to-medium codes the number of ancilla layers needed is 5 or less. 

Next we analyze the detector graph $\mathcal{D}$ described in Section~\ref{sec:schedule} and give lower bounds on fault distance. For our purposes it makes sense to analyze two different fault distances: the \emph{logical} fault distance and the \emph{measurement} fault distance. These are defined similarly to fault distance in \cite{beverland2024fault}.

\begin{definition}\label{def:fault_distance} Consider a phenomenological noise model on the detector graph $\mathcal{D}$: Pauli errors can occur on qubits in between any rounds of measurement, and any measurement outcome can be faulty. We say a detector \emph{triggers} when the associated parity between measurement outcomes is 1, and we say an error is \emph{undetected} if it triggers no detectors. 
The \emph{logical fault distance} is the minimum weight of an undetected error that causes a logical Pauli error on a qubit in $\mathcal{G}$ other than the measured qubit and does not cause a measurement error. 
The \emph{measurement fault distance} is the minimum weight of an undetected error that causes an error in the measurement outcome. 
\end{definition}

Our proof uses the idea that code switching using Pauli measurements can be analyzed in terms of subsystem codes \cite{vuillot2019deformation}. Let $\bar{Z}_M$ be the symplectic partner to $\bar{X}_M$, $\hat e_l$ denote a basis vector, $C_1^\text{sep}=\{Z(\hat e_l\in C_1):l\in C_1\}$ denote the set of all single qubit Pauli Z operators on the qubits $C_1$, and $\mathcal{C}_0^Z$ be checks in $\mathcal{C}^Z$ that are not supported on $V_0$. We also let $V_{j\text{ odd}}=V_1\cup V_3\cup\dots V_L$ and $C_{i\text{ even}}=C_0\cup C_2\cup\dots C_{L-1}$. Consider the stabilizer group of the merged code $\mathcal{G}_X$
\begin{equation}
\mathcal{S}_X=\langle C^X,\mathcal{C}_0^Z,V_{j\text{ odd}},C_{i\text{ even}},U_L\rangle,
\end{equation}
and the stabilizer group of the original code, interface qubits in $\ket{0}$, and module, with or without $\bar{X}_M$ included
\begin{equation}
\mathcal{S}=\langle C^X,\mathcal{C}_0^Z,V_{j\ge3\text{ odd}},C_{i\text{ even}},C_1^\text{sep}\rangle,\quad \mathcal{S}'=\langle\mathcal{S},\bar{X}_M\rangle.
\end{equation}
Note $U_L\subseteq\mathcal{S}$ as written even though it is not explicitly included. Also consider the gauge group
\begin{equation}
\mathcal{U}=\langle C^X,\mathcal{C}_0^Z,V_{j\text{ odd}},C_{i\text{ even}},C_1^\text{sep}\rangle=\langle\mathcal{S},\mathcal{S}_X\rangle.\end{equation}
These groups satisfy the subgroup relations
\begin{equation}\label{eq:subgroup_relation}
\mathcal{S}'\le\mathcal{U},\quad \mathcal{S}_X\le\mathcal{U}.
\end{equation}
In particular, $\mathcal{S}'$ and $\mathcal{S}_X$ are two different ways of fixing the gauge of $\mathcal{U}$ to get stabilizer subspace codes.

All three codes corresponding to $\mathcal{S}'$, $\mathcal{S}_X$, and $\mathcal{U}$ have $k-1$ logical qubits and a logical basis for each can be chosen to consist of the same set of operators $\mathcal{L}=\langle\bar{X}_i,\bar{Z}_i:i=1,\dots,k-1\rangle$. To see this, note that any logical $X$ operator $\bar{X}_i$ (inequivalent to $\bar{X}_M$) of $\mathcal{S}$ on the original code qubits is still a logical operator of each of the three codes. Also, if $\bar{Z}'_i$ is a nontrivial logical $Z$ operator of $\mathcal{S}$ that commutes with $\bar{X}_M$, then we claim it can be multiplied by stabilizers of the original code to obtain an equivalent $\bar{Z}_i$ suitable for inclusion in $\mathcal{L}$ because it has no support on $V_0$. This statement is a corollary of the following lemma.
\begin{lemma}~\label{lem:supportlemma}
    Suppose $\bar{X}_M$ is irreducible. Then for any vector $v$ with even weight supported on $V_0$, there exists a vector $c$ supported on $C_0$ such that $c F = v$. 
\end{lemma}
\begin{proof}
    Since the only nonzero column vector $w$ such that $Fw = 0$ is the all $1$s vector, the rows of $F$ must span the checks of a classical repetition code. Any vector of even weight can be expressed as a linear combination of rows of $F$ since the checks of a repetition code generate all vectors of even weight. 
\end{proof}
\noindent We also use $\mathcal{L}^*=\mathcal{L}\setminus\{I\}$ to denote the nontrivial logical operators.

For Pauli $p$ and sets of Pauli operators $\mathcal{T}$ and $\mathcal{T}'$, let $p\mathcal{T}=\{pq:q\in\mathcal{T}\}$ and $\mathcal{T}\mathcal{T}'=\{pq:p\in\mathcal{T},q\in\mathcal{T}'\}$. Also, let $d(\mathcal{T})$ be the minimum weight of any Pauli in $\mathcal{T}$. Then the code distances of codes $\mathcal{S}'$, $\mathcal{S}_X$, and $\mathcal{U}$ are written $d(\mathcal{L}^*\mathcal{S}')$, $d(\mathcal{L}^*\mathcal{S}_X)$, and $d(\mathcal{L}^*\mathcal{U})$, respectively, and because of Eq.~\eqref{eq:subgroup_relation},
\begin{equation}\label{eq:distance_relations}
d(\mathcal{L}^*\mathcal{U})\le d(\mathcal{L}^*\mathcal{S}_X),\quad d(\mathcal{L}^*\mathcal{U})\le d(\mathcal{L}^*\mathcal{S}'),\quad d(\bar{Z}_M\mathcal{L}\mathcal{U})\le d(\bar{Z}_M\mathcal{L}\mathcal{S}').
\end{equation}

As the ``weakest" code out of the three, $\mathcal{U}$ and its code distance ends up determining the logical and measurement fault distances of our protocol. Let $R$ be the number of rounds we measure the checks of the merged code.
\begin{lemma}\label{lem:subsystem_codes_and_fault_distance}
The measurement fault distance of the $X$ system is $\min(R,d(\bar{Z}_M\mathcal{L}\mathcal{U}))$. The logical fault distance of the $X$ system is $d(\mathcal{L}^*\mathcal{U})$.
\end{lemma}
\begin{proof}
We denote by $\mathcal{M}_t$ the (abelian) group of Paulis generated by the checks measured in timestep $t$. In particular, $\mathcal{M}_0=\mathcal{S}=\mathcal{M}_{R+1}$ measures the original code, interface qubits in $\ket{0}$, and module (see the top and bottom of Figure~\ref{fig:matchinggraph}). For $0<t<R+1$, $\mathcal{M}_t$ measures the merged code stabilizers $\mathcal{S}_X$.

Measurement differences are denoted $\mathcal{M}_{t\rightarrow t'}=\mathcal{M}_{t}\cap\mathcal{M}_{t'}$ for $t'>t$. Each element of a detector group $\mathcal{M}_{t\rightarrow t'}$ is indeed a detector because, given its value at time $t$, in the absence of errors it will have the same value at time $t'$.  The explicit parities $\mathcal{D}_t$ described in Section~\ref{sec:schedule} suffice to indicate which detectors in $\mathcal{M}_{t\rightarrow t+1}$ changed from $t$ to $t+1$. While the $\mathcal{M}_{t\rightarrow t+1}$ are a generating set of detectors sufficient for decoding, it is also true that an undetected error in the sense of Definition~\ref{def:fault_distance} should be undetected by (i.e.~commute with) any detector group $\mathcal{M}_{t\rightarrow t'}$. This means that undetected errors are elements of the centralizer $\mathcal{C}(\mathcal{M}_{t\rightarrow t'})$.

Detector groups $\mathcal{M}_{t\rightarrow t'}$ fall into three cases. If $t=0$ and $0<t'<R+1$ or $0<t<R+1$ and $t'=R+1$, then $\mathcal{M}_{t\rightarrow t'}=\mathcal{S}\cap\mathcal{S}_X$, the stabilizers that are both in the original code and merged code. If $0<t,t'<R+1$, then $\mathcal{M}_{t\rightarrow t'}=\mathcal{S}_X$, the merged code stabilizers. Finally, if $t=0$ and $t'=R+1$, then $\mathcal{M}_{0\rightarrow R+1}=\mathcal{S}$, the original code stabilizers.

We make the standard modeling assumption that the final round is free of measurement errors, see Ref.~\cite{fowler2009high}, for instance. 
Because the final round is perfect, any qubit errors must at this point amount to an undetectable error, either a stabilizer or non-trivial logical error, which makes the definition of a logical failure clear-cut. We point out that an ideal round of measurement is also done in our numerical case study (see Section~\ref{sec:gross_code}) at $t=R+2$ for the same reason.

If some nontrivial Pauli error $P$ occurred on the qubits we will bound its weight. Such a Pauli error is necessary to cause a logical fault and is one way to cause a measurement fault. The other way to cause a measurement fault is with only measurement errors. Because we repeat the measurement of $\mathcal{S}_X$ (which contains $\bar{X}_M$) $R$ times, this way to cause a measurement fault requires at least $R$ measurement errors.

Let $t$ be the latest time before the first qubit error, and $t'$ be the earliest time after the last qubit error, so that by time $t'$ the accumulated qubit error is $P$. Hiding this error from detection with measurement errors in steps $t'$ and later will fail due to the perfect measurement round at $R+1$. Thus, $P$ must commute with all detectors in $\mathcal{M}_{t\rightarrow t'}$ (although it may also require measurement errors at times between $t+1$ and $t'-1$ to avoid detection at those times). Given the discussion above, this means $P$ is in the centralizer of $\mathcal{S}\cap\mathcal{S}_X$, $\mathcal{S}_X$, or $\mathcal{S}$ depending on $t$ and $t'$.

We provide generating sets of these centralizers so we can see exactly which elements lead to logical or measurement faults using Definition~\ref{def:fault_distance}. We deal with the three cases one at a time. First,
\begin{equation}\label{eq:generating_set_case_1}
\mathcal{C}(\mathcal{S}\cap\mathcal{S}_X)=
\left\langle\begin{array}{cccc}
\text{\colorbox{gray!30}{$\mathcal{S}\cap\mathcal{S}_X$}}&\text{\colorbox{gray!30}{$V_1$}}&\text{\colorbox{gray!30}{$\bar{X}_M$}}&\bar{X}_{i=1,\dots,k-1}\\
&\text{\colorbox{gray!30}{$C_1^\text{sep}$}}&\bar{Z}_M&\bar{Z}_{i=1,\dots,k-1}
\end{array}\right\rangle.
\end{equation}
The elements highlighted in gray are a generating set of $\mathcal{U}$ and the $\bar{X}_i$ and $\bar{Z}_i$ in the last column are the generators of $\mathcal{L}$. 
A measurement fault is an element of $\mathcal{C}(\mathcal{S}\cap\mathcal{S}_X)$ that includes $\bar{Z}_M$, the only generator of $\mathcal{C}(\mathcal{S}\cap\mathcal{S}_X)$ that anticommutes with $\bar{X}_M$. 
So, for this first case, measurement faults are exactly the elements of $\bar{Z}_M\mathcal{L}\mathcal{U}$. Similarly, logical faults are elements of $\mathcal{C}(\mathcal{S}\cap\mathcal{S}_X)$ that do not include $\bar{Z}_M$ (since Definition~\ref{def:fault_distance} specifies that logical faults do not flip the logical measurement) and include at least one of $\bar{X}_i$ or $\bar{Z}_i$. Thus, logical faults are exactly the elements of $\mathcal{L}^*\mathcal{U}$.

Second,
\begin{equation}\label{eq:generating_set_case_2}
\mathcal{C}(\mathcal{S}_X)=
\left\langle\begin{array}{cc}
\mathcal{S}_X&\bar{X}_{i=1,\dots,k-1}\\
&\bar{Z}_{i=1,\dots,k-1}
\end{array}\right\rangle.
\end{equation}
In this case, it is not possible to get a measurement fault, because $\bar{X}_M\in\mathcal{S}_X$ is a detector. Similar to the first case, logical faults are exactly the elements of $\mathcal{L}^*\mathcal{S}_X$.

Third,
\begin{equation}\label{eq:generating_set_case_3}
\mathcal{C}(\mathcal{S})=
\left\langle\begin{array}{ccc}
\text{\colorbox{gray!30}{$\mathcal{S}$}}&\text{\colorbox{gray!30}{$\bar{X}_M$}}&\bar{X}_{i=1,\dots,k-1}\\
&\bar{Z}_M&\bar{Z}_{i=1,\dots,k-1}
\end{array}\right\rangle.
\end{equation}
The highlighted elements are generators of $\mathcal{S}'$. Similar to the first case above, measurement faults in this third case are exactly elements of $\bar{Z}_M\mathcal{L}\mathcal{S}'$ and logical faults are exactly elements of $\mathcal{L}^*\mathcal{S}'$.

Therefore, combining the three cases we see that the logical fault distance is at least 
\begin{equation}
\min\left(d(\mathcal{L}^*\mathcal{S}'),d(\mathcal{L}^*\mathcal{U}),d(\mathcal{L}^*\mathcal{S}_X)\right)=d(\mathcal{L}^*\mathcal{U}),
\end{equation}
where we used Eq.~\eqref{eq:distance_relations}. The logical fault distance is exactly $d(\mathcal{L}^*\mathcal{U})$ because an error from $\mathcal{L}^*\mathcal{U}$ can happen between $\mathcal{M}_0$ and $\mathcal{M}_1$ without necessitating any measurement errors to hide it from detection.

Similarly, combining the three cases with the case of $R$ measurement errors we find the measurement fault distance is at least $\min\left(R,d(\bar{Z}_M\mathcal{L}\mathcal{S}'),d(\bar{Z}_M\mathcal{L}\mathcal{U})\right)=\min\left(R,d(\bar{Z}_M\mathcal{L}\mathcal{U})\right)$, again using Eq.~\eqref{eq:distance_relations}. Again, this is exactly the measurement fault distance because all qubit errors could occur between $\mathcal{M}_0$ and $\mathcal{M}_1$.
\end{proof}

Though we proved it in the special case of measuring $\bar{X}_M$ for simplicity, this lemma can be generalized to apply to any of the ancilla measurement schemes in this paper. There is always an original stabilizer group $\mathcal{S}_O$ and merged code group $\mathcal{S}_M$ measured $R\ge1$ times. Let $\mathcal{U}=\langle\mathcal{S}_O,\mathcal{S}_M\rangle$. If $\bar{P}_M$ is the logical operator being measured, let $\bar{Q}_M$ be a Pauli that anticommutes with it. We can establish a group of logical operators for $k-1$ qubits $\mathcal{L}$ that is the same for all the codes, and such that elements of $\mathcal{L}$ commute with both $\bar{P}_M$ and $\bar{Q}_M$. The lemma then applies with $\bar{Q}_M$ replacing $\bar{Z}_M$. For instance, if we are measuring $\bar{P}_M=\bar{X}_1\bar{X}_2$ as in the next section, we can take $\bar{Q}_M=\bar{Z}_1$, set up $\mathcal{L}=\langle\bar{X}_2,\bar{Z}_1\bar{Z}_2,\bar{X}_i,\bar{Z}_i,\forall i=3,...k\rangle$, and calculate $d(\bar{Q}_M\mathcal{L}\mathcal{U})$ and $d(\mathcal{L}^*\mathcal{U})$ to bound the measurement and logical fault distances.

However, when we measure $X$-type (or, symmetrically, $Z$-type) operators on a CSS code, we can be more specific about the fault distances in Lemma~\ref{lem:subsystem_codes_and_fault_distance}. We say a Pauli is $X$-type (resp.~$Z$-type) if it acts as either $I$ or $X$ (resp.~$Z$) on all qubits. Let $d_X(\mathcal{T})$ and $d_Z(\mathcal{T})$ indicate the lowest weight of any $X$- or $Z$-type operator in Pauli set $\mathcal{T}$, respectively. Then we have the following lemma for the case of measuring $\bar{X}_M$.
\begin{lemma}\label{lem:CSS_subsystem_case}
With $\mathcal{S}$, $\mathcal{S}_X$, $\mathcal{U}$, and $\mathcal{L}$ defined above for the $X$ ancilla system, $d(\mathcal{L}^*\mathcal{U})=\min(d_Z(\mathcal{L}^*\mathcal{S}),d_X(\mathcal{L}^*\mathcal{S}_X))$ and $d(\bar{Z}_M\mathcal{L}\mathcal{U})=d_Z(\bar{Z}_M\mathcal{L}\mathcal{S})$.
\end{lemma}
\begin{proof}
Since $\mathcal{U}$ and $\mathcal{L}$ are generated by only $X$- and $Z$-type Paulis (i.e.~they are CSS) and $\bar{Z}_M$ is $Z$-type, 
\begin{equation}
d(\mathcal{L}^*\mathcal{U})=\min(d_Z(\mathcal{L}^*\mathcal{U}),d_X(\mathcal{L}^*\mathcal{U}))\quad\text{and}\quad d(\bar{Z}_M\mathcal{L}\mathcal{U})=d_Z(\bar{Z}_M\mathcal{L}\mathcal{U}).
\end{equation}
Now note that the $Z$-type operators of $\mathcal{U}$ are exactly the same group of $Z$-type operators in $\mathcal{S}$, and so $d_Z(\mathcal{L}^*\mathcal{U})=d_Z(\mathcal{L}^*\mathcal{S})$ and $d(\bar{Z}_M\mathcal{L}\mathcal{U})=d_Z(\bar{Z}_M\mathcal{L}\mathcal{S})$. Likewise, the $X$-type operators of $\mathcal{U}$ are the same group as the $X$-type operators of $\mathcal{S}_X$, and so $d_X(\mathcal{L}^*\mathcal{U})=d_X(\mathcal{L}^*\mathcal{S}_X)$.
\end{proof}

In particular, if the original code has code distance $d$, then $d_Z(\mathcal{L}^*\mathcal{S})\ge d$ and $d_Z(\bar{Z}_M\mathcal{L}\mathcal{S})\ge d$ are automatic. Thus, Lemmas \ref{lem:subsystem_codes_and_fault_distance} and \ref{lem:CSS_subsystem_case} combine to show the fault distance of the $X$ ancilla system is large provided the merged code $\mathcal{S}_X$ has large $X$ distance.

\begin{theorem}\label{thm:logical_and_measurement_distance}
Let $d$ be the distance of the original code. If the stabilizers of the merged code are measured at least $d$ times, the measurement fault distance of the $X$ system is at least $d$. If the merged code has $X$ distance at least $d$ (e.g.~satisfies Theorem~\ref{thm:Xdistance}), the logical fault distance of the $X$ system is also at least $d$.
\end{theorem}

We note that initializing the module code on qubits $V_{i\ge2\text{ even}}\cup C_{j\ge3\text{ odd}}$ was used in our $\bar{X}_M$ measurement protocol just to make a specific choice of codes at the merge and split steps. This is flexible to some extent. For instance, one could instead initialize all qubits $V_{i\ge2\text{ even}}\cup C_{j\ge1\text{ odd}}$ to $\ket{0}$ and make the same arguments to arrive at Lemma~\ref{lem:subsystem_codes_and_fault_distance} and Theorem~\ref{thm:logical_and_measurement_distance}. The key property we need from a choice of initialization is that the $Z$ checks are reliable at the merge and split.

\subsection{Redundancy of Gauge Checks}\label{sec:checkweight}
As discussed in section~\ref{sec:gaugedXancilla}, without further assumption on the induced Tanner graph $F$, we cannot prove that $\text{null}(F)$ has a low weight basis. 
This means we do not have a worst case upper bound on the weight of gauge checks in $U_L$. 
We believe that this is not an issue in average case applications, for two reasons.
First, many of these gauge checks may be redundant and therefore do not need to be explicitly measured. 
Second, for a particularly heavy gauge check, we may decompose it into small gauge checks by a process we call cellulation. 

We first discuss redundancy. Let $c\in C_1$ be such that $cF = 0$, then $Z(c)$ is a gauge-fixed operator in the ancilla system. Observe that 
\[
Z(c)\cdot \mathcal{H}_Z(c\in C_0) = H_Z(c\in C_0),
\]
where the right hand side is a stabilizer in the original code $G$, with no support on $V_0$. 
If there exists a collection of checks $\bar{c}\subset \mathcal{C}_0^Z$ such that $H_Z(\bar{c}) = H_Z(c)$, then the gauge check $Z(c)$ will be redundant in the merged code. 
To summarize, we have just shown the following lemma.

\begin{lemma}~\label{lem:redundancy}
    For a vector $c\in C_0$ with $cF = 0$, if there exists vector $\bar{c}\in \mathcal{C}_0^Z$ such that $H_Z(\bar{c}) = H_Z(c)$, then the gauge check $Z(c\in C_1)$ is redundant in the merged code $\mathcal{G}_X$ and can be omitted from the Tanner graph. 
\end{lemma}

This observation suggests an important heuristic for applying our scheme: if the stabilizers of the original code has large degrees of redundancy, then many (if not all) the gauge checks in our ancilla system may be omitted. 
We note that this redundancy condition is not hard to satisfy. 
Many important families of QLDPC codes, such as hypergraph product codes and lifted product codes, have structured redundancy among its stabilizer checks. We demonstrate this principle with the example of hypergraph product codes. 

Consider a classical code specified by a binary $m\times n$ parity check matrix $H$. 
The hypergraph product code $Q = \text{HGP}(H, H^\top)$ has parity check matrices $H_Z = [H\otimes I_n, I_m\otimes H^\top]$ and $H_X = [I_n\otimes H, H^\top\otimes I_m]$. Suppose $\rank(H) = r$, then $Q$ is defined on $n^2 + m^2$ physical qubits, and has logical dimension $(n-r)^2 + (m-r)^2$. 

We specify a basis for the $X$-logical operators of $Q$. Let $\im(\cdot)$ denote the image of a linear operator, we first choose bases for four linear spaces. Let $W = \{w_1, \cdots, w_{n-r}\}$ be an arbitrary basis of $\ker{H}$, $V = \{v_1, \cdots, v_{m-r}\}$ be an arbitrary basis of $\ker{H^\top}$. For the vector space $\FF_2^n/\im(H^\top)$, 
we take the reduced row echelon form of $H$, and let $J\subset [n]$ denote the indices of columns with pivots. We can take $R = \{e_i: i\in [n]\setminus J\}$, which is a basis of $\FF_2^n/\im(H^\top)$ such that every vector has weight one. For $\FF_2^m/\im(H)$, we pick $S = \{s_1, \cdots, s_{m-r}\}$ similarly. 
Now consider the sets
\begin{align}
    X_{n\times n} &= \{(w\otimes r, 0_{m\times m}): w\in W\subset \ker(H), r\in R\subset \FF_2^n/\im(H^\top) \}, \\
    X_{m\times m} &= \{(0_{n\times n}, s\otimes v): s\in S\subset \FF_2^m/\im(H), v\in V\subset \ker(H^\top)\}.
\end{align}
We cite as a fact that these two sets specify a complete basis of the $X$-logical operators of $Q$. Note that $|X_{n\times n}| = (n-r)^2$ and $|X_{m\times m}| = (m-r)^2$.
Fix an $X$-operator $P$ supported on $(w\otimes r, 0_{m\times m})\in X_{n\times n}$, we will analyze the gauge-fixed ancilla system built on $P$. The case where $P$ is supported on a basis vector in $X_{m\times m}$ follows similarly. 

Consider the induced Tanner graph of $P$, which consists of qubits $V_0 = \supp(w\otimes r)$ and all $Z$-stabilizer checks that touch qubits in $V_0$. By construction of the hypergraph product code $Q$, we can label the $Z$-checks of $Q$ by $f_j\otimes e_i$ where $f_j$ ranges over standard basis vectors of $\FF_2^m$, and $e_i$ ranges over standard basis vectors of $\FF_2^n$. The support of such a $Z$-stabilizer is
\begin{align}
    H_Z(f_j\otimes e_i) = ((H^\top f_j)\otimes e_i, f_j\otimes (He_i)).
\end{align}
Let $C = \{f_j\in \FF_2^m: H^\top f_j\cap w\ne \varnothing\}$, then $C\otimes r = \{f\otimes r: f\in C\}$ is exactly the set of $Z$-checks of $Q$ that touches $V_0$. 
Let $F$ be the submatrix of $H$ induced by $C$ and $\supp(w)$, namely, $F: C \rightarrow \supp(w)$. Since $r$ is a standard basis vector, we see that $F$
is also the submatrix of $H_Z$ induced by the qubits $V_0$ and $Z$-checks $C_0 = C\otimes r$. Therefore, in a $L$ layer gauge-fixed ancilla system build on $P$, we can identify every $C_i$ as a copy of $C$, and every $V_i$ as a copy of $\supp(w)$. 

The gauge $Z$-check of our system corresponds to vectors $c\in \text{null}(F)\subset C$, i.e., $c F = 0$. 
We consider a subset of these gauge checks which corresponds to vectors in $\text{null}(H)$. 
Specifically, for every $u\in \text{null}(H)$, consider the restriction of $u$ to the indices in $C$, which we denote by $u\vert_C$. 
By construction of $C$, we see that $u\vert_C F = 0$ which means $u\vert_C\in \text{null}(F)$.
We show that the gauge checks corresponding to these vectors are redundant. 

Fix a vector $c\in C$ such that $c = u\vert_C$ for some $u\in \text{null}(H)$.
Consider the gauge check $Z(c)$ and its equivalent representative supported on qubits in $C_1$. 
Let $\MH_Z$ denote the new $Z$-check matrix with the ancilla system. Consider the product of $Z$-checks $c\otimes r\in \mathcal{C}^Z$, which is
\begin{align}
    \MH_Z(c\otimes r\in \mathcal{C}^Z) 
    &= Z(c\in C_1)\cdot H_Z(c\otimes r) = Z(c\in C_1)\cdot Z((H^\top c)\otimes r, c\otimes (Hr)).
\end{align}
Since $H^\top c = 0$, we have
\begin{align}
    Z(c\in C_1)\cdot \MH_Z(c\otimes r\in \mathcal{C}^Z)
    &= Z(0_{n\times n}, c\otimes (Hr)).
\end{align}
By construction, $r\in \FF_2^n/\im(H^\top)$. Therefore, there does not exist $\bar{c}\in \FF_2^m$ such that $H^\top\bar{c} = r$. In other words, $r$ is not in the row space of $H$, so $r$ is not orthogonal to the codespace $\ker(H)$. This implies that there exists vector $\bar{v}\in \ker(H)$ such that $\bar{v}\cdot r = 1$. Let $v = \bar{v} - r$, then $Hv = Hr$. Consider the product of $Z$-checks $c\otimes v$, we have 
\begin{align}
    \MH_Z(c\otimes v) 
    &= H_Z(c\otimes v)\\
    &= Z((H^\top c)\otimes v, c\otimes (Hv)) \\
    &= Z(0_{n\times n}, c\otimes (Hr)).\\
    Z(c\in C_1) 
    &= \MH_Z(c\otimes v)\cdot \MH_Z(c\otimes r\in \mathcal{C}^Z).
\end{align}
Note that the first equation is true because the $Z$-checks in $c\otimes v$ are kept unchanged by the ancilla system, which is only attached to the checks $C\otimes r$. This proves that the gauge check $Z(c\in C_1)$ is a product of $Z$-stabilizers in $\mathcal{C}^Z$, and is therefore redundant.
We expect that similar analysis can be made for lifted product codes and potentially other families of QLDPC codes.

In an earlier version of this paper, we mistakenly claimed that all gauge checks $Z(c)$ for $c\in \text{null}(F)$ are redundant. 
We thank John Blue for pointing out an error in the earlier proof, which is now corrected.

In cases where there is a heavy, non-redundant gauge check, a further technique called cellulation \cite{hastings2021weight,sabo2024weight} could be employed. 
For a gauge check $Z(c)$, enumerate the checks in $c$ as $c_1, \cdots, c_k$. 
For a subset of indices $I\subset [k]$, let $\bar{c}_I$ denote the product of checks $\prod_{i\in I}\mathcal{H}^Z(c_i)$. 
If there exists $I\subset [k]$ such that $\bar{c}_I$ is low weight (where the definition of low depends on the use case), we can add $\bar{c}_I$ as a redundant check into the original code, and build the ancilla system on the modified code. 
This procedure breaks $Z(c)$ into a product of two gauge checks: one with $Z$ acting on $c_i$ for $i\in I$ and $\bar{c}_I$, the other with $Z$ acting on $c_j$ for $j\in [k]\setminus I$ and $\bar{c}_I$.
Each of these gauge checks now have lower weight than $Z(c)$. 
We remark that \cite{williamson2024low} utilized cellulation in conjunction of another technique called decongestion~\cite{freedman2021building} to guarantee that the measurement ancilla is LDPC for any QLDPC code.

\subsection{Modular Decoder Achieving Measurement Fault-Distance}\label{sec:modular_decoder}
Above we showed that the phenomenological fault-distances of the decoding graph $\mc{D}$ are $d$ for both measurement errors and logical errors whenever it is constructed with at least $R \geq d$ rounds. To ensure practicality of the scheme we must also ensure that a decoder can correct errors of weight $< d/2$. It is reasonable to assume that such a decoder already exists for the code $\mc{G}$ on its own. In this section we show that if there is a good decoder for the original code $\mc{G}$, then there also exists a good decoder for extracting measurement outcomes from the measurement protocol specified by $\mc{D}$.  

Since $\mc{G}$ and $\mc{G}_X$ are CSS codes, $X$ and $Z$ errors (and $Z$ and $X$ measurement errors) can be decoded separately. Since the measurement outcome is only affected by $Z$ errors (and $X$ measurement errors), we restrict our attention to these errors specifically. 

We divide the parity checks in $\mc{D}$ into two parts. $\DG$ contains the detectors on $\mc{C}^{X}$, the $X$ checks of the original code $\mc{G}$. $\Danc$ contains the checks on the remaining $X$ checks of $\mc{G}_X$ corresponding to vertices $V_j$ for odd $j$. We consider a modular approach where we first run a given decoder on $\DG$ only, and propagate the resulting corrections to the $X$ checks in $\Danc$. We show that there exists a decoder for decoding the remaining errors in $\Danc$ that cannot result in a correction that flips the measurement outcome unless the error has weight $\geq d/2$.

The measured observable $\bar X$ is the product of all $X$ checks in $\mc{G}_X$. Hence an error that flips the measurement outcome must flip an odd number of $X$ checks. Indeed, errors that flip an even number of checks can be identified using Lemma~\ref{lem:supportlemma}.

The main idea is to reduce decoding to a matching problem on a chain, where every node in the chain corresponds to a step in $\Danc$. Since errors that flip an even number of $X$-checks within a round of measurements cannot affect the measurement outcome, we can restrict our attention to the total parity of the detectors within each step. Then the matching decoder is the only part of the decoding algorithm that could introduce a measurement error. It is well-known that a matching decoder on a line of length $R$ can only fail to correct an error if the error has weight $\geq R/2$, which is $\geq d/2$ since $R \geq d$.  

\begin{theorem} Suppose there exists a decoder that, for any collection $E_\mc{G}$ of $Z$-qubit and $X$-measurement errors of weight $< d/2$ in the detector graph $\DG$, outputs a correction $E'_\mc{G}$ consisting of $Z$-qubit errors and $X$-measurement errors consistent with the detectors and equivalent to $E_\mc{G}$ up to the stabilizers of $\mc{G}$. 

Then there exists a decoder for $\mc{D}$ as a whole that outputs a correction $E'$ such that, if the true error is $E$ of weight $< d/2$, the product $EE'$ triggers no detectors and commutes with the measured observable $\bar X_M$.
\end{theorem}
\begin{proof} The decoder operates as follows. First we run the provided decoder on $\DG$ to obtain the correction $E'_\mc{G}$. The $Z$ errors in $E'_\mc{G}$ may affect some $X$-checks on $V_0$. We flip the associated detectors in $\Danc$, so that the detector values in $\Danc$ only indicate $Z$-qubit errors a $X$-measurement errors on the interface and module.

Consider the steps $\Danc_1, \Danc_2, ... \Danc_{R}$. We call a step $\Danc_i$ \emph{odd} if an odd number of detectors trigger, and \emph{even} otherwise. Now we solve a matching problem on a line: the odd steps are sorted into adjacent pairs, possibly involving the endpoints of the sequence. In particular, suppose $\Danc_{j_1}, \Danc_{j_2},..., \Danc_{j_k}$ are the odd steps with $j_1 < j_2 < ... < j_k$. If $k$ is even then the two possible matchings are:
\begin{align}
(\Danc_{j_1}, \Danc_{j_2}), (\Danc_{j_3}, \Danc_{j_4})..., (\Danc_{j_{k-1}}, \Danc_{j_k}),\\
(\text{start},\Danc_{j_1}), (\Danc_{j_2}, \Danc_{j_3}), ..., (\Danc_{j_{k}}, \text{end}).
\end{align}
and if $k$ is odd, the two possible matchings are:
\begin{align}
(\Danc_{j_1}, \Danc_{j_2}), (\Danc_{j_3}, \Danc_{j_4})..., (\Danc_{j_{k}}, \text{end}),\\
(\text{start},\Danc_{j_1}), (\Danc_{j_2}, \Danc_{j_3}), ..., (\Danc_{j_{k-1}}, \Danc_{j_{k}}).
\end{align}
We select the matching that minimizes the length between paired steps. For example, the lengths of the matchings when $k$ is odd are $(j_2 - j_1) + (j_4 - j_3) + ... + (R - j_k)$, and $(j_1 - 0) + (j_2 - j_1) + ... + (j_{k} - j_{k-1})$ respectively.

We now identify errors in $E'$ that cause each pair of odd steps $(\Danc_j, \Danc_k)$ to become even. We select any single $X$-check in $V_1$ and insert measurement errors for all measurements in between the two steps. Inserting such errors flips one detector each in $\Danc_j$ and $\Danc_k$, making them even, while keeping the steps in between identical. A similar process works for $(\text{start},\Danc_{j_1})$ and $(\Danc_{j_{k}}, \text{end})$ pairs since flipping a $V_1$ measurement result in the first or final round of measurements only affects detector steps $\Danc_{1}$ and $\Danc_{R}$ respectively.

After introducing the above errors, all the steps are even. However, a particular even step may still feature an odd number of triggered detectors within each group of $X$-check vertices $V_j$ for odd $j$. Since the step itself is even, there must be an even number of such odd groups. $Z$-qubit errors on $V_i$ for even $i$ in that layer flip pairs of detectors based on $V_j$ checks within one step.  Hence we can find a choice of such $Z$-errors that cause all groups of $X$-checks on $V_j$ for odd $j$ to be even.

To finish the decoding, we leverage Lemma~\ref{lem:supportlemma} to identify $Z$ errors on each $C_j$ for odd $j$ that flip the corresponding $X$ checks on $V_j$. We return the product of all identified errors and $E'_\mc{G}$ as $E'$ for $\mc{D}$ as a whole. Since inserting these errors untriggered all detectors, the product $EE'$ must trigger no detectors.

Suppose the true error on $\mc{D}$ was $E$ of weight $< d/2$. It remains to show that $EE'$ commutes with the measured observable $\bar X$. $\bar X$ is the product of all the $X$-checks on $V_j$ for odd $j$. Here we point out that several parts of $E$ introduced above do not affect the measurement result. Hence a worst-case error $E$ does not feature these parts, since they needlessy increase the weight. This lets us restrict our attention to corrections $E$ of a certain form. Since $E_\mc{G}E'_\mc{G}$ is equivalent to the identity up to the stabilizer of $\mathcal{G}$ by assumption, we can restrict our attention to an error $E$ of weight $<d/2$ supported on $\Danc$ only. Furthermore, since the errors identified from Lemma~\ref{lem:supportlemma} flip an even number of $X$-checks in any $V_j$ for odd $j$, these must commute with $\bar X$ as well. So we may furthermore assume that $E$ flips at most one detector for each $V_j$ within each layer. Similarly, flipping an even number of bits within the odd $V_j$ groups within an even layer also commutes with the measurement outcome, so can assume that errors only occur on $V_1$.

So the only part of the decoding process that could cause $EE'$ to flip the measurement outcome is the matching problem. Measurement errors on $V_1$ introduced as part of $E'$ could form a chain of errors within $EE'$ connecting the first round to the final round. Then $EE'$ triggers no detectors but flips the logical measurement outcome. Suppose $EE'$ is as such. Observe that the sum of the distances between paired layers for the two possible matchings is $R$, for both the even and odd $k$ case. That means that one of them has length $\leq R/2$ and the other one $\geq R/2$. Since the decoder always picks the shorter matching, it will always pick one of length $\leq R/2$. For $EE'$ form the entire chain, $E$ must then correspond to the matching of length $\geq R/2$. 
\end{proof}

\begin{remark}
    We note that while the decoding strategy described above correctly achieves decoding distance $\lfloor d/2\rfloor$, as stated it does not have a threshold. 
    The reason is that under stochastic noise, the parity of each step $\Danc_i$ is odd with probability converging to $1/2$, which means the matching decoder on a line will fail with probability converging to $1/2$. This is a shortcoming of the proof, but not of the modular decoding method itself. 
    We can instead use minimum weight decoders on both $\DG$ and $\Danc$, which we expect to have a threshold asymptotically.

    We further observe that if every qubit in $C_1$ is adjacent to exactly two checks from $V_1$, then we can simply use a minimum-weight perfect matching (MWPM) decoder on $\Danc$. This condition is equivalent to every $Z$-check in $C_0$ touching two qubits in $V_0$. This is not true for arbitrary codes and logical operators. However, in our case study on the $\llbracket 144,12,12\rrbracket$ code in Section~\ref{sec:gross_code}, we find logical operators where this condition is satisfied. We therefore employ the modular decoder with BPOSD and MWPM in our numerical simulations presented in Section~\ref{sec:gross_code}.
\end{remark}

\subsection{Joint \texorpdfstring{$X$}{X} Measurement}\label{sec:jointmeasurement}

In this section, given ancilla systems for measuring two logical $X$ operators $\bar{X}_1$ and $\bar{X}_2$ that do not share support, we present a method to measure the joint operator $\bar{X}_1\bar{X}_2$ by connecting the two systems. For instance, this joint system can be used to connect two distinct codeblocks provided they are each equipped with an ancilla system, and in this case, we prove that the fault and code distances of the joint system are the minimum of the distances of the individual systems. If $\bar{X}_1$ and $\bar{X}_2$ do share support, then $\bar{X}_1\bar{X}_2$ is best measured with an $X$ ancilla system as described in Section~\ref{sec:gaugedXancilla}.

Our method of joint measurement is distinct from the proposal in \cite{cohen2022low}, not just from the fact that our ancilla systems are gauge-fixed, but also because their approach actually requires $2d-1$ layers to preserve the code distance and, more specifically, to preserve the weight of $\bar{X}_1$ or $\bar{X}_2$, which are still logical operators in the merged code. Instead of more layers, we introduce a collection of $\min(|\bar{X}_1|,|\bar{X}_2|)$ qubits that are connected to $X$ checks in both ancilla systems. We term these qubits \emph{bridge} qubits. The joint system is depicted in Figure~\ref{fig:joint_ancilla}.

\begin{figure}
\centering
\includegraphics[width=0.75\textwidth]{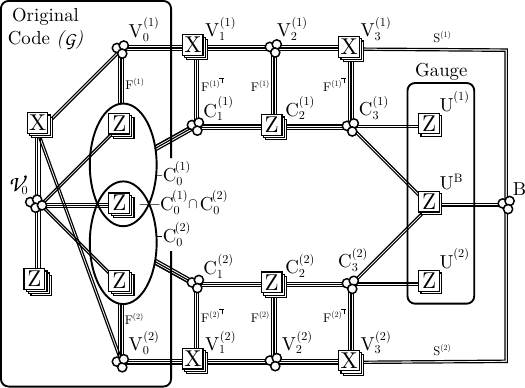}
\caption{\label{fig:joint_ancilla} The ancilla system for measuring the joint logical operator $\bar{X}_1\bar{X}_2$, where $\supp{\bar{X}_1}=V_0^{(1)}$ and $\supp{\bar{X}_2}=V_0^{(2)}$ are disjoint. The joint system is exactly the two systems for measuring $\bar{X}_1$ and $\bar{X}_2$ with the addition of bridge qubits $B$ and gauge checks $U^{B}$. Checks in $C_0^{(1)}\cap C_0^{(2)}$ are each connected to two ancilla qubits, one from $C_1^{(1)}$ and one from $C_1^{(2)}$. Although not strictly necessary, the drawing assumes that the number of layers in the $\bar{X}_1$ and $\bar{X}_2$ systems are the same and the bridge connects the last layers together. In general, the bridge $B$ can be used to connect together any odd layer $j_1$ from the first system and odd layer $j_2$ of second.}
\end{figure}

To describe how the bridge qubits are connected, we need to establish notation distinguishing the two ancilla systems for measuring $\bar{X}_1$ and $\bar{X}_2$ individually. Generally, we keep the notation from Figure~\ref{fig:Xancilla}, but we superscript sets of qubits and checks by $(1)$ or $(2)$. The collection of bridge qubits are indicated by $B$ and a new set of gauge checks that act on bridge qubits is denoted $U^{B}$. See again Figure~\ref{fig:joint_ancilla}.

Suppose $j_1$ and $j_2$ are any odd integers, each specifying a particular layer from the $\bar{X}_1$ and $\bar{X}_2$ systems respectively. Each bridge qubit is checked by one check from $V_{j_1}^{(1)}$ and one check from $V_{j_2}^{(2)}$. To specify connections between checks $V_{j_1}^{(1)}$ and qubits $B$, we need only give a binary matrix $S^{(1)}$ with exactly one 1 per column (each bridge qubit is connected to some check) and at most one 1 per row (checks are connected to at most one bridge qubit). A similar matrix $S^{(2)}$ specifies how $V_{j_2}^{(2)}$ and $B$ are connected. Since $|B|=\min(|\bar{X}_1|,|\bar{X}_2|)=\min(|V_{j_1}^{(1)}|,|V_{j_2}^{(2)}|)$, at least one of $S^{(1)}$, $S^{(2)}$ is a square matrix with exactly one 1 per row and column.

The next lemma characterizes the additional gauge operators introduced by including the bridge.

\begin{lemma}\label{lem:XX_system_gauges}
A complete, linearly independent basis of gauge checks of $U^B$ in the joint ancilla system consists of $|B|-1$ checks. Moreover, there is a basis in which each check acts on exactly two bridge qubits.
\end{lemma}
\begin{proof}

Define a matrix specifying connectivity between checks in sets $V_{j_1}^{(1)}$ and $V_{j_2}^{(2)}$ and qubits in sets $C^{(1)}_{j_1}$, $C^{(2)}_{j_2}$, and $B$.
\begin{equation}
G = \left(\begin{array}{ccc}
    F^{(1)\top} & 0 & S^{(1)} \\
    0 & F^{(2)\top} & S^{(2)}
\end{array}\right).
\end{equation}
A gauge check $Z(v_1\in C_{j_1}^{(1)})Z(v_2\in C_{j_2}^{(2)})Z(v_B\in B)$ in $U^{(1)}$, $U^{(2)}$, or $U^B$ is a solution to 
\begin{equation}\label{eq:XX_gauge_kernel}
G\left(\begin{array}{ccc}v_1 & v_2 & v_B\end{array}\right)^{\top}=0.
\end{equation}
Gauge checks in $U^{(1)}$ are solutions for which $v_2=v_B=0$ and checks in $U^{(2)}$ are solutions for which $v_1=v_B=0$. Clearly, $U^{(1)}$ and $U^{(2)}$ provide a basis for all solutions with $v_B=0$.

If $v_B\neq0$, we claim it has even weight if and only if $v_1,v_2$ exist so that $(v_1,v_2,v_B)$ solves Eq.~\eqref{eq:XX_gauge_kernel}. This is because, for all $v_1$, $Z(v_1\in C_{j_1}^{(1)})$ must flip an even number of checks in $V_{j_1}^{(1)}$ (Lemma~\ref{lem:supportlemma}), or, equivalently $F^{(1)\top} v_1$ is even weight. Likewise, $F^{(2)\top} v_2$ is even weight. Therefore, $(v_1,v_2,v_B)$ is a solution if and only if $S^{(1)}v_B$ and $S^{(2)}v_B$ are even weight, which implies $v_B$ is even weight since $S^{(1)}$ and $S^{(2)}$ both have exactly one 1 per column.

As a result, we can choose a basis of $U^B$ to consist of solutions in which $v_B=\hat e_i+\hat e_{i+1}$, for $i=0,1,\dots,|B|-2$, where $\hat e_i$ is the unit vector with a 1 in only the $i^{\text{th}}$ position.
\end{proof}

We note that the weight of this basis of gauge checks $U^B$ we specified depends on the logical operators being measured. 
In this work, we do not attempt to optimize these weights, and note that in the worst case they may not be LDPC.
We discuss the problem of minimizing weights of $U^B$, and the progress made by~\cite{swaroop2024universal}, in Section~\ref{sec:bridge}. 
In our study on the gross code, we found a basis of $U^B$ with weight at most $5$. 

With the help of these gauge checks, we can show that the merged code indeed has $k-1$ logical qubits.

\begin{lemma}
The merged code $\mathcal{G}_{XX}$ for measuring the product of two logical $X$ operators has $k-1$ logical qubits.
\end{lemma}
\begin{proof}
First, notice that the original code $\mc{G}$ with ancilla systems for measuring $\bar{X}_1$ and $\bar{X}_2$ would encode $k-2$ logical qubits by applying Theorem~\ref{thm:validXancilla} twice. Since $\mathcal{G}_{XX}$ adds $|B|$ qubits and $|B|-1$ independent checks (see Lemma~\ref{lem:XX_system_gauges}) to that code, it encodes $k-1$ logical qubits.
\end{proof}

In practice, the joint measurement can be carried out with a schedule similar to the $X$ measurement, Section~\ref{sec:schedule} -- starting from the original code, measure the stabilizers of the joint merged code for $R$ rounds, and finally measure the stabilizers of the original code again. The qubits in $V_i$ for even $i\ge2$, $C_j$ for $j\ge1$ odd, and in $B$ can be initialized in $\ket{0}$ at the zeroth round and measured out in the $Z$ basis at the final round. Note this dispenses with the idea of preparing the module in a stabilizer state before the protocol and preserving it in that state after the protocol in favor of simpler single-qubit preparations and final measurements. However, we can reach the same conclusions below if the modules for measuring $\bar{X}_1$ and $\bar{X}_2$ were kept intact as well. 

\begin{theorem}~\label{thm:XX_faultdistance_sameblock}
In the merged code $\mathcal{G}_{XX}$, denote the number of layers in the $\bar{X}_1$ ($\bar{X}_2$) subsystem as $L_1$ ($L_2$), and the boundary Cheeger constant of $F^{(1)^{\top}}$ ($F^{(2)^{\top}}$) as $\beta_1$ ($\beta_2$). 
If $\lceil L_1/2\rceil \ge 1/\beta_1$, $\lceil L_2/2\rceil \ge 1/\beta_2$, and 
we measure the merged code for $R\ge d$ rounds, then the joint measurement system connecting those ancilla systems via a bridge (see Figure~\ref{fig:XX_separate_codes}) has measurement and logical fault distances at least $d$.
\end{theorem}
\begin{proof}
Let $\mathcal{S}$ be the stabilizer group of the original code and each ancilla qubit, including bridge qubits, in $\ket{0}$, and $\mathcal{S}_{XX}$ be the stabilizer group of the joint merged code, Figure~\ref{fig:joint_ancilla}. 
Also let $\mathcal{L}$ be the group of logical operators of the unmeasured $k-1$ qubit logical space and $\mathcal{L}^*=\mathcal{L}\setminus\{I\}$. 
The versions of Lemmas~\ref{lem:subsystem_codes_and_fault_distance} and \ref{lem:CSS_subsystem_case} that apply to the joint code imply that the measurement fault distance is $\min(R,d_Z(\bar{Z}_1\mathcal{L}\mathcal{S}))$ and the logical fault distance is $\min(d_Z(\mathcal{L}^*\mathcal{S}),d_X(\mathcal{L}^*\mathcal{S}_{XX}))$. 
The original code $\mathcal{S}$ has code distance $d$, therefore $d_Z(\bar{Z}_1\mathcal{L}\mathcal{S})\ge d$ and $d_Z(\mathcal{L}^*\mathcal{S})\ge d$. To complete the proof, it remains only to show that the $X$-distance of the joint merged code is good, i.e.~$d_X(\mathcal{L}^*\mathcal{S}_{XX})\ge d$.

Let $P=X(t\in\mathcal{V})$ be an $X$-type logical operator of $\mathcal{G}$ that is inequivalent to $\bar{X}_1$ and $\bar{X}_2$ (note that $\bar{X}_1\equiv \bar{X}_2$). Consider the product of $P$ with a collection of $X$ checks, $\{v_i\in V_i^{(1)}\}_{\text{odd } i\in [L_1]}$ and $\{u_j\in V_j^{(2)}\}_{\text{odd } j\in [L_2]}$. Suppose $t\cap V_0^{(1)} = s_1$ and $t\cap V_0^{(2)} = s_2$. By the Expansion Lemma (Lemma~\ref{lem:expansion}), the weight of the new operator $P'$ is 
\begin{align}
    |P'| &= |X(t)\cdot \prod_{\substack{i=1\\i\text{ odd}}}^{L_1} \MH_X(v_i) \cdot \prod_{\substack{j=1\\j\text{ odd}}}^{L_2} \MH_X(u_j)| \\
    &\ge |t| - |s_1| - |s_2| + \min(|s_1|, |V_0^{(1)}| - |s_1|) + \min(|s_2|, |V_0^{(2)}| - |s_2|) + |v_{L_1} + u_{L_2}|.
\end{align}
If $|s_1|\ge |V_0^{(1)}| - |s_1|$ or $|s_2|\ge |V_0^{(2)}| - |s_2|$, then by Lemma~\ref{lem:Xopsprt} we have $|P'|\ge d$. Otherwise, we have
\begin{align}
    |P'| &\ge |t| - |s_1| - |s_2| + |V_0^{(1)}| - |s_1| + |V_0^{(2)}| - |s_2| \\
    &= |P\cdot \bar{X}_1\cdot \bar{X}_2| \ge d,
\end{align}
since the code distance of the original code is at least $d$. Now suppose $P\equiv \bar{X}_1 \equiv \bar{X}_2$, then $P\equiv X(B)$. We will show that the weight of $X(B)$ cannot be reduced below $d$ by stabilizers in $\mathcal{S}_{XX}$. 

Consider the product of $X(B)$ with a collection of $X$ checks in the ancilla system, $\{v_i\in V_i^{(1)}\}_{\text{odd } i\in [L_1]}$ and $\{u_j\in V_j^{(2)}\}_{\text{odd } j\in [L_2]}$. We again apply the expansion lemma, while also accounting for the support on $V_0^{(1)}$ and $V_0^{(2)}$. 
\begin{align}
    |P'| &= |X(B)\cdot \prod_{\substack{i=1\\i\text{ odd}}}^{L_1} \MH_X(v_i) \cdot \prod_{\substack{j=1\\j\text{ odd}}}^{L_2} \MH_X(u_j)|  \\
    &\ge |B + S^{(1)}v_{L_1} + S^{(2)}v_{L_2}| + |v_1 + v_m| + |u_1 + u_m| + |v_m + v_{L_1}| + |u_m + u_{L_2}| \notag \\
    & + \min(|v_m|, |\bar{X}_1| - |v_m|) + \min(|u_m|, |\bar{X}_2| - |u_m|) + |v_1| +|u_1|.
\end{align}
$P'$ now has support $v_1$ on $V_0^{(1)}$ and $u_1$ on $V_0^{(2)}$. If we multiply $P'$ by $X$-stabilizers in $\mathcal{S}$, then this support could be further reduced.
For any operator $O$, let $\mathcal{R}_{\mathcal{S}}(O)$ denote the minimum weight of equivalent representitives of $O$ under multiplication by stabilizers in $\mathcal{S}$. Then we have
\begin{align}
    \mathcal{R}_{\mathcal{S}}(P')
    &\ge |B + S^{(1)}v_{L_1} + S^{(2)}u_{L_2}| + |v_1 + v_m| + |u_1 + u_m| + |v_m + v_{L_1}| + |u_m + u_{L_2}| \notag \\
    & + \min(|v_m|, |\bar{X}_1| - |v_m|) + \min(|u_m|, |\bar{X}_2| - |u_m|) + \mathcal{R}_S(v_1u_1).
\end{align}
We do a case analysis based on $|v_m|, |u_m|$.
\begin{enumerate}
    \item If $|v_m| \le |\bar{X}_1| - |v_m|$ and $|u_m| \le |\bar{X}_2| - |u_m|$, we have
    \begin{align}
        \mathcal{R}_{\mathcal{S}}(P')
        &\ge |B + S^{(1)}v_{L_1} + S^{(2)}u_{L_2}| + |v_m + v_{L_1}| + |u_m + u_{L_2}| + |v_m| + |u_m|\\
        &\ge |B + S^{(1)}v_{L_1} + S^{(2)}u_{L_2}| + |v_{L_1}| + |u_{L_2}| \\
        &\ge |B| \ge d.
    \end{align}
    \item If $|v_m| \ge |\bar{X}_1| - |v_m|$ and $|u_m| \le |\bar{X}_2| - |u_m|$, we have
    \begin{align}
        \mathcal{R}_{\mathcal{S}}(P')
        &\ge |u_1 + u_m| + |u_m| + |\bar{X}_1| - |v_m| + \mathcal{R}_{\mathcal{S}}(v_1u_1)\\
        &\ge |u_1| + |\bar{X}_1| - |v_1| + \mathcal{R}_{\mathcal{S}}(v_1u_1).
    \end{align}
    Note that $\mathcal{R}_{\mathcal{S}}(v_1u_1) + |u_1| \ge \mathcal{R}_{\mathcal{S}}(v_1)$. Therefore $\mathcal{R}_{\mathcal{S}}(P')\ge |\bar{X}_1| - |v_1| + \mathcal{R}_{\mathcal{S}}(v_1)\ge \mathcal{R}_{\mathcal{S}}(\bar{X}_1) \ge d$.
    \item The case $|v_m| \le |\bar{X}_1| - |v_m|$ and $|u_m| \ge |\bar{X}_2| - |u_m|$ follows similarly.
    \item If $|v_m| \ge |\bar{X}_1| - |v_m|$ and $|u_m| \ge |\bar{X}_2| - |u_m|$, we have
    \begin{align}
        \mathcal{R}_{\mathcal{S}}(P')
        &\ge |\bar{X}_1| - |v_m| + |\bar{X}_2| - |u_m| + \mathcal{R}_{\mathcal{S}}(v_1u_1)\\
        &\ge |\bar{X}_1| - |v_1| + |\bar{X}_2| - |u_1| + \mathcal{R}_{\mathcal{S}}(v_1u_1) \\
        &\ge \mathcal{R}_{\mathcal{S}}(\bar{X}_1\cdot \bar{X}_2) \ge d. 
    \end{align}
\end{enumerate}
This proves our claim.
\end{proof}

\subsection{Gauge-fixed \texorpdfstring{$Y$}{Y} Ancilla System}\label{sec:ysystem}

Given the $X$ and $Z$ logical measurement schemes, in this section we present a $Y$ measurement scheme that combines the two. Our proposal differs from the one presented in~\cite{cohen2022low} because that proposal relies on using $2d-1$ layers to achieve fault-distance $d$. 

Again we make use of a bridge system like in the joint measurement case of the previous section. Here, however, we use the bridge to connect the $X$ and $Z$ systems together on the first layer, where they are also joined by a mixed-type check $q_1$. See Figure~\ref{fig:Y_ancilla} for the full construction.

For convenience in this section, we make the assumption that $\bar{X}_M$ and $\bar{Z}_M$ overlap on exactly one qubit. This makes the structure of gauge checks in the $Y$-system easier to reason about in proofs. However, we do not believe this is necessary and that all arguments in this section should generalize to cases of larger overlap (see also Remark~\ref{rmk:Yoverlap} below).

\begin{figure}
\centering
\includegraphics[width=0.75\textwidth]{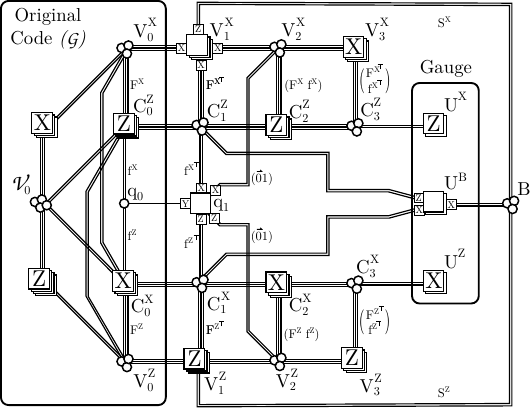}
\caption{\label{fig:Y_ancilla} The ancilla system for measuring $\bar{Y}$ where $\bar{X}$ and $\bar{Z}$ overlap on one qubit $q_0$. It largely consists of the ancilla systems for measuring $\bar{X}$ and $\bar{Z}$ separately, but with changes to the first layer, and the addition of the bridge qubits $B$ and more gauge checks $U^B$. The first layer introduces just one check $q_1$ connected to $q_0$ and has sets of checks $V_1^X$ and $V_1^Z$ which are each one smaller than sets in higher layers, e.g.~$V_3^X$ and $V_3^Z$.
}
\end{figure}

\begin{lemma}\label{lem:Y_system_gauges}
A complete, linearly independent basis for the gauge checks in $U^B$ in the $Y$ system consists of $|B|$ checks. Moreover, there is a choice of basis in which there is exactly one check acting on each bridge qubit in $B$.
\end{lemma}
\begin{proof}
We describe connectivity between checks $q_1,V_1^Z,V_1^X$ and qubits $C_1^X,C_1^Z,B$ by defining a matrix with rows corresponding to those checks and column corresponding to those qubits. Looking at Figure~\ref{fig:Y_ancilla}, this is
\begin{equation}
G=\left(\begin{array}{ccc}
    f^{Z^\top} & f^{X^\top} & 0\\
    F^{Z^\top} & 0  & S^Z \\
    0 & F^{X^\top}  & S^X
\end{array}\right).
\end{equation}
Let $v_X$ indicate a subset of $C_1^X$ qubits, $v_Z$ a subset of $C_1^Z$, and $v_B$ a subset of $B$. Gauge check $X(v_X\in C_1^X)Z(v_Z\in C_1^Z)X(v_B\in B)$ is defined as a solution to \begin{equation}\label{eq:gauge_kernel}
G\left(\begin{array}{ccc}v_X&v_Z&v_B\end{array}\right)^\top=0.
\end{equation}
In particular, gauge checks in $U^X$ are solutions for which $v_Z=v_B=0$ and gauge checks in $U^Z$ are solutions for which $v_X=v_B=0$.

We claim that solutions $(v_X,v_Z,0)$ with $v_B=0$ are the sum of solutions in $U^X$ and $U^Z$. Note that, for any $v_X$, applying $X(v_X\in C_1^X)$ flips an even number of checks from $V_1^Z\cup q_1$ (see Lemma~\ref{lem:supportlemma}). Therefore, because $F^{Z^\top} v_X^\top=0$, meaning $v_X$ flips no checks in $V_1^Z$, $v_X$ also does not flip $q_1$ or $ f^{Z^\top} v_X^\top =0$. This shows $(v_X,0,0)$ is a solution to Eq.~\eqref{eq:gauge_kernel}. Likewise, it can be argued $(0,v_Z,0)$ is a solution.

Therefore, beyond those corresponding to gauge checks in $U^X$ and $U^Z$, the only other independent solutions are those with $v_B\neq0$. In fact, for all $v_B\neq0$, there exist $v_X$ and $v_Z$ such that $(v_X,v_Z,v_B)$ solves Eq.~\eqref{eq:gauge_kernel}. We show this for a basis $v_B=\hat e_i$, the unit vector with a 1 in only the $i^{\text{th}}$ position. 

The $i^{\text{th}}$ bridge qubit is connected to two checks, one $i_X$ in $V_1^X$ and one $i_Z$ in $V_1^Z$. This is equivalent to saying $S^X\hat e_i^\top =\hat e_{i_X}^\top$ and $S^Z\hat e_i^\top=\hat e_{i_Z}^\top$. 

Using Lemma~\ref{lem:supportlemma}, we can find a vector $v_X$ such that $f^{Z^\top} v_X^\top=1$ and $F^{Z^\top} v_X^\top =\hat e_{i_X}^\top$ because $X(v_X\in C_1^X)$ flips exactly two checks, $i_X$ and $q_1$. Likewise, we find $v_Z$ such that $f^{X^\top} v_Z^\top=1$ and $F^{X^\top} v_Z^\top=\hat e_{i_Z}^\top$. It is then easy to verify that $(v_X,v_Z,v_B)$ solves Eq.~\eqref{eq:gauge_kernel}.
\end{proof}
\begin{remark}~\label{rmk:Yoverlap}
    It is important to note that if $\bar{X}_M$ and $\bar{Z}_M$ overlap on more than one qubit, the above proof may no longer hold, and the $Y$ merged code may have extra gauge qubits. These gauge checks take the form of $X(u\in C_1^X)Z(v\in C_1^Z)$, where $u$ and $v$ are chosen to ensure commutation with checks in $V_1^X$, $V_1^Z$, and the new collection of checks $Q_1$, $|Q_1|>1$, which replaces $q_1$. In particular, $Z(v\in C_1^Z)$ flips an even number of checks in $V_1^X\cup Q_1$ and $X(u\in C_1^X)$ flips an even number of checks in $V_1^Z\cup Q_1$. If the only checks they flip are both the same set of checks from $Q_1$, then $X(u\in C_1^X)Z(v\in C_1^Z)$ is a gauge operator. This also explains why there are no gauge checks of this form if $Q_1=\{q_1\}$, because it is impossible to flip an even number of checks from a set of size one. In general, these gauge checks should be fixed as well to ensure the system encodes just $k-1$ qubits.
\end{remark}

Knowing the structure of the gauge checks allows us to prove that the $Y$ system encodes exactly $k-1$ qubits, though we defer the proof to Appendix~\ref{apx:Yproofs}.

\begin{restatable}{theorem}{Ylogicalqubit}~\label{thm:Ylogicalqubit}
The $Y$-system encodes $k-1$ qubits.
\end{restatable}

To show that the $Y$-system has good fault-distance, we
apply Lemma~\ref{lem:subsystem_codes_and_fault_distance} with $\mathcal{U}=\langle\mathcal{S}_O,\mathcal{S}_Y\rangle$ and $\mathcal{L}=\langle \bar{X}_i,\bar{Z}_i,i=1,2,\dots,k-1\rangle$, where $\mathcal{S}_O$ is the stabilizer group of the original code (with the additional qubits starting in $\ket{+}$ for the bottom half of Figure~\ref{fig:Y_ancilla} and the bridge and in $\ket{0}$ for the top half), $\mathcal{S}_Y$ the stabilizer group of the merged code shown in Figure~\ref{fig:Y_ancilla}, and $\mathcal{L}$ a simultaneous $k-1$ qubit logical basis for each of the codes $\mathcal{S}_O'=\langle\mathcal{S}_O,\bar{Y}_M\rangle$, $\mathcal{S}_Y$ and $\mathcal{U}$. We have that the logical fault distance is at least $d(\mathcal{L}^*\mathcal{U})$ and the measurement fault distance is $\min(R, d(\bar{Z}_M\ML\MU), d(\bar{X}_M\ML\MU))$ for $R$ rounds of measurement.

\begin{restatable}{theorem}{Ydistance}~\label{thm:Y_distance}
If the $X$-system and $Z$-system individually have code distances at least $d$, then the $Y$-system has logical fault distance at least $d$.
\end{restatable}

\begin{proof}
Given how the qubits in the ancilla system are initialized, let 
\begin{align}
    \init = \langle Z(v)\text{ for singleton }v\in C^Z_{j\text{ odd}}, X(v)\text{ for singleton }v\in C^X_{j\text{ odd}}, X(b)\text{ for singleton }b\in B\rangle,
\end{align}
where we used $C^Z_{j\text{ odd}}$ to mean $\bigcup_{j\text{ odd}}C^Z_j$ and similarly for the other variants of this notation. 
Then we have
\begin{align}
    \MU = \langle \MS_Y, \init \rangle.
\end{align}
Suppose $\bar{P}$ is a nontrivial logical operator of $\mathcal{G}_Y$, the merged code of the $Y$ system. Let $\bar{P}=P_XP_Z$ be its decomposition into an $X$-type Pauli $P_X$ and $Z$-type Pauli $P_Z$.
Throughout we let $\mathcal{G}_X$ and $\mathcal{G}_Z$ denote the merged codes of the $X$- and $Z$-systems, respectively. When we write $P'_X=P_X|_{\mathcal{G}_X}$ or $P'_Z=P_Z|_{\mathcal{G}_Z}$, we mean to restrict those operators originally supported on $\mathcal{G}_Y$ to the qubits that exist in $\mathcal{G}_X$ or $\mathcal{G}_Z$, respectively. To show $d(\ML^*\MU)\ge d$, we will show that either $|P'_X|\ge d$ or $|P'_Z|\ge d$. This is sufficient since $P'_X, P'_Z$ are unaffected by operators in $\init$, which means if $|P'_X|\ge d$, then
\begin{align}
    d(\ML^*\MU) = \min_{\bar{P}\in \ML^*\MS_Y} \MR_{\init}(\bar{P})\ge \min_{\bar{P}\in \ML^*\MS_Y} \MR_{\init}(P'_X)\ge |P'_X|\ge d.
\end{align}
The first inequality is true since $\init$ is generated by singleton checks. 

We note that $P'_X$ is a logical operator on $\mathcal{G}_X$. It commutes with all $Z$ checks of $\mathcal{G}_X$ because $\bar{P}$ commutes with all the checks of $\mathcal{G}_Y$. Similarly, $P'_Z$ is a logical operator of $\mathcal{G}_Z$.
We next claim that either $P'_X$ is a nontrivial logical operator of $\mathcal{G}_X$ or $P'_Z$ is a nontrivial logical operator of $\mathcal{G}_Z$. This can be seen by contradiction. Assuming both $P'_X$ and $P'_Z$ are products of stabilizers in their respective codes, then $\bar{P}$ can also be multiplied by the respective stabilizers from $\mathcal{G}_Y$ to find an equivalent operator with support that is much restricted:
\begin{equation}
\bar{P}\equiv_{\mathcal{S}_Y}Z(v_Z\in C^Z_{j\text{ odd}}\cup V^X_{i\text{ even}})X(v_X\in C^X_{j\text{ odd}}\cup V^Z_{i\text{ even}})Z(b_Z\in B)X(b_X\in B),
\end{equation}
Furthermore, by multiplying with gauge checks $U^X,U^Z,U^B$ and checks from $C^Z_{j\text{ even}}$ and $C^X_{i\text{ even}}$ we can move the support of $\bar{P}\equiv_{\mathcal{S}_Y}Z(b'_Z\in B)X(b'_X\in B)$ entirely onto the bridge qubits $B$ (this goes like the proof of Theorem~\ref{thm:Zdistance} so we omit the details). However, in order to commute with all checks in $U^B$ (see Lemma~\ref{lem:Y_system_gauges}), $V_1^Z$, and $V_1^X$, it must be that $b'_Z=b'_X=0$ or $\bar{P}\equiv_{\mathcal{S}_Y}I$. This is impossible because then $\bar{P}$ would not be a nontrivial logical operator of $\mathcal{G}_Y$.

Now to conclude the proof, we simply note that $P'_X$ being nontrivial in $\mathcal{G}_X$ implies $|P'_X|\ge d$ by assumption. 
Likewise, if $P'_Z$ is nontrivial in $\mathcal{G}_Z$, we find $|P'_Z|\ge d$. This completes our proof.
\end{proof}

We defer the proof of the following theorem to Appendix~\ref{apx:Yproofs}, as it is similar to the proof of Theorem~\ref{thm:XX_faultdistance_sameblock}.
\begin{restatable}{theorem}{Yfaultdistance}~\label{thm:Y_faultdistance}
In the merged code $\mathcal{G}_{Y}$, denote the number of layers in the $\bar{X}$ ($\bar{Z}$) subsystem as $L_X$ ($L_Z$), and the boundary Cheeger constant of $(F^{X} f^{X})^\top$ ($(F^{Z} f^{Z})^\top$) as $\beta_X$ ($\beta_Z$). 
If $\lceil L_X/2\rceil \ge 1/\beta_X$, $\lceil L_Z/2\rceil \ge 1/\beta_Z$, and 
we measure the merged code for $R\ge d$ rounds, then the joint measurement system connecting those ancilla systems via a bridge (see Figure~\ref{fig:Y_ancilla}) has measurement fault distance at least $d$.
\end{restatable}

\begin{remark}
Without the bridge qubits, $\bar{Z}_M\mathcal{U}$ would clearly contain a low-weight element, namely $P=\bar{Z}_Mq_1\mathcal{H}_Z(V_{\text{j odd}}^Z)$, where $\mathcal{H}_Z(V_{\text{j odd}}^Z)$ is the product of all $Z$ checks in the sets $V_1, V_3,\dots$. Notice that $P$ would then have weight equal to the weight of the $X$-type part of $q_1$, which is a constant-weight check. The presence of the bridge qubits is therefore essential to eliminating this and similar cases that lead to low measurement fault distance.   
\end{remark}

\subsection{Bridging Codes}\label{sec:bridge}
In this section, we consider joint measurements of logical operators on different code blocks. 
A version of our fault-tolerance guarantee is as follows.
We defer the proof of this theorem to Appendix~\ref{apx:XXproofs}, as it is similar to our earlier joint system proofs. 

\begin{restatable}{theorem}{XXtwoblocks}~\label{thm:XX_faultdistance_twoblocks}
Suppose we begin with two separate codeblocks $\mathcal{G}_1$ and $\mathcal{G}_2$ and choose logical operators $\bar{X}_1$ and $\bar{X}_2$, respectively, from those codes such that the $X$-ancilla systems for measuring them have measurement and logical fault distances at least $d$, and both $\bar{X}_1$ and $\bar{X}_2$ are minimal weight representatives of their logical equivalence classes. If we measure the merged code for $R\ge d$ rounds, then the joint measurement system connecting those ancilla systems via a bridge (see Figure~\ref{fig:XX_separate_codes}) has measurement and logical fault distances at least $d$.
\end{restatable}
\begin{remark} 
We say that this is a version of our fault-tolerance guarantee because the assumption on $\bar{X}_1, \bar{X}_2$ needed for this theorem to hold is flexible. In the current statement, we ask that $\bar{X}_1$ and $\bar{X}_2$ are minimal weight representatives to account for edge cases. If we assume that their induced Tanner graphs are expanding and their ancilla systems have enough layers, then fault distance of the joint system can be proved as in Theorem~\ref{thm:XX_faultdistance_sameblock}. 
\end{remark}
\begin{figure}[ht]
    \centering
    \includegraphics[width=0.75\textwidth]{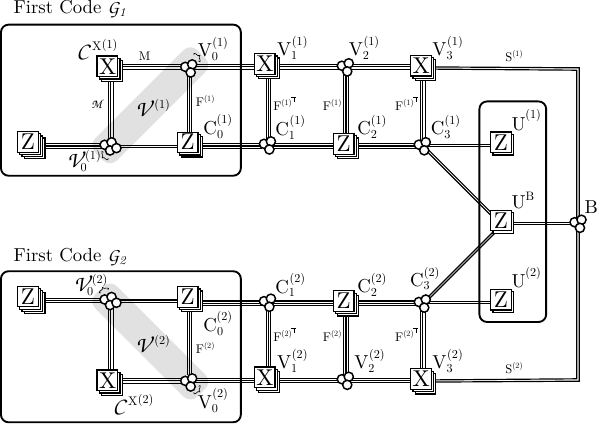}
    \caption{Redrawing the $XX$-system from Figure~\ref{fig:joint_ancilla} in the special case that $\bar{X}_1$ and $\bar{X}_2$ come from separate code blocks.}
    \label{fig:XX_separate_codes}
\end{figure}

This result enables us to connect different families of QLDPC codes together for a fault-tolerance architecture. 
Similar ideas based on the CKBB scheme have been considered and studied before~\cite{cohen2022low,bravyi2024highthreshold,xu2024constant}. However, these constructions mostly focused on connecting a QLDPC code with a surface code, and
lose fault-tolerant when we use fewer than $O(d)$ layers in the ancilla system. 

Our bridge construction overcomes this barrier and can be used as a code adapter: two arbitrary code blocks can be fault-tolerantly connected, enabling teleportation of logical information. 
This capability opens many possibilities for a QLDPC-based computing architecture. 
In particular, we can connect a QLDPC code with efficient logical Clifford operators to a magic state factory based on a vastly different code, achieving universality by combining individually optimized methods. 

On the other hand, our bridge construction comes with a noteworthy caveat.
As discussed in Section~\ref{sec:jointmeasurement}, the weight of the gauge checks $U^B$ depends on the logical operators being measured, which means in the worst case these checks may have high weight.
Let us discuss the choice of $U^B$ in more detail. 
Let $R$ be the repetition code on $B$ bits, from Lemma~\ref{lem:XX_system_gauges}, we see that the checks $U^B$ restricted to the bridge qubits generate the checks of a repetition code, namely $R^{\perp}$. 
Given the induced Tanner graph of a logical operator, we can choose the connecting matrices $S^{(1)}, S^{(2)}$ and a basis for $R^{\perp}$ to minimize the weight of $U^B$. 
This is captured in the following formulation.
\begin{problem}~\label{prob:bridge_weight}
Find a low weight $(|B|-1)\times |B|$ matrix $C_R$ with $\text{rowspace}(C_R)=R^{\perp}$, such that there exist low weight matrices $T^{(1)}, T^{(2)}$ and connecting matrices $S^{(1)}, S^{(2)}$ satisfying 
$T^{(1)}F^{(1)} = C_RS^{(1)^\top}$ and $C_RS^{(2)^\top} = T^{(2)}F^{(2)}$.
\end{problem}
In the case where $F^{(1)} = F^{(2)}$, we can simply take $C_R$ to be a submatrix of $F^{(1)}$.
In average cases, we expect reasonable choices of $C_R, S^{(1)}, S^{(2)}$ to exist. For arbitrary $F^{(1)}$ and $F^{(2)}$, however, such choices may not exist, and our bridge cannot connect these logical operators without breaking the LDPC property. 

This limitation is fully resolved in~\cite{swaroop2024universal}, where it was shown that when $F^{(1)}$, $F^{(2)}$ are simple expander graphs (such as those constructed in~\cite{williamson2024low, ide2024fault}), one can always efficiently solve Problem~\ref{prob:bridge_weight} to guarantee a LDPC bridge construction. 
This means the adapters constructed in~\cite{swaroop2024universal} are truly universal, in the sense that they can connect arbitrary QLDPC code families while maintaining fault-tolerance and LDPC.
We refer readers to~\cite{swaroop2024universal} for further studies and applications of the idea of code adapting.

\section{Case Study: Gross Code}\label{sec:gross_code}

In this section we apply the techniques from the prior section to construct an ancilla system tailored for the $[[144,12,12]]$ code presented by \cite{bravyi2024highthreshold}, known as the gross code. We present an ancilla system with an additional 103 physical qubits, 54 data and 49 check, that are capable of implementing various logical measurements. Together with the automorphism gates that are readily implementable in the gross code, these measurements enable synthesis of all logical Clifford gates and all logical Pauli measurements of eleven of the twelve logical qubits in the code. Prior work in \cite{bravyi2024highthreshold} gave an ancilla system that requires 1380 additional qubits and is merely capable of logical data transfer into and out of the code.

We begin with a review of the construction of the gross code, and Bivariate Bicycle codes more generally. The Tanner graph of the code features vertices of four types: $X$ and $Z$ check vertices, and $L$ and $R$ data vertices. For some $l,m \in \mathbb{Z}^+$, there are $lm$ vertices in each of these categories. The vertices are labeled by elements of an abelian group $\mathcal{M} := \langle x,y | x^l, x^m\rangle \cong \mathbb{Z}_{l} \times \mathbb{Z}_m$ which can be viewed as gridpoints on a $l \times m$ torus. Elements of the $X$ set can be written as $X(x^iy^j)$ for some $0 \leq i < l$ and $0 \leq j < m$, and similarly for the $Z,L,$ and $R$ sets. This makes for $2lm$ check qubits and $2nm$ data qubits. For the gross code, $l=12$ and $m=6$.

The edges of the Tanner graph are specified by two polynomials $A,B \in \mathbb{F}_2[\mathcal{M}]$ with three terms each. For example, for the gross code we have:
\begin{align}
A := x^3 + y + y^2, \hspace{2cm} B = y^3 + x + x^2.
\end{align}
Since the polynomials are over $\mathbb{F}_2$, they are also naturally interpreted as subsets of $\mathcal{M}$.
Then the Tanner graph is given as follows. For some $m \in \mathcal{M}$, the vertex $X(m)$ is connected to the vertices $L(Am) := \{L(m')$ for $m'$ in $Am\}$, and to $R(Bm) := \{R(m'')$ for $m''$ in $Bm\}$. Similarly, if $A^\top$ is the polynomial consisting of the inverses of the monomials of $A$, that is, $A^{T} = x^{-3} + y^{-1} + y^{-2}$, then the vertex $Z(m)$ is connected to $L(B^\top m) := \{L(m')$ for $m'$ in $B^\top m\}$, and to $R(A^\top m) := \{R(m'')$ for $m''$ in $A^\top m\}$.

We verify that the checks commute. Lemma~2 of \cite{bravyi2024highthreshold} shows that an $X$ operator on $L(p)\cup R(q)$ anticommutes with a $Z$ operator on $L(\bar p)\cup R(\bar q)$ if and only if $p\bar p^\top  +q\bar q^\top $ contains the monomial $1\in \mathcal{M}$. The $X(m)$ check is supported on $L(Am)\cup R(Bm)$, and the $Z(m')$ check is supported on $L(B^\top m')\cup R(A^\top m')$.  We observe that $(Am)(B^\top m')^\top  +(Bm)(A^\top m')^\top  = m(m')^\top  ( AB + AB) = 0$ which never contains 1, so the stabilizers must always commute.

Logical Pauli operators on the gross code feature two symmetries of the type studied by \cite{breuckmann2024foldtransversal}: automorphisms and a $ZX$-duality. Automorphisms correspond to shifts of a logical operator by a monomial $w \in \mathcal{M}$. Suppose a logical $X$ operator supported on $L(p)\cup R(q)$ commutes with all $Z$ stabilizers, and hence obeys $pB + qA = 0$. Then we have $(wp)B + (wq)A = 0$ for all $w \in \mathcal{M}$, so  $L(wp)\cup R(wq)$ is also the support of a logical $X$ operator. The same argument holds for $Z$ operators. The permutation of the physical qubits corresponding to multiplying by a monomial has a fault tolerant quantum circuit.

The $ZX$-duality is a permutation that converts the support of an $X$ check to that of a $Z$ check, and vice versa. Again, if a logical $X$ operator is supported on  $L(p)\cup R(q)$, then $pB + qA = 0$. This implies $B^\top  p^\top  + A^\top  q^\top  = 0$, so a $Z$ operator supported on  $L(q^\top )\cup R(p^\top )$ commutes with the $X$ stabilizer. \cite{bravyi2024highthreshold} constructed a high-depth fault-tolerant circuit implementing this permutation, but this scheme does not seem practical. Nonetheless we find the $ZX$-duality useful in the construction of our ancilla system. 

For every qubit defined by an anticommuting $\bar X,\bar Z$ pair, there exists a `primed' qubit defined by applying the $ZX$-duality to these operators, obtaining $\bar Z',\bar X'$. Hence, the twelve logical qubits of the gross code decompose into a primed and unprimed block, each with six logical qubits, on which the automorphism gates act simultaneously and identically.

Having reviewed the gross code, we proceed with the description of the ancilla system. 

\newpage 
\subsection{Mono-layer ancilla system}~\label{sec:monolayer}

\begin{figure}[t]
    \centering
    \includegraphics*[width=0.7\textwidth]{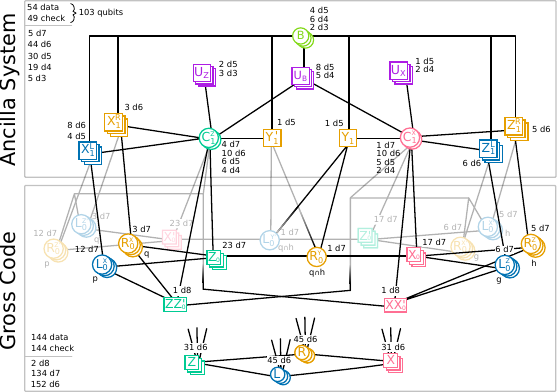} 
    \caption{\label{fig:gross_ancilla_counts} Qubit and degree counts of a gross code equipped with an ancilla system. This ancilla system implements eight different logical measurements by realising different subgraphs as Tanner graphs. Circular qubits are data, and squares are checks. Gauge checks are in purple, the bridge is green, and the remaining ancilla system qubits are colored according to their partners in the gross code: X is red, Z is teal, R is oRange, and L is bLue. The ancilla system is connected to two pairs of logical operators that are ZX-duals of one another: the $\bar X, \bar Z$ qubits are opaque and the $\bar X', \bar Z'$ qubits are transparent. Their Tanner-subgraphs overlap on the $ZZ'_0$ and $XX'_0$ qubits. Note that checks are labeled $X$ or $Z$ type according to how they are used in $X$ or $Z$ systems individually, but in the joint measurement and $Y$ systems that employ the bridge, they are no longer mono-type checks (see Figs.~\ref{fig:joint_ancilla} and \ref{fig:Y_ancilla}). An annotation `3 d7' means that there are three degree-seven qubits.}
\end{figure}

A sketch of the system is shown in Figure~\ref{fig:gross_ancilla_counts}, displaying an overview of the layout of the ancilla system as well as its resource requirements. The construction is based on a pair of logical $X$ and $Z$ operators where a mono-layer measurement scheme preserves code distance. In particular, we have an $\bar X$ operator supported on $L(p)\cup R(q)$, and a $\bar Z$ operator supported on $L(r)\cup R(s)$ where:
\begin{align}
    p &= y + xy^5 + x^2y^2 + x^2y^4 + x^3y + x^3y^2 + x^4y + x^4y^2 + x^9y^3 + x^{10} + x^{11} + x^{11}y^3,\\
    q &= y^2 + xy + xy^5 + x^6,\\
    r &= x^2y^3 + x^2y^5 + x^3 + x^3y^3 + x^3y^4 + x^3y^5,\\
    s &= xy^3 + xy^5 + x^2 + x^2y^4 + x^3y^2 + x^3y^4.
\end{align}

We verify numerically using CPLEX that the distance of the merged codes of the $X$ and $Z$ systems individually are both 12, preserving the distance of the original code.  These operators overlap only on $R(xy^5)$, so they anticommute.

Several additional ancilla systems can be constructed from these $X$ and $Z$ systems. We present a single graph of connected qubits, such that a supercomputing architecture based on this graph can realize all of these systems either with a direct correspondence of qubits to Tanner graph vertices, or with a few flag qubits for $R(xy^5)$ only. 
\begin{itemize} 
\item Following the construction in section~\ref{sec:ysystem}, we can build a system that measures $\bar Y := i\bar X \bar Z$. This introduces eleven bridge qubits that connect the two ancilla systems - one fewer than the distance since they already share $R(xy^5)$. These are shown in Figure~\ref{fig:gross_bridge}. In the resulting Tanner graph there is one qubit in the ancilla system connected to $R(xy^5)$ that connects the $X$ and $Z$ systems together. For reasons we explain shortly, we will split this vertex into three logical qubits connected in a line. These three qubits can act as a single check by initializing all but the middle qubit in a Bell state, and finishing with a Bell measurement. This yields a few flag measurements \cite{DiVincenzo_2007}.

\item Due to the $ZX$-duality, the operator $\bar Z$ has a dual $X$ logical operator supported on $L(s^\top )\cup R(r^\top )$. To arrive at disjoint operators and also to minimize overlap of check qubits, we subsequently apply an automorphism $w = x^{10}y^{5}$, and arrive at an operator $\bar X'$ supported on $L(ws^\top )\cup R(wr^\top )$. Similarly, we construct an operator $\bar Z'$ supported on $L(wq^\top )\cup R(wp^\top )$. Since $\bar X'$ was constructed from $\bar Z$ using geometric symmetries of the code, the subgraph of the Tanner graph used to construct the ancilla system is identical for these operators up to a switch of $X$ and $Z$. By connecting the $\bar Z$ ancilla system to both of the $\bar Z$ and $\bar X'$ subgraphs, a single system can measure both $\bar Z$ and $\bar X'$ by considering different subsets of the connections as the Tanner graph. Similarly, the $\bar X$ system can also measure $\bar Z'$. We can also measure $\bar Y' := i\bar X' \bar Z'$.
\item Following section~\ref{sec:jointmeasurement}, we construct measurements of the joint operators $\bar X\bar X'$ and $\bar Z\bar Z'$. This explains the need to split the $R(xy^5)$ qubit: since $\bar X$ and $\bar X'$ are disjoint, the ancilla qubits connected to $R(xy^5)$ and its dual $L(w(xy^5)^\top )$ need to be separate. Now, twelve bridge qubits are required to ensure preservation of code distance. As anticipated above, we place this bridge qubit between the two qubits connected to $R(xy^5)$, so that these qubits can act as a single check together when we are measuring $\bar Y$ or $\bar Y'$.
\end{itemize}
Hence the same ancilla system can measure any of the operators $\bar X, \bar Y, \bar Z, \bar X', \bar Y', \bar Z', \bar X\bar X'$, and $\bar Z \bar Z'$. In section~\ref{sec:logicalcliffords} we show that these measurements and the automorphisms together can synthesize all Clifford gates on eleven of the twelve qubits.

The theorems proved in Section~3 generally do not provide enough evidence to conclude that the $\bar Y, \bar Y', \bar X\bar X',$ and $\bar Z\bar Z'$ merged codes preserve code distance. Nevertheless we find numerically that this is the case using CPLEX.

In the original gross code all vertices in the Tanner graph have degree six. Looking at Figure~\ref{fig:gross_ancilla_counts}, the degrees of the qubits in the ancilla system are mostly equal or less than this since the subgraphs of the gross code corresponding to the logical operators do not include connections to one type of check. However, the additional gauge checks introduced by the bridge system result in some degree seven data qubits in the ancilla. Many qubits in the gross code must be connected to the ancilla system. Through careful choice of the automorphism shift $w$, it was achieved that all but two of these qubits connect to the ancilla system only once, so they have degree seven. The qubits $X(x^{11}x^5)$ and $Z(x^{11})$ are the only degree eight qubits in the construction. If it had not been split via the bridge, the shared $R(xy^5)$ qubit in the ancilla system would have had degree eight as well.
 
Several gauge checks were introduced, either supported in the $X$ or $Z$ systems alone, or to deal with loops introduced by the bridge. The gauge checks are specified in Figures~\ref{fig:gross_bridge} and~\ref{fig:gross_ancilla_spec}. For the full $\bar X \bar X'$ and $\bar Z \bar Z'$ measurement operations all of these gauge checks are necessary. But for the other measurements, it turns out that several of the gauge checks are redundant. In the $\bar X,$ and $\bar Z'$ systems one \emph{only} needs gauge checks 0, 1, and 2. In the $\bar Z $ and $ \bar X'$ systems one only needs gauge check $6$. In the $\bar Y$ and $\bar Y'$ systems, gauge checks 10 and 17 must have their support from bridge qubit 0 removed in order to be well-defined. Then, gauge checks 6, 12, 16, 17, 18, and 20 are redundant.

\begin{figure}[t]
    \centering
    \includegraphics*[width=0.8\textwidth]{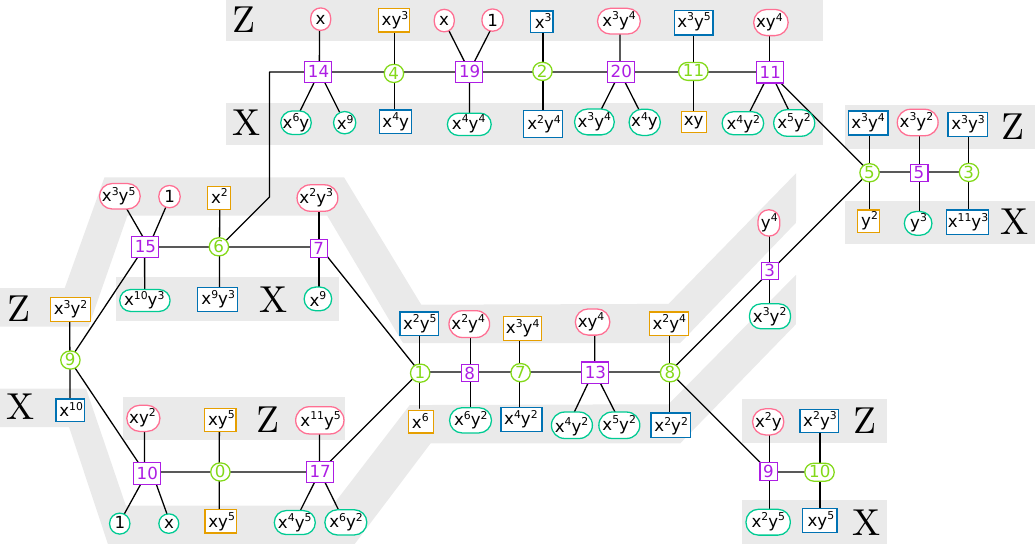}
    \caption{\label{fig:gross_bridge} Subgraph of the Tanner graph corresponding to the bridge and its corresponding gauge checks. The bridge system is necessary to ensure that $\bar X\bar X'$, $\bar Z \bar Z'$, $\bar Y$, and $\bar Y'$ measurements are fault tolerant. Each green bridge qubit connects a pair of checks in the $X$ and $Z$ systems. This introduces several loops, which give rise to gauge degrees of freedom, which are removed through additional checks in purple.}
\end{figure}

\begin{figure}[p]
    \centering
    \includegraphics*[height=0.8\textheight]{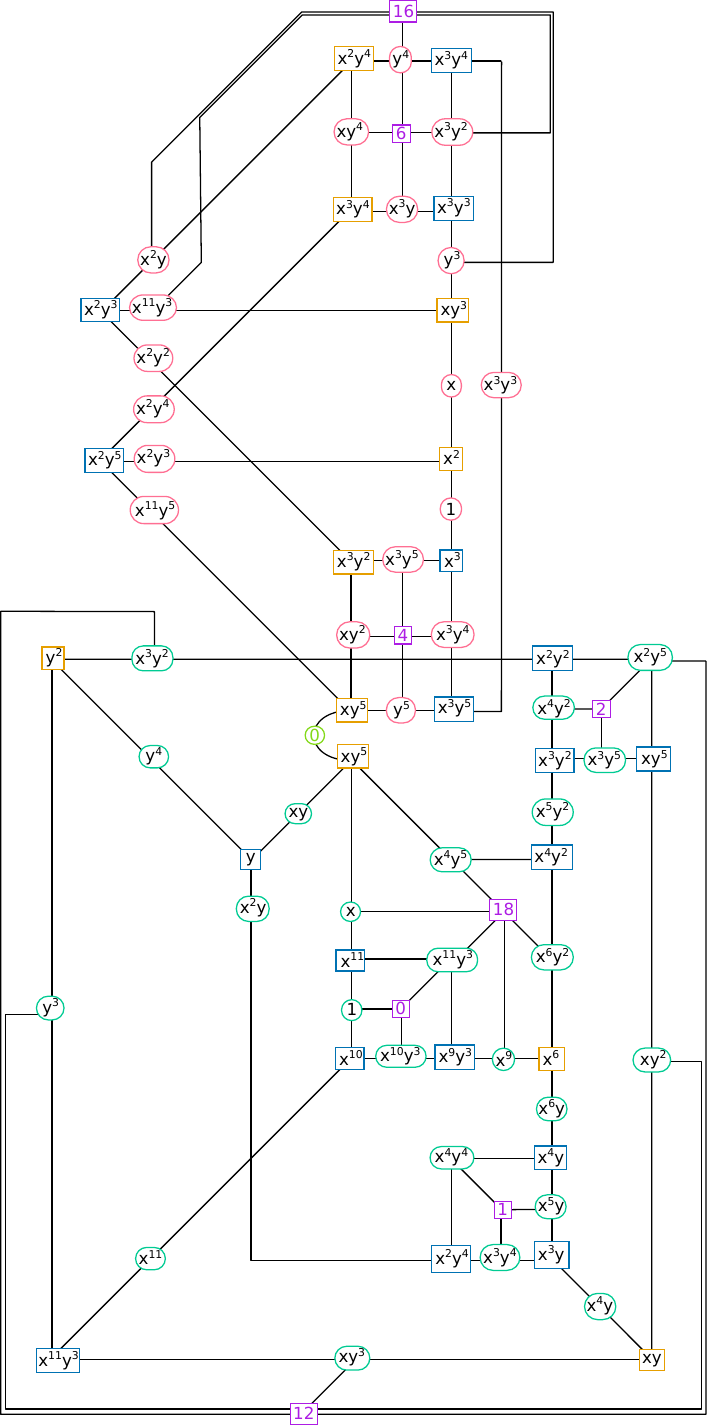}
    \caption{\label{fig:gross_ancilla_spec} Subgraph of the Tanner graph of the ancilla system ignoring the bridge, except for bridge qubit 0 which connects the qubits corresponding to the $R(xy^5)$ where $\bar X$ and $\bar Z$ overlap. Qubits are labeled by the qubit they are connected to in the gross code, following the coloring scheme described in Figure~\ref{fig:gross_ancilla_counts}. Observe how the role of check and data qubits are reversed relative to the gross code qubits.}
\end{figure}
 
\subsection{Numerical simulations}

We perform circuit-level noise simulations of the measurement protocols for the $X$ and $Z$ ancilla systems. We find that the syndrome circuits for the corresponding merged codes preserve the circuit-level distance of the gross code's circuit, which is ten. Due to the larger circuit size, the resulting logical error rates per syndrome cycle are about a factor of ten larger than a syndrome cycle with no logical gate. These results suggest that measurements implemented via ancilla systems are a practical avenue for implementing fault-tolerant logical gates in an architecture based on the gross code.

The syndrome circuit is derived from an edge-coloring of the Tanner subgraph shown in Figure~\ref{fig:gross_ancilla_spec}. The full circuit is shown in Figure~\ref{fig:syndromecircuits}. The coloring yields a quantum circuit with three rounds marked $R,G,$ and $B$ in the figure, and one additional `early' round, that together cover all edges in the subgraph supported on the ancilla system. The syndrome circuit of the gross code has eight rounds of CNOTs or measurements and features a natural gap where the data qubits are idling. This gap is leveraged for implementing the edges between the ancilla system and the gross code, so the merged syndrome circuit has the same depth per cycle.

Recall that for the $X$ and $Z$ systems, we can ignore the gauge checks 4, 12, and 18 since they are redundant, as well as the bridge qubit 0. Hence, the $X$ and $Z$ Tanner graphs in Figure~\ref{fig:gross_ancilla_spec} have vertices with degree at most three, except for the gauge check 6 which has degree four. We find that, if we leave out one of check 6's edges, the edges of the graph are three-colorable, yielding a depth-three quantum circuit that includes all required CNOT gates. If all the checks in the circuit were the same type of check (all $X$ or all $Z$), then the order of these CNOTs would not matter. However, since gauge checks are of a different type, the original coloring circuit results in some incorrect computation of the stabilizers near the gauge checks. We find that this is easily fixed by introducing an additional `early' round of CNOTs, and moving some of the problematic CNOTs into this round. This is also a natural spot for the fourth edge of the gauge check 6. See Figure~\ref{fig:coloring_circuits} for a derivation and specification.

We perform three experiments on the syndrome circuits. First, we perform a memory experiment to emperically determine the circuit-level distance and to compare to the code's idle cycle. Second, we determine how many syndrome cycles are suitable for a logical measurement by comparing the logical error rate to the measurement error rate. Both of these experiments rely on the modular decoding approach described in Section~\ref{sec:modular_decoder}. In the third experiment we assess the benefits of the modular decoding approach over a simpler implementation relying on BPOSD alone.  Our software is implemented in python and leverages the \textsf{stim} \cite{gidney2021stim}, \textsf{pymatching} \cite{higgott2023sparse}, and \textsf{bp\_osd} \cite{roffe_decoding_2020, Roffe_LDPC_Python_tools_2022} libraries. 

The memory experiment is shown in Figure~\ref{fig:numericalexperiments}a). We compare the $X$ and $Z$ systems as well as the operation of the gross code syndrome cycle on its own. We find that our simulations of the gross code match the fits from \cite{bravyi2024highthreshold} within statistical error, confirming the accuracy of our implementation and that our error model is the same as the one described in their Section~2. The logical error rates of the $X$ and $Z$ systems are within a factor of 10x and 5x larger than that of the code on its own. We also conclude that the circuit-level distance of all the syndrome circuits is ten, both from the similarity of the slopes of the curves, and also BPOSD's inability to identify undetectable logical errors of weight less than ten.

In Figure~\ref{fig:numericalexperiments}b) we compare the error rate of the logical measurement to the logical error rate - the probability of a logical error on any of the qubits not being measured. As the number of rounds increases there are more opportunities for high-weight errors and the logical error rate increases. But more rounds also provide more measurements of the stabilizers composing $\bar X$, yielding a more reliable logical measurement outcome. We find that these probabilities are about the same at seven rounds, although this crossover point will depend on the physical error rate. In this experiment the physical error rate was 0.001.

All experiments above rely on the modular decoding method from Section~\ref{sec:modular_decoder}. 
In Figure~\ref{fig:matching} we compare this modular decoder to a purely BPOSD decoding approach in the critical region of about seven syndrome cycles at a physical error rate of 0.001. Here we find that the performance of the two decoders is very similar up to statistical error, with a slight advantage towards a purely BPOSD approach. However we also see that the matching decoder is over 100x faster than the BPOSD decoders, and the BPOSD decoder for the merged code is about 10x slower than that of the gross code alone. This shows that the modular decoding approach can provide a tangible improvement in decoding time.

\begin{figure}[t]
    \centering
    \subfloat[a) Comparison of logical error rates of the $X$ and $Z$ ancilla systems and the idle gross code in the circuit-level noise model for 12 syndrome cycles. Logical error rates are normalized by the number of syndrome rounds. A logical error is any error on the logical qubits not being measured.]{\includegraphics*[width=0.45\textwidth]{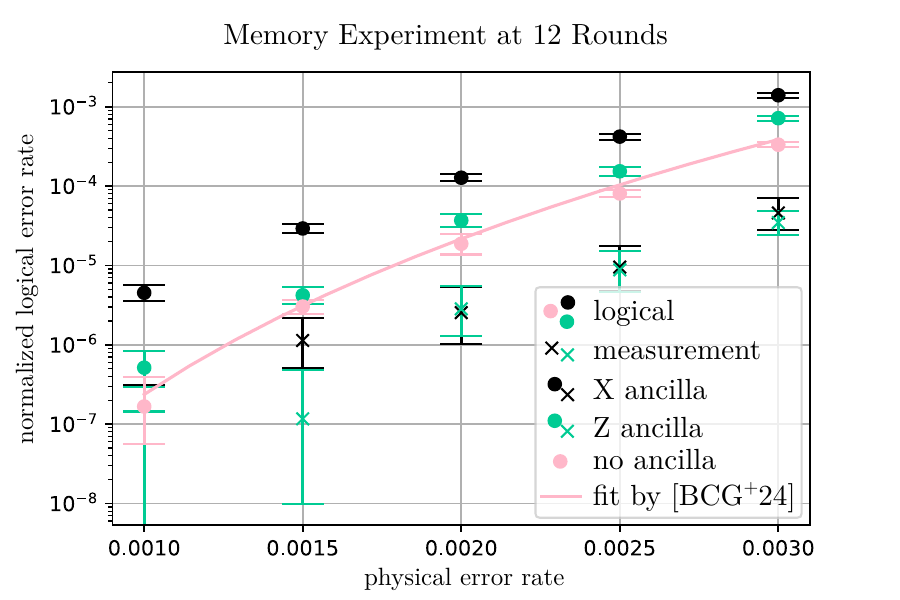}}\qquad
    \subfloat[b) Comparison of the logical measurement error to the probability of a logical error on the remaining qubits. Here we consider the $X$ system at a physical error rate of $p = 0.001$. At this physical error rate the quantities are balanced at about seven rounds.]{\includegraphics*[width=0.45\textwidth]{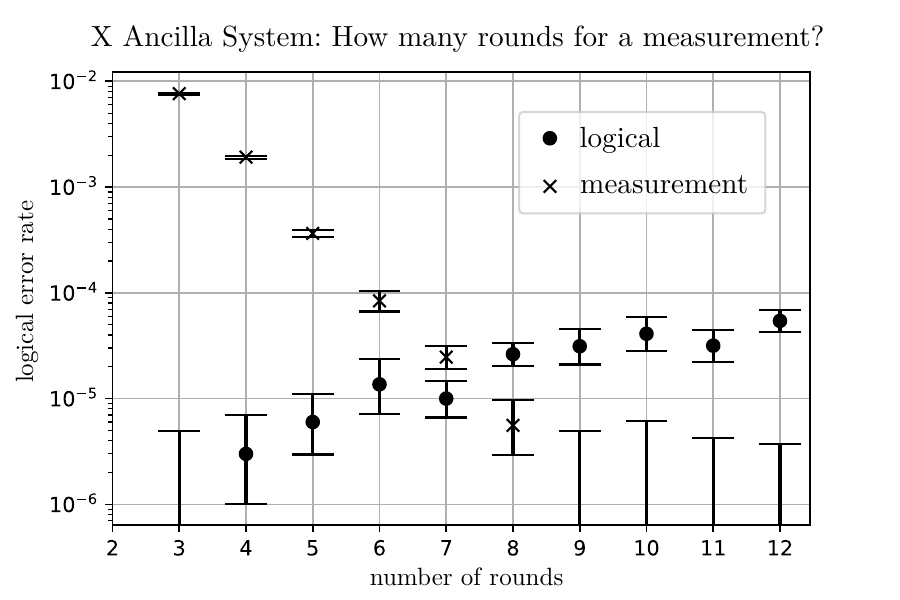}}\\
    \caption{\label{fig:numericalexperiments} Circuit-level noise model benchmarks of the merged syndrome circuits shown in Figure~\ref{fig:syndromecircuits}. The noise model is the same as that of \cite{bravyi2024highthreshold}. The error bars denote confidence intervals of $1\sigma.$}
\end{figure}

\begin{figure}[t]
    \centering\hspace{-5mm}
    \subfloat[a) Measurement and logical error rates of two different decoding strategies. We see that these are very similar up to statistical error, with a purely BPOSD strategy performing slightly better.]{\includegraphics*[width=0.33\textwidth]{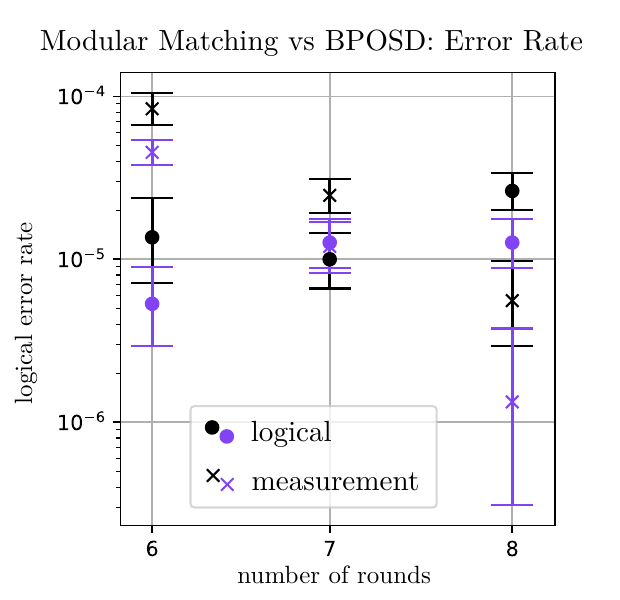}}
    \subfloat[\parbox{0.4\textwidth}{b) Decoding compute times for the decoding strategies, leveraging the \textsf{bp\_osd} and \textsf{pymatching} python packages. We see a significant improvement for the modular approach.}]{\includegraphics*[width=0.55\textwidth]{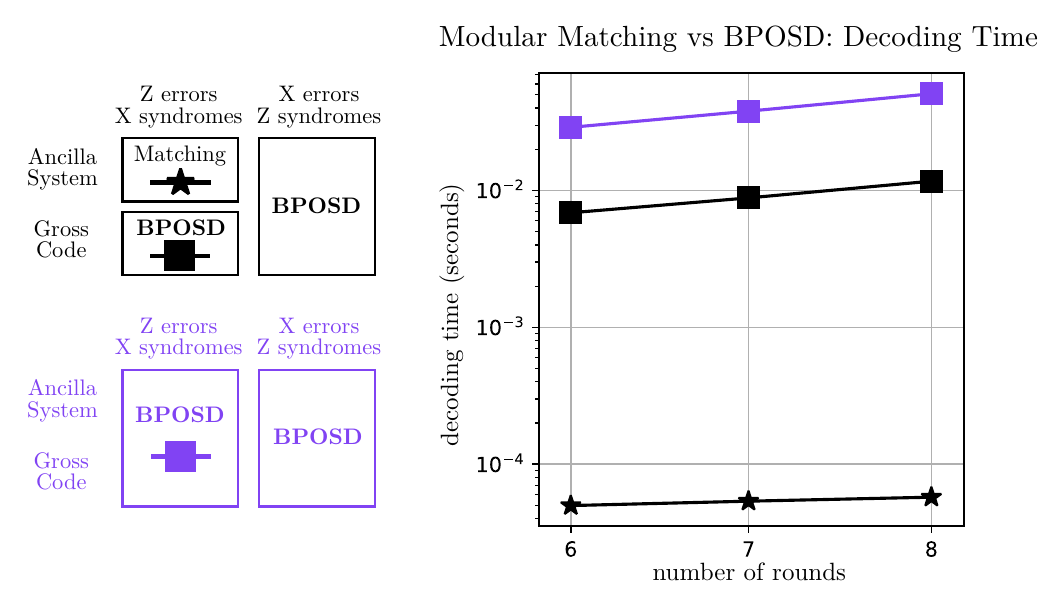}}\\
    \caption{\label{fig:matching} Comparison of the modular decoding strategy to a decoder based purely on BPOSD. Here we consider the $X$ ancilla system, a number of cycles where measurement and logical errors are similar, and a physical error rate $p=0.001$. The error bars denote confidence intervals of $1\sigma.$}
\end{figure}

\subsection{Implementing logical Clifford gates}\label{sec:logicalcliffords}

Together with the automorphisms the logical measurements implemented by the ancilla system suffice to implement all logical Clifford gates on eleven of the twelve logical qubits. 

Automorphism gates are permutations of the physical qubits, and hence take $X$ operators to $X$ operators and $Z$ operators to $Z$ operators. Since $\bar X$ is supported on $L(p)\cup R(q)$, we may define $\bar X_\alpha$ to be supported on $L(\alpha p)\cup R(\alpha q)$ for $\alpha \in \mathcal{M}$. A measurement of $\bar X_\alpha$ can be implemented by conjugating an $\bar X$ measurement with the automorphism gate $\alpha$. We further let $\bar Z_\alpha$ be supported on $L(\alpha s^\top )\cup R(\alpha r^\top )$, and similarly for the dual operators $\bar X'_\alpha$ and $\bar Z'_\alpha$. For the joint measurements $\bar Y, \bar Y', \bar X \bar X'$ and $\bar Z\bar Z'$ we cannot apply automorphisms to each operator in the product separately. We only have access to $\bar Y_\alpha := i\bar X_\alpha\bar Z_\alpha,$ and $ \bar Y'_\alpha := i\bar X'_\alpha\bar Z'_\alpha,$ as well as $ \bar X_\alpha \bar X'_\alpha,$ and $ \bar Z_\alpha \bar Z'_\alpha$. Note that there is a set of $\alpha$ such that $\bar X_\alpha, \bar Z_\alpha, \bar X'_\alpha, \bar Z'_\alpha$ form a symplectic basis for the logical Hilbert space. For such a basis can individually initialize and read all logical qubits.

To transform our measurement-based gateset into a unitary gate set, we make use of the following well-known identity (e.g., \cite{Litinski_2019} Figure~11b). 
\begin{align}
\includegraphics[]{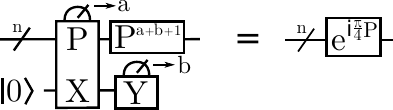}\label{eqn:rotationsynth}
\end{align}
Here $P$ is any $n$-qubit Pauli matrix, and for our purposes we are interested in $n=11$. A Clifford gate of the form $e^{i\frac{\pi}{4}P}$ is called a Pauli rotation.

Any measurement-based unitary synthesis strategy must sacrifice one qubit to act as a `pivot' - here the bottom qubit in the circuit. The pivot ensures that the measurement outcome is never deterministic, and hence cannot decohere the quantum information on the data qubits. Measurement-based operations repeatedly entangle and disentangle the data qubits with the pivot qubit. The measurement results generate Pauli corrections which are cheap to apply, either physically or by tracking the Pauli frame.

There are some variations of (\ref{eqn:rotationsynth}). The $\ket{0}$ initialization can be performed through a $Z$ measurement, resulting in a $Z \to X \to Y$ sequence of basis changes. By inserting single-qubit Clifford gates and their inverses, this sequence can be changed to any ordering of the $X,Y,Z$ bases. 

We select the pivot qubit to be the logical qubit with operators $\bar X_1$ and $\bar Z_1$. This way we can implement measurements of $\bar X_1, \bar Z_1$, and $\bar Y_1$ natively using the ancilla system. Hence, \emph{any} implementable measurement can be transformed into a Pauli rotation so long as it has support on both the pivot qubit and at least one other logical qubit. No matter if the support on the pivot is $X,Y,$ or $Z$, we can select the initialization and final measurement in (\ref{eqn:rotationsynth}) to be the other two. 

The automorphism group is given by $\mathcal{M}$, with $|\mathcal{M}| = 72.$ However, we find that logical operators are invariant under the automorphism $x^6$: $\bar X_1 \bar X_{x^6}$ and $\bar Z_1 \bar Z_{x^6}$ are in the stabilizer. Hence if we consider the group of automorphism unitaries modulo the stabilizer, this group has order 36.

With 36 (logical) automorphisms and eight logical Pauli measurements from the ancilla system, we can implement 288 logical measurements. But any conjugated measurement of $\bar X', \bar Y', \bar Z'$ would have support on the primed qubits only, and hence cannot have support on the pivot since it is in the unprimed block. Therefore, only $(8-3)\times 36 = 180$ measurements are useful for rotation synthesis. Of these, 95 have support on both the pivot and the remaining qubits. 

We find that these 95 rotations suffice to generate the Clifford group. This is verified numerically via brute-force search. We leverage the following identity:
\begin{align}
e^{i\frac{\pi}{4}P}e^{i\frac{\pi}{4}Q}e^{-i\frac{\pi}{4}P} = \bigg\{ \begin{array}{l}e^{\frac{\pi}{4}PQ} \text{ if } PQ = -QP\\ e^{i \frac{\pi}{4} P} \text{ otherwise}\end{array},
\end{align}
noting that $e^{\frac{\pi}{4}PQ}$ is unitary since $PQ$ is anti-Hermitian. We maintain a list of implementable Pauli rotations, and grow this list by conjugating them by rotations from the 95 natively implementable gates. When a new implementable rotation is found it is added to the list. When a shorter implementation of a prior rotation is found the recorded implementation is updated with the shorter one. We terminate the search when all implementations are found. Since the Clifford group is generated by Pauli rotations and Pauli matrices \cite{Callan1976TheGO} this suffices to show that all Clifford gates can be implemented.

The search algorithm above finds partially optimized but not necessarily optimal implementations of all $4^{11}$ Pauli rotations. The overhead of this synthesis method is presented in Figure~\ref{fig:implementations}a). For Pauli-based computation \cite{Bravyi_2016, Litinski_2019, Chamberland_2022} where universal quantum computation is implemented via magic state preparation and multi-qubit Pauli measurements, implementation of arbitrary Pauli measurements may be more interesting than Pauli rotations. Figure~\ref{fig:implementations}b) shows the results of a similar search, where the 288 native logical measurements are conjugated by Pauli rotations until all measurements are found.

If we are implementing a sequence of gates using (\ref{eqn:rotationsynth}) then we require three measurements per gate. However, depending on the support of the multi-qubit measurement the final measurement of one rotation circuit can be used as the initialization of the next rotation circuit. This would allow us to implement rotations using two measurements per rotation. Depending on the particular sequence of rotations, somewhere between two or three measurements per rotation are required.

\begin{figure}[t]
    \centering
    \subfloat[a) $e^{i\frac{\pi}{4}P} = e^{-i\frac{\pi}{4}Q_{1}} ... e^{-i\frac{\pi}{4}Q_{n-1}} e^{i\frac{\pi}{4}Q_n} e^{i\frac{\pi}{4}Q_{n-1}} ... e^{i\frac{\pi}{4}Q_1}$]{\includegraphics*[width=0.45\textwidth]{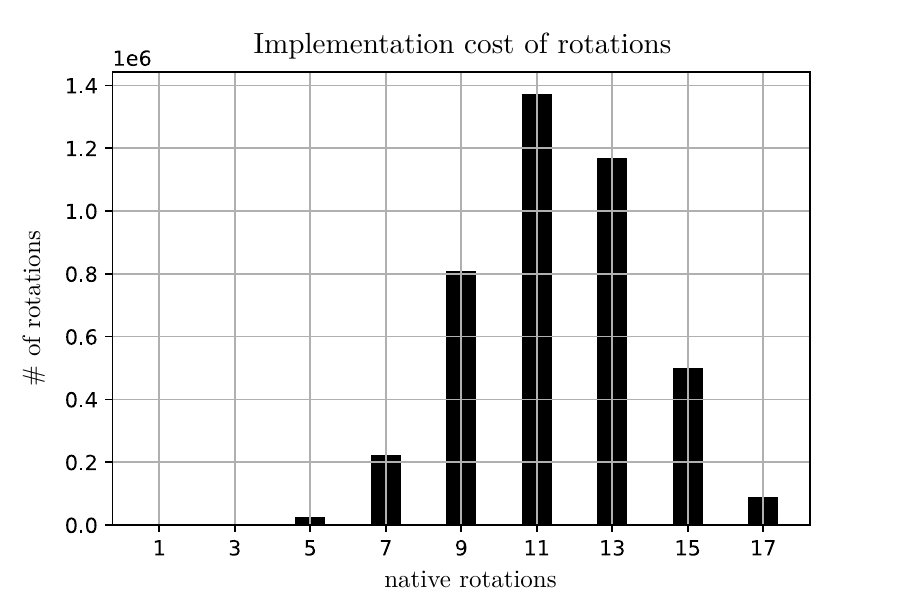}}\qquad
    \subfloat[b) $\mathcal{M}_P = e^{-i\frac{\pi}{4}Q_{1}} ... e^{-i\frac{\pi}{4}Q_{n-1}} \mathcal{M}_{Q_n} e^{i\frac{\pi}{4}Q_{n-1}} ... e^{i\frac{\pi}{4}Q_1}$]{\includegraphics*[width=0.45\textwidth]{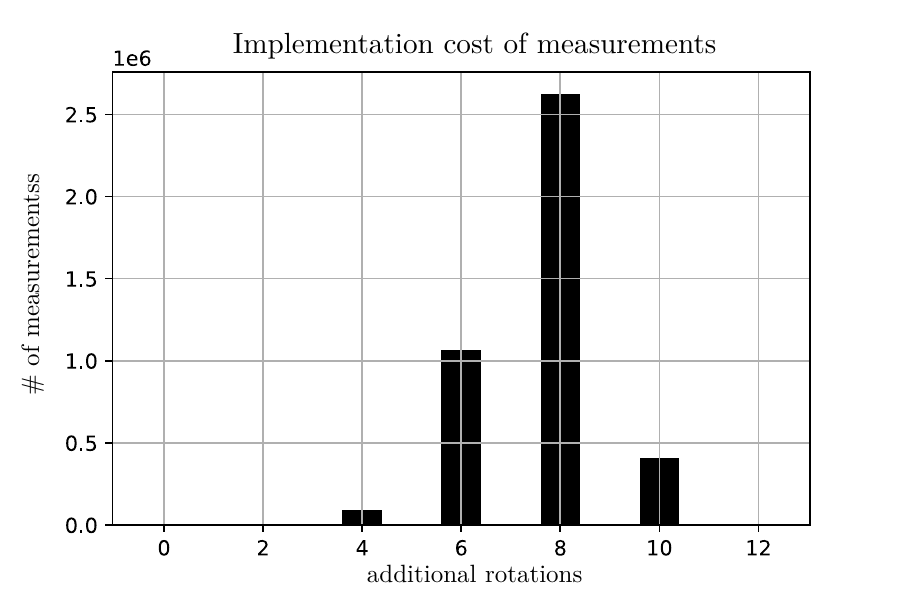}}\\
    
    \caption{\label{fig:implementations} Partially optimized implementation costs of a) all $4^{11}$ Pauli rotations $e^{i\frac{\pi}{4}P}$ and b) all $4^{11}$ Pauli measurements $\mathcal{M}_P$. Cost is measured in the total number of rotation circuits (\ref{eqn:rotationsynth}) required to implement the gate using the form displayed beneath the figure.}
\end{figure}

\section{Future Work}

In this work, we integrated three new ideas into the study of QLDPC surgery: using graph expansion to lower space overhead, gauge-fixing extraneous logical qubits to enhance fault-tolerance, and building bridge systems to connect multiple ancilla systems and codes.
As demonstrated on the gross code, the gauge-fixed surgery scheme is a promising method of performing low-overhead logical computation on QLDPC codes. 
We sample several future directions arising from this work.

Evidently, ancilla systems are particularly affordable when only one layer of additional qubits is necessary. We were able to find logical operators on the gross code where this is the case. Are there any other families of codes where this is possible? Our analytical techniques did not suffice to prove that a monolayer system suffices for the gross code. Is it possible to prove that a single layer is enough? Quantum expander codes seem promising candidates.

For our simulations, we derived explicit syndrome circuits for the $\bar X$ and $\bar Z$ measurement schemes for the gross code. Are these syndrome circuits optimal in terms of the number of fault locations? Can one design optimal or otherwise high-performance syndrome circuits for other measurement schemes, perhaps assuming that the original code $\mc{G}$ also has an efficient syndrome circuit? 

Finally, from Section~\ref{sec:logicalcliffords} it is clear that implementing logical computations will require very many measurements, which in turn required several rounds of error correction. For certain error correction codes, such as perhaps \cite{scruby2024highthreshold}, is there a single-shot error correction scheme for merged codes? 
Lowering the time overhead of surgery operations would have significant practical and theoretical implications.

\paragraph{Acknowledgements}

We are grateful to Ben Brown, John Blue, Michael Beverland, Christopher Pattison and Dominic Williamson for enlightening discussions and valuable feedback.
Z.~He is supported by National Science Foundation Graduate Research Fellowship under Grant No. 2141064. 
This work was done in part while Z.~He was visiting the Simons Institute for the Theory of Computing, supported by NSF QLCI Grant No.~2016245.

\appendix

\section{Appendices}

\subsection{Proofs for the \texorpdfstring{$X$}{X} Ancilla System}\label{apx:Xproofs}

In this appendix, we prove Theorem~\ref{thm:validXancilla} for the $X$ measurement scheme. We begin with an explicit description of the stabilizer check matrices $H_L^X$ and $H_L^Z$, corresponding to the code $\mathcal{G}_X$. Recall that the matrices $J_{V_0} : \mathcal{V} \to V_0$ and $J_{C_0} : \mathcal{C}^Z \to C_0$ are projection isometries onto the subsets $V_0 \subset \mathcal{V}$ and $C_0 \subset \mathcal{C}^Z$ respectively.  These define $F := J_{C_0}^\top H^Z J_{V_0}$ and the full-rank $G$ satisfying $GF = 0$.  We may also consider isometries $J_{V_i} : \mathcal{V} \cup V_1 \cup \cdots V_{i-1} \to V_i$ for all odd $i$ that take elements in $V_0 \subset \mathcal{V}$ to their corresponding element in $V_i$, and similarly for  $J_{C_i} : \mathcal{C}^Z \cup C_1 \cup V_2 \cup \cdots C_{i-1} \to C_i$. These give a convenient way of writing $H_{L}^X$ and $H_{L}^Z$.

Having outlined the construction of $\mathcal{G}_X$, we show that it gives a valid code that features $\bar{X}_M$ as a stabilizer, and has exactly one fewer qubit than $\mathcal{G}$.

\logicalqubit*

\begin{proof}[Proof of Theorem~\ref{thm:validXancilla}.]

First we show that the checks defined by the matrices $H^X_{L}$ and $H_{L}^Z$ commute, which is equivalent to $H_{L}^X (H_{L}^Z)^\top = 0$. We follow their recursive definition. For the base case, since $H^X$ and $H^Z$ commute, we may restrict our attention to the qubits $V_0, C_1$, and must show that the $X$ checks on $V_1 :\begin{bmatrix} I & F^\top\end{bmatrix}$ commutes with the $Z$ checks on $C_0 : \begin{bmatrix} F & I\end{bmatrix}$ and, if $L=1$, $U_1: \begin{bmatrix} 0 & G\end{bmatrix}$.
\begin{align}
X(V_1) \sim Z(C_0)&: \hspace{2mm}  \begin{bmatrix}I & F^\top\end{bmatrix} \cdot \begin{bmatrix}F & I\end{bmatrix}^\top = F^\top + F^\top = 0, \\
X(V_1) \sim Z(U_1)&: \hspace{2mm}  \begin{bmatrix}I & F^\top\end{bmatrix} \cdot \begin{bmatrix}0 & G\end{bmatrix}^\top = F^\top G^\top = 0,
\end{align}
since $GF = 0$ by definition of $G$. For the recursive case, we restrict our attention to qubits $C_{L-2}, V_{L-1}, V_{L}$.
The $X$ checks $V_{L-2} : \begin{bmatrix} F^\top & I & 0\end{bmatrix}$ and $V_{L} : \begin{bmatrix} 0 & I & F^\top \end{bmatrix}$ should commute with the $Z$ checks $C_{L-1} : \begin{bmatrix} I & F & I \end{bmatrix}$. If this is the final layer, then we also need to include $Z$ checks $U_L$. Note that $X$ checks $V_{L-2})$ and $Z$ checks $U_L$ commute since they act on disjoint sets of qubits.
\begin{align}
X(V_{L-2}) \sim Z(C_{L-1}) &: \hspace{2mm}  \begin{bmatrix}F^\top & I & 0\end{bmatrix} \cdot \begin{bmatrix}I & F & I\end{bmatrix}^\top = F^\top + F^\top = 0, \\
X(V_{L}) \sim Z(C_{L-1}) &: \hspace{2mm}  \begin{bmatrix}0 & I & F^\top\end{bmatrix} \cdot \begin{bmatrix}I & F & I\end{bmatrix}^\top = F^\top + F^\top = 0,\\
X(V_{L}) \sim Z(U_{L}) &: \hspace{2mm}  \begin{bmatrix}0 & I & F^\top\end{bmatrix} \cdot \begin{bmatrix}0 & 0 & G\end{bmatrix}^\top = F^\top G^\top = 0.
\end{align}

Second, we demonstrate that $\bar{X}_M$ is a stabilizer in $\mathcal{G}_{X}$: it is simply the product of all the $X$ stabilizers on all the $V_i$ with odd $i$. To see this, we note that for odd $j<L$, 
\begin{align}
    H_X(V_j) = X(V_{i-1})X(V_{i+1}).
\end{align}
The product is not supported on any of the $C_j$ since all the individual contributions from each of the checks in $V_j$ cancel. This is equivalent to the claim that each $Z$ check on $C_0$ is supported on an even number of qubits in $V_0$, which is guaranteed since $\bar{X}_{M}$, whose support is exactly $V_0$, is a logical operator. Similarly, $L$ is odd and $H_X(V_L) = X(V_{L-1})$, i.e.~the product of checks in the final layer acts on all qubits in $V_{L-1}$. Therefore, we see that
\begin{align}
    \prod_{\substack{j=1\\j \text{ odd}}}^LX(V_j) = X(V_0) = \bar{X}_M.
\end{align}

Finally, we show that if $\mathcal{G}$ has $k$ logical qubits, then $\mathcal{G}_X$ has $k-1$ logical qubits. We proceed by induction on the number of layers, and we use $\mathcal{G}^l_X$ to indicate the merged code for measuring $\bar{X}_M$ with $l$ layers, and $H^X_l$, $H^Z_l$ its parity check matrices.  

Make the following observation: the checks $U_L$ can be connected to any $C_i$ for odd $i$, and hence we may as well assume they are always connected to $C_1$. To see this, let $(I,F,I)$ describe the $Z$ checks in $C_{2i}$ with the three components of the vector indicating the support on $C_{2i-1}$, $V_{2i}$, and $C_{2i+1}$, respectively. Then, suppose the checks $U_L$ can be supported on $C_{2i+1}$ (this is the case for $2i+1=L$), so that $(0,0,vG)$ represents some combination of these checks. We can move the checks to be entirely supported on $C_{2i-1}$ by adding a combination of checks from $C_{2i}$, namely
\begin{equation}
(0,0,vG)+vG(I,F,I)=(vG,0,0).
\end{equation}
In other words, the gauge operators $U_L$ have equivalent representatives over all $C_i$ for odd $i$.

We begin the inductive proof starting with $L = 1$. If $\mathcal{G}$ has $n$ physical qubits, then we have $k = n - \text{rank}(H^X) - \text{rank}(H^Z)$. Then $\mathcal{G}_X^1$ has $n + |C_0|$ physical qubits. The rows of $\text{rank}(H^X_1)$ corresponding to $V_1$ form a linearly independent set (owing to their support on $V_0$), and no product of them can be in $H^X$. Therefore $\text{rank}(H^X_1) = \text{rank}(H^X) + |V_1|$.

Now we compute the rank of $H_{1}^Z$. We consider the null space $\mathrm{null}(H_1^Z) = \{c \text{ s.t. } cH_1^Z = 0\}$, and show that it has the same dimension as $\rn(H^Z)$. Suppose $u\subset U_1$ and $v\subset \mc{C}_Z$ are $Z$ checks such that their product is trivial. Then:
\begin{align}
 0 =\begin{bmatrix}u & v\end{bmatrix} H_{1}^Z  = \begin{bmatrix}u & v\end{bmatrix} \begin{bmatrix} 0 & G \\ H^Z & J_{C_1}  \end{bmatrix}  =  u \begin{bmatrix} 0 & G \end{bmatrix}  + v\begin{bmatrix}H^{Z}  & J_{C_1} \end{bmatrix}
\end{align}
implying $v H^Z = 0$, so $|\rn(H_{1}^Z)| \leq |\rn(H^Z)|.$ Conversely, suppose that $v \in \rn(H^Z)$. We claim there exists a vector $u$ such that together $u,v$ are in $\rn(H_{1}^Z)$. By the calculation above, this demands that 
and $u G = v J_{C_1}$, where we can view $v J_{C_1}$ as a $Z$ Pauli product on $C_1$. This Pauli operator is in the stabilizer by construction, and must hence commute with the $X$ stabilizers supported on $V_1$. The corresponding submatrix of $H_{1}^X$ is $\begin{bmatrix}J_{V_1} & F^\top\end{bmatrix}$, so commutativity is equivalent to $F^\top (v J_{C_1})^\top = 0$, so $(vJ_{c_1})F = 0$. Because the rows of $G$ span $\rn(F)$, there must exist a $u$ such that $u G = v J_{C_1}$. We conclude that $|\rn(H_1^Z)| = |\rn(H^Z)|.$ Hence:
\begin{align}
\text{rank}(H_{1}^Z) = (m_z + |U_1|) - \text{dim}(\rn(H^Z)) = \text{rank}(H^Z) + |U_1|.
\end{align}

The size of $U_1$ is the rank of $G$. Observe that the rank of $F$ is $|V|-1$, since we assumed that the only logical operator supported on $V$ is $\bar{X}_M$.
Therefore $|U_1| = \rank(G) = |C_0| - \rank(F) = |C_0| - |V_0| + 1$,
and the number logical qubits in $\mathcal{G}_X^1$ is
\begin{align}
n + |C_0| - \text{rank}(H_{1}^X) - \text{rank}(H_{1}^Z) &= n - |V_0| + |C_0| - \text{rank}(H^X) - (\text{rank}(H^Z) + |U_1|)  \\
&=  k -  |V_0| + |C_0| - |U_1| \\
&=  k -  |V_0| + |C_0| - (|C_0| - |V_0| + 1)  = k-1.
\end{align}

We proceed with the inductive step: assuming that $\mathcal{G}^{L-2}_X$ has $k-1$ logical qubits, we show the same for $\mathcal{G}^{L}_X$. The newly added set of $X$ checks $V_L$ is linearly independent. A product of any subset of $V_L$ will have non-trivial support on $C_L$, and while $H_X(V_L) = X(V_{L-1})$ has no support on $C_L$, by assumption it cannot be a stabilizer in $H^X_{L-2}$. Similarly, the newly added $Z$ checks $C_{L-1}$ is linearly independent, and their products always have support on $C_{L}$. We conclude that $\mathcal{G}_X^L$ has $k-1$ logical qubits. 

In fact, these $k-1$ qubits are the remaining logical qubits from $\mc{G}$. For any $\bar{Z}$ that commutes with $\bar{X}_M$, its support on $V_0$ is even in size, and 
as a consequence of Lemma~\ref{lem:supportlemma}, we can clean its support from $V_0$ by multiplying by $Z$ checks of $\mathcal{G}$. This implies some $\bar{Z}'\equiv_{\mathcal{G}}\bar{Z}$ supported only on $R$, which is then also a logical operator of $\mathcal{G}_X$. For a logical operator $\bar{X}\not\equiv_{\mathcal{G}}\bar{X}_M$ of $\mathcal{G}$ supported only on $V_0\cup R$, notice that $\bar{X}$ is also a logical operator of $\mathcal{G}_X$. It is nontrivial in $\mathcal{G}_X$ because there was some $\bar{Z}$ that anticommuted with $\bar{X}$ in $\mathcal{G}$ and commuted with $\bar{X}_M$, and therefore the construction of $\bar{Z}'$ still anticommutes with $\bar{X}$ in $\mathcal{G}_X$.
\end{proof}

\subsection{Proofs for the \texorpdfstring{$Y$}{Y} Ancilla System}\label{apx:Yproofs}
\Ylogicalqubit*
\begin{proof}
We omit the proof of commutativity, since all checks commute as they do in the $X$ and $Z$ measurement ancilla patches, respectively. Similarly, it is straightforward to see that the operator $\bar{X}_M\bar{Z}_M$ is a product of  stabilizers in the final code, as we can take the product of all checks in the interface and two modules. We prove the number of logical qubits by an induction argument on the rank of the symplectic parity check matrix $\mc{H}$. This matrix describes parity checks as row vectors of length $2n$, where the first $n$ bits specify the locations of Pauli $X$ and the last $n$ specify the locations of Pauli $Z$, which is standard for non-CSS codes. One can write $\mc{H}$ down explicitly using Figure~\ref{fig:Y_ancilla}.

We first note that Lemma~\ref{lem:Y_system_gauges} implies that there are as many linearly independent checks in $U^B$ as there are qubits in $B$. Therefore, the addition of $U^B$ and $B$ does not change the logical qubit count, so next we analyze the $Y$-system without them. This eliminates $|B|$ rows and $|B|$ columns from $\mc{H}$.

To begin the induction, consider the case $L = 1$, where we have $n + |C^X| + |C^Z|$ many physical qubits. We compute the dimension of the row nullspace of $\mc{H}$, $\rn(\mc{H})$.
Suppose there is a set of checks whose product acts trivially on all qubits. 
First, observe that the check corresponding to $q_1$ cannot be in this set. This is because the check on $q_1$ acts as $Y$ on the qubit $q_0$, which means for the product of checks to act as identity on $q_0$, we need a collection of $X$-checks on $\mc{G}$ that multiply to identity on all of $\mc{G}$ except $W_0^Y$, and similarly a collection of $Z$-checks from $\mc{G}$ that multiply to identity on all of $\mc{G}$ except $q_0$. Such collections cannot exist as their product will anti-commute. Let $\mc{H}'$ be the matrix $\mc{H}$ with the row $q_1$ removed, we see that $\rn(\mc{H}) = \rn(\mc{H}')$.

With all the non-CSS check removed, we can now bound $\dim(\rn(\mc{H}'))$ by computing the redundancies in $X$-checks and $Z$-checks separately. Using the same argument as in the proof of Theorem~\ref{thm:validXancilla}, we have
\[
\dim(\rn(\mc{H})) = \dim(\rn(\mc{H}')) = \dim(\rn(H^X)) + \dim(\rn(H^Z)). 
\]
Therefore, 
\begin{align*}
    \rank(\mc{H}) 
    &= m_X+m_Z+|V^X|+|V^Z|-1+|U^X|+|U^Z| - \dim(\rn(H^X)) - \dim(\rn(H^Z)) \\
    &= \rank(H^X) + \rank(H^Z) + |C^X| + |C^Z| + 1,
\end{align*}
which implies the number of logical qubits is $k-1$. 

Inductively, we make the same observation as in the proof of Theorem~\ref{thm:validXancilla}: the newly added set of $X$ checks $V^X_L$ is linearly independent. Any $H_X(v\in V^X_L)$ will have non-trivial support on $C_L$, and while $H_X(V^X_L) = X(V_{L-1})$ has no support on $C_L$, by assumption it cannot be a stabilizer in $\mc{G}^{L-2}_Y$. Similarly, the newly added $Z$ checks in $C^X_{L-1}$ are linearly independent, and their products always have support on $C_{L}$. We conclude that $\mc{G}^L_Y$ has $k-1$ logical qubits, which are the original logical qubits from $\mc{G}$.
\end{proof}

\Yfaultdistance*
\begin{proof}
We will show that $d(\bar{X}_M\ML\MU)\ge d$. The proof for $d(\bar{Z}_M\ML\MU)\ge d$ follows similarly. 
As before, let $\mathcal{G}_X$ and $\mathcal{G}_Z$ denote the merged codes of the $X$- and $Z$-systems, respectively. 
For an operator $\bar{P}$, let $P_X, P_Z$ denote its $X$ and $Z$ components, and let $P_X' = P_X\vert_{\MG_X}, P_Z' = P_Z\vert_{\MG_Z\cup B}$. 
Due to the generators of $\init$, we have
\begin{align}
    d(\bar{X}_M\ML\MU) = \min_{\bar{P}\in \bar{X}_M\ML\MS_Y} \MR_{\init}(\bar{P}) = \min_{\bar{P}\in \bar{X}_M\ML\MS_Y} |P_X'| + |P_Z'|.
\end{align}
Denote by $q_1^X$ and $q_1^Z$ the $X$ and $Z$ parts of the check $q_1=q_1^Xq_1^Z$, we see that $\bar{X}_M\equiv_{\MS_Y} Z(B)q_1^Z$. 
Consider stabilizers in the ancilla system, we observe that all gauge checks and all checks on even layers does not affect $|P_X'|$ or $|P_Z'|$. Therefore, let $\tilde{S} = \langle \MS_O, V^X_{i\text{ odd}}, V^Z_{i\text{ odd}}, q_1\rangle$, we have
\begin{align}
    \min_{\bar{P}\in \bar{X}_M\ML\MS_Y} |P_X'| + |P_Z'| = \MR_{\ML\tilde{S}}(Z(B)q_1^Z).
\end{align}
It suffices for us to show that the right hand side is at least $d$.
Suppose we multiply $Z(B)q_1^Z$ by a collection of checks $v_i\in V_i^X$, $u_j\in V_j^Z$ for $i,j$ odd, and $(q_1)^s$ for $s = 0$ or $1$. Denote this new operator as $P$, namely, 
\begin{align}
    P = Z(B)q_1^Z \prod_{\substack{i=1\\i\text{ odd}}}^{L_X} \MH(v_i) \cdot \prod_{\substack{j=1\\j\text{ odd}}}^{L_Z} \MH(u_j)\cdot (q_1)^s.
\end{align}
To simplify notation, we collect the terms $q_1^Z$ and $(q_1)^s$ in the following way: 
if $s = 0$, we add vertex $q_1$ to $u_1$ (because effectively $P$ has a factor of $q_1^Z$); if $s = 1$, we add vertex $q_1$ to $v_1$ (because effectively $P$ has a factor of $q_1^X$). Then by the Expansion Lemma
\begin{align}
    |P|
    &\ge |B + S^Xv_{1} + S^Zu_{1}| + |v_m + v_{L_X}| + |u_m + u_{L_Z}| + |v_1 + v_m| + |u_1 + u_m| \notag \\
    &+ \min(|v_m|, |\bar{X}| - |v_m|) + \min(|u_m|, |\bar{Z}| - |u_m|) + |v_1| + |u_1|.
\end{align}
$P$ now has $X$-support $v_1\in V_0^X$ and $Z$-support $u_1\in V_0^Z$. 
If we further multiply $P$ by stabilizers in $\MS_O$ or operators in $\ML$, this support could be reduced to $\MR_{\ML\MS_O}(X(v_1)Z(u_1))$. Therefore, we have
\begin{align}
    \MR_{\ML\MS_O}(P) 
    &\ge |B + S^Xv_{1} + S^Zu_{1}| + |v_m + v_{L_X}| + |u_m + u_{L_Z}| + |v_1 + v_m| + |u_1 + u_m| \notag \\
    &+ \min(|v_m|, |\bar{X}| - |v_m|) + \min(|u_m|, |\bar{Z}| - |u_m|) + \MR_{\ML\MS_O}(X(v_1)Z(u_1)).
\end{align}
We do a case analysis based on $|v_m|, |u_m|$.
\begin{enumerate}
    \item If $|v_m| \le |\bar{X}| - |v_m|$ and $|u_m| \le |\bar{Z}| - |u_m|$, we have
    \begin{align}
        \MR_{\ML\MS_O}(P)
        &\ge |B + S^{X}v_{1} + S^{Z}u_{1}| + |v_1 + v_m| + |u_1 + u_m| + |v_m| + |u_m|\\
        &\ge |B + S^{X}v_{1} + S^{Z}u_{1}| + |v_1| + |u_1| \\
        &\ge |B| \ge d.
    \end{align}
    \item If $|v_m| \ge |\bar{X}| - |v_m|$ and $|u_m| \le |\bar{Z}| - |u_m|$, we have
    \begin{align}
        \MR_{\ML\MS_O}(P)
        &\ge |u_1 + u_m| + |u_m| + |\bar{X}| - |v_m| + \MR_{\ML\MS_O}(v_1u_1)\\
        &\ge |u_1| + |\bar{X}| - |v_1| + \MR_{\ML\MS_O}(v_1u_1).
    \end{align}
    Note that $\MR_{\ML\MS_O}(v_1u_1) + |u_1| \ge \MR_{\ML\MS_O}(v_1)$. Therefore $\MR_{\ML\MS_O}(P)\ge |\bar{X}| - |v_1| + \MR_{\ML\MS_O}(v_1)\ge \MR_{\ML\MS_O}(\bar{X}) \ge d$.
    \item The case $|v_m| \le |\bar{X}| - |v_m|$ and $|u_m| \ge |\bar{Z}| - |u_m|$ follows similarly.
    \item If $|v_m| \ge |\bar{X}| - |v_m|$ and $|u_m| \ge |\bar{Z}| - |u_m|$, we have
    \begin{align}
        \mathcal{R}_{\mathcal{S}}(P')
        &\ge |\bar{X}| - |v_m| + |\bar{Z}| - |u_m| + \MR_{\ML\MS_O}(v_1u_1)\\
        &\ge |\bar{X}| - |v_1| + |\bar{Z}| - |u_1| + \MR_{\ML\MS_O}(v_1u_1) \\
        &\ge \MR_{\ML\MS_O}(\bar{X}\cdot \bar{Z}) \ge d. 
    \end{align}
\end{enumerate}
Therefore we have $d(\bar{X}_M\ML\MU) = \MR_{\ML\tilde{S}}(Z(B)q_1^Z) \ge d$, as desired. 
\end{proof}

\subsection{Proofs for the Multi-block \texorpdfstring{$XX$}{XX} Ancilla System}\label{apx:XXproofs}
\XXtwoblocks*
\begin{proof}
Let $\mathcal{S}$ be the stabilizer group of the original codes and each ancilla qubit, including bridge qubits, in $\ket{0}$, and $\mathcal{S}_{XX}$ be the stabilizer group of the joint merged code, Figure~\ref{fig:XX_separate_codes}. 
Also let $\mathcal{L}$ be the group of logical operators of the unmeasured $2k-1$ qubit logical space and $\mathcal{L}^*=\mathcal{L}\setminus\{I\}$. 
The versions of Lemmas~\ref{lem:subsystem_codes_and_fault_distance} and \ref{lem:CSS_subsystem_case} that apply to the joint code imply that the measurement fault distance is $\min(R,d_Z(\bar{Z}_1\mathcal{L}\mathcal{S}))$ and the logical fault distance is $\min(d_Z(\mathcal{L}^*\mathcal{S}),d_X(\mathcal{L}^*\mathcal{S}_{XX}))$. The original code $\mathcal{S}$ has code distance at least $d$ -- if it did not, then the individual systems would not have fault distances at least $d$ -- and so $d_Z(\bar{Z}_1\mathcal{L}\mathcal{S})\ge d$ and $d_Z(\mathcal{L}^*\mathcal{S})\ge d$. To complete the proof, it remains only to show that the $X$-distance of the joint merged code is good, i.e.~$d_X(\mathcal{L}^*\mathcal{S}_{XX})\ge d$.

Suppose we have a logical $X$ operator $P$ from the merged code with stabilizer group $\mathcal{S}_{XX}$. We aim to show $P$ is either trivial or has weight at least $d$.

Let $Q_1$ (resp.~$Q_2$) denote all qubits of the merged code for measuring $\bar{X}_1$ out of $\mathcal{G}_1$ (resp.~measuring $\bar{X}_2$ out of $\mathcal{G}_2$). Then, it is not hard to see $P_1=P|_{Q_1}$ and $P_2=P|_{Q_2}$ are logical operators (possibly trivial) of their respective merged codes. This is because they commute with all the checks of those codes. 

If either $P_1$ or $P_2$ are nontrivial in their respective merged codes, then, by assumption, their weight and thus the weight of $P$ is at least $d$. So assume both $P_1$ and $P_2$ are trivial.

This implies we can find $P'\equiv_{\mathcal{S}_{XX}}P$ such that $P'$ is supported only on the bridge qubits $B$. Since $P'$ must also commute with all checks in $U^B$ and they have the structure of a repetition code on $B$ (Lemma~\ref{lem:XX_system_gauges}) this implies $P'=X(B)$ or $I$. Obviously, the latter case implies $P$ is a stabilizer of $\mathcal{S}_{XX}$, so we proceed with the case $P'=X(B)$. Notice that $|P'|=|B|=\min(|\bar{X}_1|,|\bar{X}_2|)\ge d$. If we can show that the weight of $X(B)$ cannot be reduced by multiplying by stabilizers, then also the weight of $P$ is at least $|B|\ge d$.

Let $S_1$ be a product of $X$ checks from the top half of Figure~\ref{fig:XX_separate_codes}. That is, from the sets of checks $\mathcal{C}^{X(1)}$ and $V_j^{(1)}$ for odd $j$. We may write $S_1=X(b_1\in B)X(e_1)$ where $e_1$ is supported only on qubits in $\mathcal{V}^{(1)}_0$, $V_i^{(1)}$ for even $i$, and $C_j^{(1)}$ for odd $j$. Likewise, define $S_2=X(b_2\in B)X(e_2)$ as a product of checks from the bottom half of Figure~\ref{fig:XX_separate_codes}. 

This means that if $P''=S_1S_2X(B)$ were to have less weight than $X(B)$, then $|e_1|+|e_2|=|e_1+e_2|<|b_1+b_2|\le|b_1|+|b_2|$, where we used the fact that $e_1$ and $e_2$ have disjoint support for the first equality. In particular, it would have to be that either $|b_1|>|e_1|$ or $|b_2|>|e_2|$. Therefore, to show $|P''|\ge|B|$, it suffices to show that neither $S_1X(B)$ nor $S_2X(B)$ has less weight than $|B|$ for all possible choices of $S_1$ and $S_2$.

These are similar arguments so we just provide a proof that $|S_1X(B)|\ge|B|\ge d$. Although $X(B)\equiv_{\mathcal{S}_{XX}}\bar{X}_1$, the fact that the $\bar{X}_1$ system has code distance at least $d$ is insufficient because $\bar{X}_1$ is actually a stabilizer in the merged code that measures it. Instead, let 
\begin{equation}
S_1 = H_X(r\in \mathcal{C}^{X(1)})\prod_{\substack{j=1\\j\text{ odd}}}^LH_X(v_j\in V_j^{(1)}),
\end{equation}
for arbitrary choices of $r$ and $v_j$. Splitting the $X$-check matrix of $\mathcal{G}_1$ into $M : \mathcal{C}^{X(1)}\to_X V_0^{(1)}$ and $\mc{M} : \mathcal{C}^{X(1)}\to_X \mc{V}_0^{(1)}$, we can calculate
\begin{align}\label{eq:1st_set_of_factors}
X(B)S_1 = &X(\vec1+v_LS^{(1)}\in B)\left(\prod_{\substack{i=2\\i\text{ even}}}^{L-1}X(v_{i+1}+v_{i-1}\in V_i^{(1)})\right)X(v_1+r M\in V_0^{(1)})X(r \mc{M}\in \mathcal{V}_0^{(1)})\\\label{eq:2nd_set_of_factors}
&\times \left(\prod_{\substack{j=1\\j\text{ odd}}}^LX(v_jF^{(1)\top}\in C_j^{(1)})\right).
\end{align}
Since $|X(B)S_1|$ is at least the weight of factors in line~\ref{eq:1st_set_of_factors}, we proceed by ignoring the factors in line~\ref{eq:2nd_set_of_factors}.

We note that $|X(\vec1+v_L S^{(1)}\in B)|=|\vec1+v_L S^{(1)}|\ge|B|-|v_L|$. Using this and the triangle inequality repeatedly, we find that $|X(B)S_1|\ge|B|-|v_L|+|v_L+r M|+|r \mc{M}|$. Applying the triangle inequality again, $|X(B)S_1|\ge |B|-|r M|+|r \mc{M}|$. Since $\bar{X}_1$ is minimal weight in its equivalence class, it must be that $|r M|\le|r \mc{M} |$ for all $r$, which means $|X(B)S_1|\ge|B|=|X(B)|$ for all $S_1$.
\end{proof} 
\begin{sidewaysfigure}[p]
    \centering
    \includegraphics[width=0.85\textwidth]{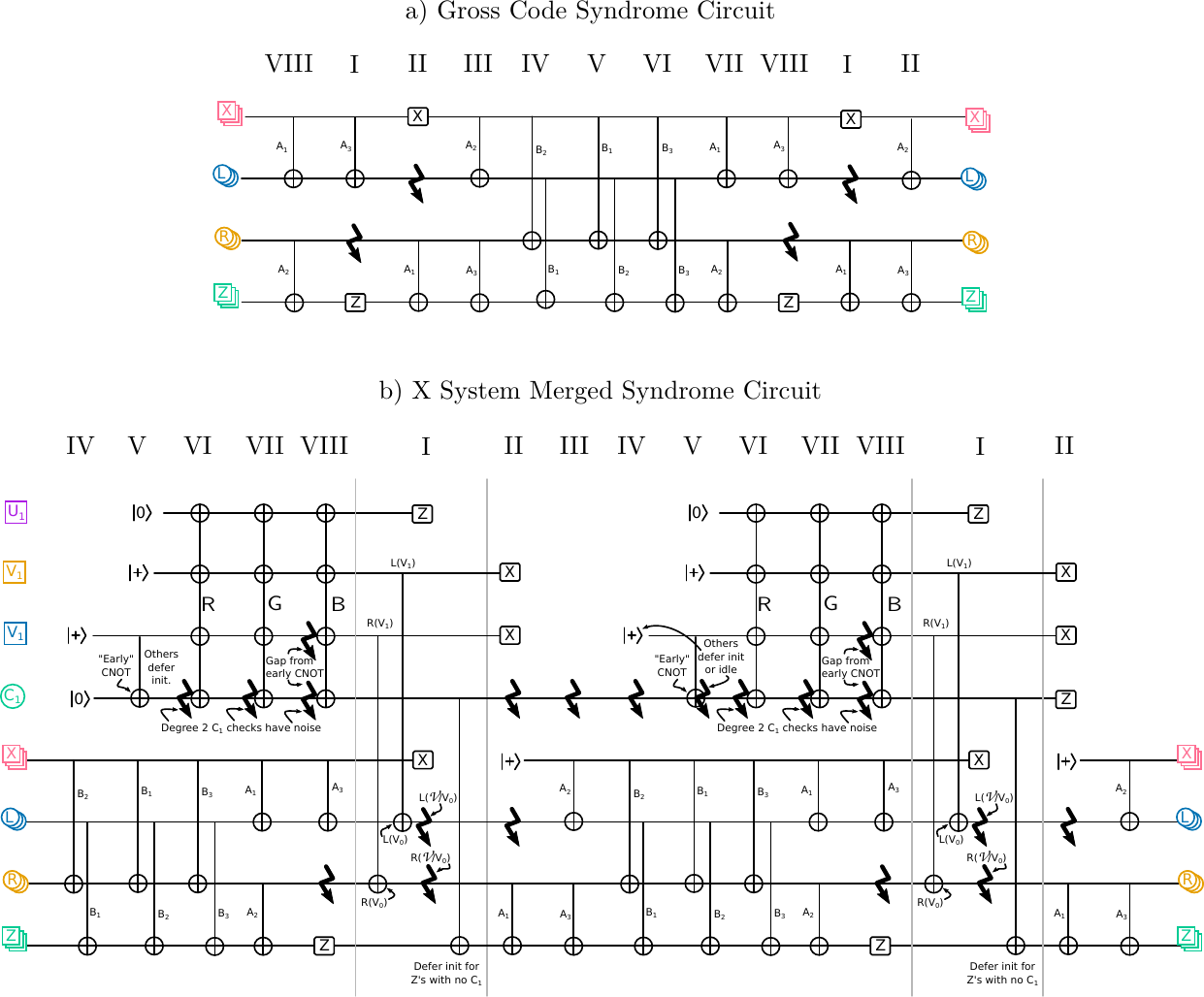}
    \caption{Syndrome circuits used for the numerical experiments. a) The original syndrome circuit for the gross code from \cite{bravyi2024highthreshold}. b) Two cycles of the merged syndrome circuit for an $X$ ancilla system. The gates marked $R$,$G$,$B$ denote single layers of CNOTs determined by an edge-coloring of Figure~\ref{fig:gross_ancilla_spec}, corrected by the procedure in Figure~\ref{fig:coloring_circuits}. Lightning bolts denote idle errors.}
    \label{fig:syndromecircuits}
\end{sidewaysfigure}

\begin{figure}[p]
    \centering
    \includegraphics[width=\textwidth]{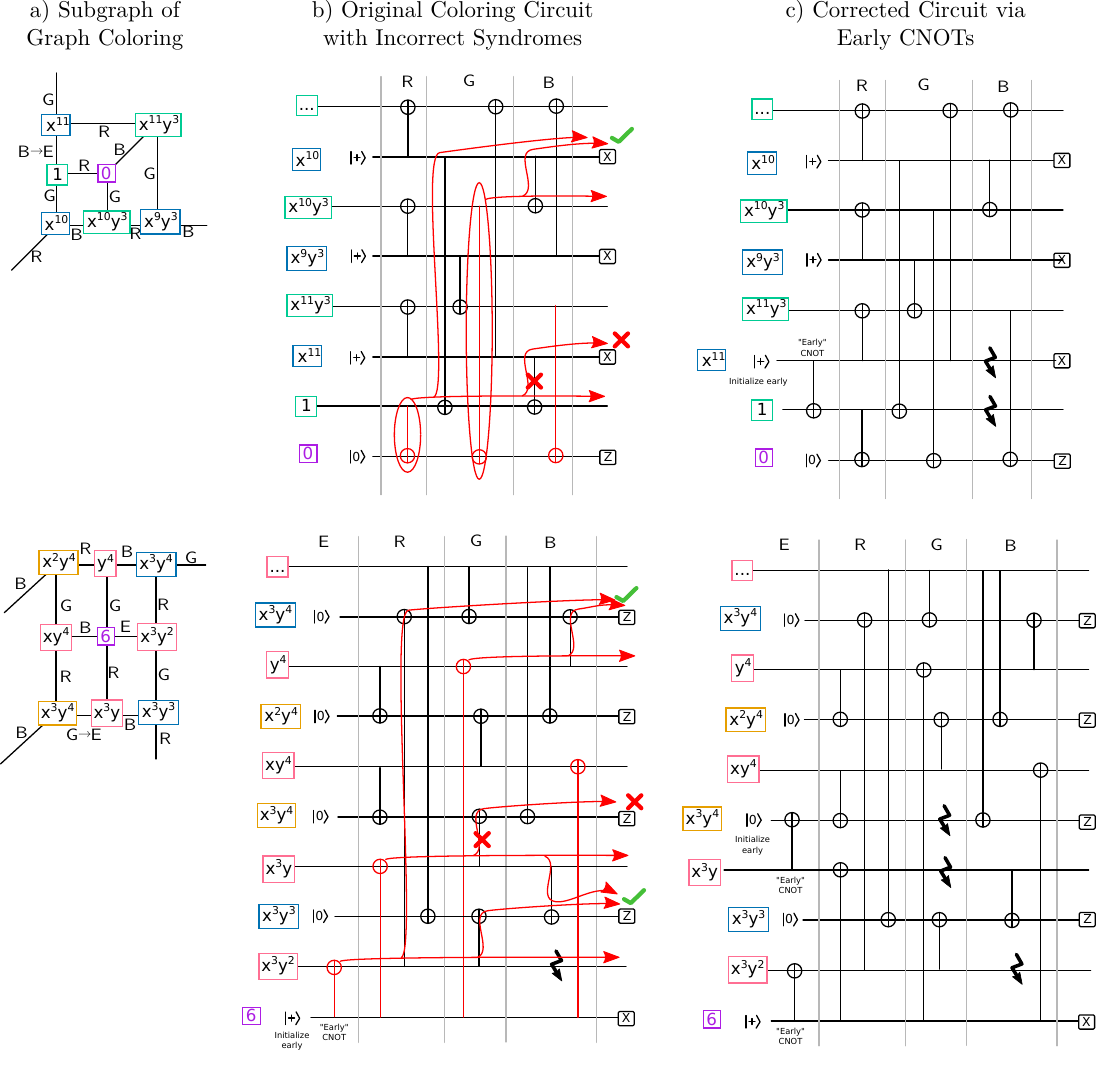}
    \caption{Procedure for correcting the syndrome circuit of the $X$ and $Z$ ancilla systems (top and bottom respectively) from the edge-coloring. a) Subregion of the Tanner graph in Figure~\ref{fig:gross_ancilla_spec} with a mixture of $X$ and $Z$ checks, hence resulting in incorrectly computed syndromes. b) Incorrect circuit with the support of the gauge stabilizer propagated to the end of the circuit. We see that this stabilizer is now supported on some of the \emph{check} qubits, which is wrong. We see that there is a particular CNOT gate that is responsible. c) Corrected circuit where the responsible CNOT is moved to the front. Lightning bolts denote idle errors.}
    \label{fig:coloring_circuits}
\end{figure}

\newpage
\sloppy
\bibliographystyle{alpha}
\bibliography{main}

\end{document}